\documentclass[sigconf, screen]{acmart}

\AtBeginDocument{%
  \providecommand\BibTeX{{%
    Bib\TeX}}}

\copyrightyear{2024}
\acmYear{2024}
\setcopyright{acmlicensed}\acmConference[HSCC '24]{27th ACM International Conference on Hybrid Systems: Computation and Control}{May 14--16, 2024}{Hong Kong, Hong Kong}
\acmBooktitle{27th ACM International Conference on Hybrid Systems: Computation and Control (HSCC '24), May 14--16, 2024, Hong Kong, Hong Kong}
\acmDOI{10.1145/3641513.3650124}
\acmISBN{979-8-4007-0522-9/24/05}

\usepackage{bm}
\usepackage{amsmath,amsfonts}
\usepackage{extarrows}
\usepackage{stmaryrd}
\usepackage{xspace}
\usepackage[font=footnotesize,labelfont=bf]{caption}
\usepackage[inline]{enumitem}
\usepackage[ruled,vlined,linesnumbered]{algorithm2e}
\SetAlFnt{\footnotesize}
\usepackage[title,page]{appendix}
\usepackage{booktabs}
\usepackage{array}
\usepackage{multirow}
\usepackage{array}
\usepackage{textcomp}
\usepackage{tikz}
\usetikzlibrary{arrows,shapes,snakes,automata,backgrounds,petri,positioning} 
\usepackage{stfloats}
\usepackage[title,page]{appendix}
\usepackage{url}
\usepackage{float}
\usepackage{verbatim}
\usepackage{graphicx}
\usepackage{marvosym}
\usepackage{colortbl}
\usepackage{hyperref}
\usepackage{mathrsfs}
\usepackage{multicol}
\usepackage{listings} 
\newcolumntype{P}[1]{>{\raggedright\arraybackslash}p{#1}}
\newcolumntype{Y}[1]{>{\centering\arraybackslash}p{#1}}
\usepackage{subcaption}
\def\BibTeX{{\rm B\kern-.05em{\sc i\kern-.025em b}\kern-.08em
    T\kern-.1667em\lower.7ex\hbox{E}\kern-.125emX}}

\usepackage{orcidlink}

\newcommand{\oomit}[1]{}

\begin{document}
%\begin{sloppypar}
\title{Learning Deterministic Multi-Clock Timed Automata}

\author{Yu Teng\,\orcidlink{0000-0002-9434-1695}}
\affiliation{%
  \institution{School of Software Engineering, Tongji University}
  \city{Shanghai}
  \country{China}
}
\email{tengyu@tongji.edu.cn}

\author{Miaomiao Zhang\,\orcidlink{0000-0001-9179-0893}}
\authornote{The corresponding authors: M. Zhang and J. An, This is the author version.}
\affiliation{%
  \institution{School of Software Engineering, Tongji University}
  \city{Shanghai}
  \country{China}
}
\email{miaomiao@tongji.edu.cn}

\author{Jie An\,\orcidlink{0000-0001-9260-9697}}
\authornotemark[1]
\affiliation{%
  \institution{National Institute of Informatics}
  \city{Tokyo}
  \country{Japan}
}
\email{jiean@nii.ac.jp}

\renewcommand{\shortauthors}{Teng et al.}

\begin{abstract}
We present an algorithm for active learning of deterministic timed automata with multiple clocks. The algorithm is within the querying framework of Angluin's $L^*$ algorithm and follows the idea proposed in existing work on the active learning of deterministic one-clock timed automata. We introduce an equivalence relation over the reset-clocked language of a timed automaton and then transform the learning problem into learning the corresponding reset-clocked language of the target automaton. Since a reset-clocked language includes the clocks reset information which is not observable, we first present the approach of learning from a powerful teacher who can provide reset information by answering reset information queries from the learner. Then we extend the algorithm in a normal teacher situation in which the learner can only ask standard membership query and equivalence query while the learner guesses the reset information. We prove that the learning algorithm terminates and returns a correct deterministic timed automaton. Due to the need of guessing whether the clocks reset at the transitions, the algorithm is of exponential complexity in the size of the target automaton. 
\end{abstract}



% \keywords{Automata learning, Active learning, Deterministic multi-clock timed automata, Timed language, Reset-clocked language}


\settopmatter{printfolios=true}
\maketitle

\section{Introduction}\label{section:Introduction}
Automata learning (a.k.a. model learning) is a classic topic with an extensive history, as a research intersection of formal methods (e.g. automata theory~\cite{HopcroftU79}) and machine learning (e.g. computation learning theory~\cite{KearnsV94}). In general, according to the choices of the date set, it can be divided into two types. 
One is named \emph{passive learning}, in which the learner has no control of the available data set. In~\cite{Gold78,Angluin78}, Gold and Angluin prove that the passive learning problem of finding the smallest finite-state automaton consistent with a given set of accepted words and rejected words is NP-complete.
The other is named \emph{active learning}, in contract, in which a mechanism is designed for the learner to select the data set. In~\cite{Angluin87}, Angluin proposes the famous $L^*$ algorithm and proves that the complexity of the learning problem is polynomial under the active learning mechanism such that the learner is allowed to ask membership queries and equivalence queries for a teacher. In a membership query, the learner can ask whether a given word belongs to the target regular language. In an equivalence query, the learner provides a hypothesis automaton and asks whether it recognizes the target regular language. 
% The learning process will terminate until return the exact model. 
The learning process will continue until returning the exact model.
Following Angulin's approach, many efficient algorithms and tools have been proposed for active learning of different kinds of automata. It has found a wide range of applications in model checking~\cite{PeledVY02,Leucker06,Fiterau-Brostean16}, compositional verification~\cite{AlurMN05}, legacy code analysis~\cite{SchutsHV16}, etc. We refer to the surveys~\cite{SteffenHM11,Vaandrager17} for a comprehensive introduction.

In recent decades, timed systems have received extensive attention. As we know, timing constraints play a key role in the correctness and safety of timed systems. Since finite-state automata are unable to describe the timing behaviors, this yields a fundamental difference and makes learning formal models of timed systems a challenging yet interesting problem.

In general, for timed systems, automata learning requires learning a timed model from either passive observations or active tests of the system, consisting of a collection of time-event sequences. The learned model should describe these timing behaviors correctly. Timed automata~\cite{AlurD94}, extending finite-state automata with clock variables, is one of the most popular formal models of timed systems. However, there are several obstacles to active learning of timed automata. Since the transitions of timed automata contain both timing constraints that test the values of clocks and resets that update the clocks, from the view of constructing a timed automaton, we need to determine (i) the number of clocks, (ii) the reset information, and (iii) the timing constraints, none of which are directly observable from time-event sequences. The existing works consider timed automata with different restrictions. Recently, an active learning algorithm of deterministic timed automata is proposed in~\cite{waga2023active}. However, the number of locations in the learned automaton are highly increased compared to the original automaton, thus lacking interpretations. More introduction can be found in Section~\ref{sec:relatedwork}.

In this paper, following the idea in~\cite{AnCZZZ20} on active learning of deterministic one-clock timed automata, we extend the active learning method for deterministic timed automata (DTA) with multi-clocks under continuous-time semantics. Suppose that the number of clocks in the target DTA $\mathcal{A}$ is known before learning, i.e. we know that $\mathcal{A}$ has $m$ clocks, our learning method can terminate and return an exact DTA $\mathcal{H}$ equipped with $m$ clocks, i.e. the recognized timed languages of $\mathcal{A}$ and $\mathcal{H}$ are equivalent.

The basic idea of our learning method is as follows. We define the reset-clocked languages of timed automata and present an equivalence relation on the reset-clocked languages. Based on the equivalence relation, we show that given a DTA, its reset-clocked language has a finite number of equivalent classes. 
Then, we prove that timed languages for two DTAs are equivalent if the reset-clocked languages are equivalent, which reduces the problem of learning the timed language of a DTA to that of learning a corresponding reset-clocked language. 

Since a reset-clocked language includes the reset information which is not observable,
we first consider a simple situation such that the learner learns from a \emph{powerful teacher} who can answer three kinds of queries. In addition to answering membership queries and equivalence queries, the powerful teacher can provide the reset information on the transitions of target DTA by answering \emph{reset information queries}. By defining the reset information query, we formalize the ``powerful'' teacher situation considered in~\cite{AnCZZZ20}. Following $L^*$ framework, we design the timed observation table to collect all membership queries and provide corresponding table operations. Compared to~\cite{AnCZZZ20}, we need to consider the clock regions due to the multiple clocks instead of clock valuations of a single clock. Moreover, we record the reset information of every clock in the table. When transforming an observation table into a hypothesis DTA, a two-step construction is proposed. We first transform the table into a DFA and then we present a partition function for recovering the timing constraints of each clock on each transition from a list of clock valuations, and thus get a DTA as the hypothesis in the following equivalence query.

After presenting learning DTA from a powerful teacher, we extend the learning algorithm to the situation of learning from a \emph{normal teacher} who can only answer membership query and equivalence query. The main difference is that the learner now needs to guess the reset information instead of making reset information queries. By guessing the reset information in the observation table, every resulted possible table instance can be handled by the method in the powerful teacher situation. However, due to these guesses, the number of table instances is increased exponentially.

For both algorithms, we prove the termination and the correctness respectively. Moreover, we give the complexity analysis in the numbers of queries the learner needs to make. 

In summary, our main contributions are as follows.
\begin{itemize}
\item An equivalence relation over reset-clocked languages. 
\item Two active learning algorithms for DTA with multi-clocks. One is learning from the powerful teacher who can answer reset information queries in addition, and the other is learning from the normal teacher who can only answer membership queries and equivalence queries.
\item A preliminary implementation of the proposed method.
\end{itemize}

%All omitted proofs of the theorems and lemmas in this paper can be found in Appendix~\ref{appendix:proofs}.

% {\color{red}Additionally, the omitted proofs for the theorems and
% lemmas in this paper can be found in Appendix~\ref{appendix:proofs}.}

\paragraph{Structure} In what follows, Section~\ref{sec:preliminaries} recalls the background on timed automata. The equivalence relation over reset-clocked languages is presented in Section~\ref{sec:myhill}. We then present the learning algorithms in the powerful teacher situation and in the normal teacher situation in Section~\ref{sec:powerfulteacher} and Section~\ref{sec:normalteacher}, respectively. Section~\ref{sec:experiment} reports the experimental results.
Finally, Section~\ref{sec:relatedwork} briefly summarizes the related work and Section~\ref{sec:conclusion} concludes the paper.

All omitted proofs of the theorems and
lemmas in this paper can be found in Appendix~\ref{appendix:proofs}.

\section{Preliminaries}\label{sec:preliminaries}
Let $\mathbb{R}_{\geq 0}$ and $\mathbb{N}$ be the set of non-negative reals and natural numbers, respectively. $\mathbb{B}=\{\top,\bot\}$ represents the Boolean set, where $\top$ is true and $\bot$ is false. Let $\Sigma$ be the fixed alphabet.

\subsection{Timed automata} \label{sbsc:timedautomata}
A \emph{timed word} is a finite sequence $\omega=(\sigma_1,t_1)(\sigma_2,t_2)\cdots(\sigma_n,t_n)$ $ \in (\Sigma\times\mathbb{R}_{\geq 0})^*$, where $t_i$ represents the delay time length before taking action $\sigma_i$ for all $1\leq i\leq n$. An equivalent definition is based on timestamps, i.e. $\omega=(\sigma_1,\tau_1)(\sigma_2,\tau_2)\cdots(\sigma_n,\tau_n)$
where $\tau_k=\sum_{i=1}^{k}{t_i}$ for all $1\leq k \leq n$. We also call such timed words as \emph{delay-timed words} and use $\epsilon$ to denote the special empty word. A \emph{timed language} $\mathcal{L}$ is a set of timed words, i.e. $\mathcal{L}\subseteq (\Sigma\times\mathbb{R}_{\geq 0})^*$.

\emph{Timed Automata}~\cite{AlurD94} is a kind of finite automata equipped with a set of real-valued clocks. Let $\mathcal{C}$ be the set of clock variables. A \emph{clock constraint} $\phi$ is a conjunctive formula of atomic constraints of the form $c\sim n$, for $c\in \mathcal{C}$ and $n\in \mathbb{N}$, where ${\sim}\in\{\leq,\textless,\geq,\textgreater,=\}$. We use $\Phi(\mathcal{C})$ to denote the set of clock constraints. A \emph{clock valuation} $\nu: \mathcal{C}\rightarrow\mathbb{R}_{\geq 0}$ is a function assigning a non-negative real value to each clock $c\in \mathcal{C}$. We write $\nu\in\phi$ if the clock valuation $\nu$ \emph{satisfies} the clock constraint $\phi$, i.e. $\phi$ evaluates to true using the values given by $\nu$. For $d\in\mathbb{R}_{\geq 0}$, let $\nu+d$ denote the clock valuation which maps every clock $c\in\mathcal{C}$ to the value $\nu(c)+d$, and for a set $\mathcal{B}\subseteq\mathcal{C}$, let $[\mathcal{B}\rightarrow 0]\nu$ denote the clock valuation which resets all clock variables in $\mathcal{B}$ to $0$ and agrees with $\nu$ for other clocks in $\mathcal{C}\backslash\mathcal{B}$.

\begin{definition}[Timed automata~\cite{AlurD94}]
  A timed automaton is a tuple $\mathcal{A}=(\Sigma, L, l_0, F,$ $ \mathcal{C}, \Delta)$, where 
\begin{itemize}
    \item $\Sigma$ is the alphabet;
    \item $L$ is a finite set of locations;
    \item $l_0$ is the initial location;
    \item $F\subseteq L$ is a set of accepting locations;
    \item $\mathcal{C}$ is the set of clocks;
    \item $\Delta \subseteq L\times\Sigma\times\Phi(\mathcal{C})\times 2^\mathcal{C}\times L$ is a finite set of transitions.
\end{itemize}
\end{definition}

A transition $\delta = (l, \sigma, \phi, \mathcal{B}, l')$ in $\Delta$ represents a jump from the source location $l$ to the target location $l'$ by performing an action $\sigma \in \Sigma$ if the constraint $\phi \in \Phi(\mathcal{C})$ is satisfied by current clock valuation. The set $\mathcal{B}\subseteq\mathcal{C}$ gives the clocks to be reset with the transition. 
% We call such transition constraints as guards.
We call \emph{guards} such transition constraints.

A \emph{state} $s$ of $\mathcal{A}$ is a pair $(l,\nu)$, where $l \in L$ is a location and $\nu$ is a clock valuation. A \emph{run} $\rho$ of $\mathcal{A}$ over a timed word $\omega=(\sigma_1,t_1)(\sigma_2,t_2)\cdots(\sigma_n,t_n)$ is a sequence $\rho = (l_{0}, \nu_{0}) \xrightarrow{t_{1}, \sigma_{1}} (l_{1}, \nu_{1})$ $\xrightarrow{t_{2}, \sigma_{2}} \cdots \xrightarrow{t_{n}, \sigma_{n}} (l_{n}, \nu_{n})$, satisfying the requirements: (1) $l_0$ is the initial location and $\nu_{0}(c)=0$ for each clock $c\in\mathcal{C}$; (2) for all $1\leq i\leq n$, there is a transition $(l_{i-1},\sigma_i,\phi_i,\mathcal{B}_i,l_i)$ such that $(\nu_{i-1}+t_i)\in\phi_i$ and $\nu_{i}=[\mathcal{B}_i\rightarrow 0](\nu_{i-1}+t_i)$. The run $\rho$ is an \emph{accepting} run if $l_{n} \in F$. 
For a run $\rho$ over a timed word $\omega$, we define the \emph{\textit{trace}} of $\rho$ as the corresponding timed word $\textit{trace}(\rho)=\omega=(\sigma_1,t_1)(\sigma_2,t_2)\cdots(\sigma_n,t_n)\in (\Sigma\times\mathbb{R}_{\geq 0})^*$ if $\rho\neq(l_{0}, \nu_{0})$ or the empty word $\textit{trace}(\rho)=\epsilon$ if $\rho=(l_{0}, \nu_{0})$. 
Given a timed automaton $\mathcal{A}$, we define its (recognized) timed language $\mathcal{L}(\mathcal{A})=\{\textit{trace}(\rho) \, \vert \, \rho \text{ is an accepting run}$ $\text{of } \mathcal{A}\}$. We say two timed automata $\mathcal{A}_1$ and $\mathcal{A}_2$ are \emph{equivalent} if and only if $\mathcal{L}(\mathcal{A}_1)=\mathcal{L}(\mathcal{A}_2)$.

Given a timed word $\omega=(\sigma_1,t_1)(\sigma_2,t_2)\cdots(\sigma_n,t_n)$ and a timed automaton $\mathcal{A}$, if there exists a run $\rho$, then by recording the reset information along the run, we get a \emph{reset-delay-timed word} corresponding to $\omega$, denoted by $\omega_{r}=\textit{trace}_{r}(\rho)=(\sigma_{1}, t_{1}, \mathbf{b}_{1})(\sigma_{2}, t_{2},$ $ \mathbf{b}_{2}) \cdots (\sigma_{n}, t_{n}, \mathbf{b}_{n})$, where $\mathbf{b}_{i}\in\mathbb{B}^{|\mathcal{C}|}$ is a $|\mathcal{C}|$-tuple recording whether the corresponding transition resets the clocks or not when taking the timed action $(\sigma_i,t_i)$. We have $\mathbf{b}_{i,j}=\top$ if the $j$-th clock $c_j$ is in $\mathcal{B}_i$, otherwise $\mathbf{b}_{i,j}=\bot$, for all $c_j\in\mathcal{C}$ and $1\leq j \leq |\mathcal{C}|$. For example, if there are $4$ clocks $c_1,c_2,c_3,c_4$, i.e., $|\mathcal{C}|=4$, and there is a transition $(l_{i-1},\sigma_i,\phi_i,\{c_2,c_3\},l_i)$ taking the timed action $(\sigma_i,t_i)$, then we have $\mathbf{b}_i=(\bot,\top,\top,\bot)$. Similarly, the \emph{(recognized) reset-timed language} $\mathcal{L}_{r}(\mathcal{A})=\{\textit{trace}_{r}(\rho) \, \vert \,  \rho$ is an accepting run of $\mathcal{A}\}$. 

Additionally, based on the run $\rho$, we introduce the \emph{clocked word} $\gamma=(\sigma_1,\mathbf{v}_1)$ $(\sigma_2,\mathbf{v}_2)\cdots(\sigma_n,\mathbf{v}_n)$, where $\mathbf{v}_i\in\mathbb{R}_{\geq 0}^{|\mathcal{C}|}$ records the clock value for each clock when taking action $\sigma_i$, i.e. $\mathbf{v}_{i,j}=\nu_{i-1}(c_j)+t_i$ for all $c_j\in\mathcal{C}$ and $1\leq j \leq |\mathcal{C}|$. After recording the reset information along the run $\rho$, we get the \emph{reset-clocked word} $\gamma_r=(\sigma_1,\mathbf{v}_1,\mathbf{b}_1)(\sigma_2,\mathbf{v}_2,\mathbf{b}_2)\cdots(\sigma_n,\mathbf{v}_n,$ $\mathbf{b}_n)$. 

For a reset-clocked word $\gamma_r$, we define two functions: $vw(\gamma_r)$ returns the clocked word $\gamma=(\sigma_1,\mathbf{v}_1)$ $(\sigma_2,\mathbf{v}_2)\cdots(\sigma_n,$ $\mathbf{v}_n)$, and $resets(\gamma_r)=\mathbf{b}_1, \mathbf{b}_2,\cdots, \mathbf{b}_n$ extract the reset information of $\gamma_r$, which is a sequence from $\mathbf{b}_1$ to $\mathbf{b}_n$.

We denote by $\preceq$ the \emph{lexicographic order} on the product of $|\mathcal{C}|$ clocks. 
Thus, $\mathbf{v}_1\preceq\mathbf{v}_2$ means that there exists $1\leq j \leq |\mathcal{C}|$ such that $\mathbf{v}_{1,j}<\mathbf{v}_{2,j}$ and $\mathbf{v}_{1,i}=\mathbf{v}_{2,i}$ for all $1\leq i<j$, or specially $\mathbf{v}_1=\mathbf{v}_2$.

Given a system modelled by a timed automaton, a (delay) timed word can be regarded as a system behavior from the view of the global clock outside of the system. However, it is difficult to reflect the changing of the internal logical clock variables $\mathcal{C}$ if the system is treated as a black box. It is the main challenge for learning the transition constraints which all correspond to the clock valuations rather than the global clock. While clocked words contain clock valuations, they can be regarded as system behaviors from the view of the internal logical clock variables and are more suitable for learning transition constraints.

Moreover, if we know the reset information $\mathbf{b}_i$ along one run $\rho$ over the timed word $\omega=(\sigma_1,t_1)(\sigma_2,t_2)\cdots (\sigma_n,t_n)$, \footnote{Note that a timed word $\omega$ may lead to several runs in a timed automaton, while for each run $\rho$, the reset information is determined. Thus one run corresponds to only one clocked word of $\omega$.}
%\footnote{Note that a timed word $\omega$ may lead to several runs in a timed automaton, while for each run $\rho$, the reset information is determined. Thus one run corresponds to only one clocked word of $\omega$.} 
we can transform $\omega$ to a clocked word $\gamma=(\sigma_{1}, \mathbf{v}_1)(\sigma_{2}, \mathbf{v}_2)\cdots (\sigma_{n}, \mathbf{v}_n)$ by taking
\begin{equation}
\label{eq:transform}
    \mathbf{v}_{i,j} = \begin{cases}
	t_{i} , & \text{if}\ \  i=1\ \ \text{or}\ \ \mathbf{b}_{i-1,j}=\top \text{ for all } 2\leq i \leq n; \\
	\mathbf{v}_{i-1,j}+t_{i}, & \text{otherwise}.
	\end{cases}
\end{equation}
where $1\leq j \leq |\mathcal{C}|$. Here we define two functions $\Gamma_{\rho}$ and $\pi_{\rho}$ such that $\Gamma_{\rho}$ maps a (reset-)delay-timed word to the (reset-)clocked counterpart and $\pi_{\rho}$ maps a clocked word $\gamma$ to its reset-clocked counterpart $\gamma_r$ according to a run $\rho$ respectively. When the run $\rho$ is determined, we omit the subscript `$\rho$' and use $\Gamma$ and $\pi$ directly. Given a timed automaton $\mathcal{A}$, the \emph{recognized clocked language} of $\mathcal{A}$ is given as $\mathscr{L}(\mathcal{A})=\{\Gamma_{\rho}(\textit{trace}(\rho))\, \vert \, \rho$ is an accepting run of $\mathcal{A}\}$ and the \emph{recognized reset-clocked language} of $\mathcal{A}$ defined as $\mathscr{L}_{r}(\mathcal{A})=\{\Gamma_{\rho}(\textit{trace}_{r}(\rho))\, \vert \, \rho$ is an accepting run of $\mathcal{A}\}$. 

Given a timed automaton $\mathcal{A}$ with clocks $\mathcal{C}$, a clocked word $\gamma\in(\Sigma\times\mathbb{R}_{\geq 0}^{|\mathcal{C}|})^*$ is a \emph{valid} clocked word of $\mathcal{A}$ if there exists at least one run $\rho$ corresponding to $\gamma$, i.e. $\gamma=\Gamma_{\rho}(\textit{trace}(\rho))$. Similarly, a reset-clocked word $\gamma_r\in(\Sigma\times\mathbb{R}_{\geq 0}^{|\mathcal{C}|}\times\mathbb{B}^{|\mathcal{C}|})^*$ is a \emph{valid} reset-clocked word of $\mathcal{A}$ if there exists a run $\rho$ corresponding to $\gamma_r$, i.e. $\gamma_r=\Gamma_{\rho} (\textit{trace}_{r}(\rho))=\pi_{\rho}(\Gamma_{\rho}(\textit{trace}(\rho)))$. Otherwise, it is \emph{invalid} w.r.t $\mathcal{A}$. Specially, we say a reset-clocked word is \emph{doomed} if there is no timed automaton having a run corresponding to it. That means it cannot be transformed into a timed word according to Equation~\eqref{eq:transform}.

\begin{definition}[Deterministic timed automata]
A timed automaton $\mathcal{A}$ is a deterministic timed automaton (DTA) if and only if there is at most one run $\rho$ for any timed word $\omega$.
\end{definition}

It means that, for all $l \in L$ and $\sigma\in\Sigma$, for every pair of transitions of the form $(l,\sigma,\phi_{1},-,-)$ and 
$(l,\sigma,\phi_{2},-,-)$ in $\Delta$, the clock constraints $\phi_{1}$ and $\phi_{2}$ are mutually exclusive. 

A \emph{complete deterministic timed automaton} (CTA) is that for any given timed word $\omega$, there is exactly one run $\rho$. 
Given a DTA, it can be transformed into a CTA accepting the same timed language in three steps: (1) augment $L$ with a ``sink” location which is not an accepting location; (2) for every $l\in L$ and $\sigma\in\Sigma$, if there is no outgoing transition from $l$ labelled with $\sigma$, introduce a transition resetting all clocks from $l$ to ``sink" with label $\sigma$ and guards $c\geq 0$ for each clock $c\in \mathcal{C}$; (3) otherwise, let $\textit{Compl}_{l,\sigma}$ be the subset of $\mathbb{R}_{\geq0}^{|\mathcal{C}|}$ not covered by the guards of transitions from $l$ with label $\sigma$. Write $\textit{Compl}_{l,\sigma}$ as a union of guards $I_{1}, \cdots, I_{k}$ in a minimal way, then introduce a transition resetting all clocks from $l$ to ``sink" with label $\sigma$ and guard $I_{j}$ for each $1\leq j\leq k$. Fig.~\ref{fig:timedautomata} depicts the transformation of DTA $\mathcal{A}$ with $\Sigma=\{a,b\}$ and $\mathcal{C}=\{c_1,c_2\}$ on the left into the CTA on the right by the above operation. 
The detailed description can be found in Appendix~\ref{appendix:transitionCTA}.

\begin{figure}
\begin{center}
\begin{minipage}[b]{0.49\textwidth}
\begin{center}
\resizebox{!}{0.35\textwidth}{
\begin{tikzpicture}[scale=0.65, ->, >=stealth', shorten >=1pt, auto, node distance=2cm, semithick, every node/.style={scale=0.8}]
% \centering
        \node[initial, state]  (0) {$l_0$};
        \node[accepting, state](1) at (0,-3) {$l_1$};
        \path  (0) edge node[right, align=center] {$a,$ \\ $c_{1}> 1\wedge c_{2} > 1,$\\ $\{c_2\}$} (1)
        
        (1) edge[in= 240, out=120, color=black] node[left, black, align=center] {$a, 0 \leq c_{1} < 3\wedge c_{2} > 1,$ \\ $\{c_1,c_2\}$} (0)
        
        (1) edge[loop below] node[left, align=center] {$b,$ \\ $ c_{1}\geq 0\wedge 0\leq c_{2}< 1,$\\ $\{c_2\}$} (1);
\end{tikzpicture}

\begin{tikzpicture}[scale=0.65, ->, >=stealth', shorten >=1pt, auto, node distance=2cm, semithick, every node/.style={scale=.8}]
% \centering
        \node[initial, state]  (0) {$l_0$};
        \node[accepting, state](1) at (0,-3) {$l_1$};
         \node[state](2) at (5.5,-1.5)  {$l_2$};
         
        \path  (0) edge node[right, align=center] {$a,$ \\$c_{1}> 1\wedge c_{2}> 1,$\\ $\{c_2\}$} (1)
        
        (1) edge[in= 240, out=120, color=black] node[left, black, align=center] {$a, 0 \leq c_{1} < 3\wedge c_{2} > 1,$ \\ $\{c_1,c_2\}$} (0)
        
        (1) edge[loop below] node[left, align=center] {$b,$ \\$c_{1}\geq 0 \wedge 0\leq c_{2}< 1,$\\ $\{c_2\}$} (1)
         
        (0) edge[in= 110, out=50, color=red] node[above, sloped, black] {$a, c_{1}> 1\wedge 0 \leq c_{2}\leq 1, \{c_1,c_2\}$} (2)
         
         (0) edge[in= 140, out=10, color=red] node[above, sloped, black] {$a, 0\leq c_{1}\leq 1\wedge c_{2}\geq 0,\{c_1,c_2\}$} (2)
         
         (0) edge[color=blue] node[above, sloped, black] {$b, c_{1}\geq 0\wedge c_{2}\geq 0, \{c_1,c_2\}$} (2)
         
         (2) edge[in= 40, out=10,loop, color=blue] node[right, black, align=center, pos=0.7] {$a, c_{1}\geq 0\wedge c_{2}\geq 0,$ \\$\{c_1,c_2\}$} (2)
     
     (2) edge[in= 340, out=310,loop, color=blue] node[right, black, align=center, pos=0.7] {$b, c_{1}\geq 0\wedge c_{2}\geq 0,$ \\$\{c_1,c_2\}$} (2)

          (1) edge[in= 240, out=340, color=red] node[above, sloped, black] {$b,  c_{1}\geq 0\wedge c_{2}\geq 1, \{c_1,c_2\}$} (2)
          
          (1) edge[in= 220, out=0, color=red] node[above, sloped, black] {$a, c_{1} \geq 3 \wedge c_{2}\geq 0, \{c_1,c_2\}$} (2)
          
          (1) edge[color=red] node[above, sloped, black] {$a, 0\leq c_{1} < 3\wedge c_{2}\leq 1, \{c_1,c_2\}$} (2);
\end{tikzpicture}
}
\end{center}
\end{minipage}
\end{center}
\caption{A DTA $\mathcal{A}$ (left) and its corresponding CTA (right). The initial location $l_0$ is indicated by `start’ and the accepting location $l_1$ is doubly cycled.}
\label{fig:timedautomata}
\end{figure}

Since every DTA can be transformed into a CTA, in order to present our learning method clearly, we assume that we work with CTA in this paper. 

According to the definition, a timed automaton has infinite states since the clock values are in $\mathbb{R}_{\geq 0}$. In the theory of timed automata, \emph{clock region} is an abstraction used to form a finite partition of the state space. It will also be used in proving the termination of our learning algorithm in this paper. Here we review the definitions of regions and symbolic states of timed automata over regions introduced in~\cite{AlurD94}.

\begin{definition}[Region equivalence~\cite{AlurD94}]
Given a timed automaton $\mathcal{A}$, let $\kappa$ be a function, called a clock ceiling, mapping each clock $c\in\mathcal{C}$ to the largest integer appearing in the guards over $c$. For a real number $d$, let $\textit{frac}(d)$ denote the fractional part of $d$, and $\lfloor d\rfloor$ denote its integer part. Two clock valuations $\nu$, $\nu'$ are region-equivalent, denoted by $\nu \sim \nu'$, 
iff all the following conditions hold:
\begin{itemize}
    \item for all $c\in\mathcal{C}$, either $\lfloor\nu(c)\rfloor=\lfloor\nu'(c)\rfloor$ or both $\nu(c)\textgreater\kappa(c)$ and $\nu'(c)\textgreater\kappa(c)$.
    \item for all $c\in\mathcal{C}$, if $\nu(c)\leq\kappa(c)$ then $\textit{frac}(\nu(c))=0$ iff $\textit{frac}(\nu'(c))=0$ and
    \item for all $c_{i},c_{j}\in\mathcal{C}$, if $\nu(c_{i})\leq\kappa(c_{i})$ and $\nu(c_{j})\leq\kappa(c_{j})$ then $\textit{frac}(\nu(c_{i})) \leq \textit{frac}(\nu(c_{j}))$ iff $\textit{frac}(\nu'(c_{i}))\leq
    \textit{frac}(\nu'(c_{j}))$.
\end{itemize}
A clock region for $\mathcal{A}$ is an equivalence class induced by the equivalence relation $\sim$. 
For a clock valuation $\nu$, we denote by $\llbracket \nu \rrbracket$ the region containing it.
\end{definition}

Given a timed automaton $\mathcal{A}$, there is only a finite number of regions. 
We denote the set of regions by $\mathcal{R}$. 
According to the definition of timed automata, given a transition constraint $\phi$ of $\mathcal{A}$ if $\nu\sim \nu'$ then $\nu\in\phi$ if and only if $\nu'\in\phi$. The number of clock regions $|\mathcal{R}|$ is bounded by $|\mathcal{C}|!\cdot 2^{|\mathcal{C}|}\cdot\prod_{c\in\mathcal{C}}(2\kappa(c)+2)$~\cite{AlurD94}.

\begin{definition}[Symbolic state~\cite{AlurD94}]\label{def:symbolicstate}
    A symbolic state of timed automata $\mathcal{A}=(L, l_0, F, \mathcal{C},\Sigma, $ $\Delta)$ is a pair $(l,\llbracket\nu\rrbracket)$, where $l\in L$ and $\llbracket\nu\rrbracket$ is a clock region.
\end{definition}

Therefore, the number of the symbolic states of $\mathcal{A}$ is bounded by $|L|\cdot |\mathcal{C}|!\cdot 2^{|\mathcal{C}|} \cdot\prod_{c\in\mathcal{C}}(2\kappa(c)+2)$.

\subsection{Myhill-Nerode Theorem and $L^*$ Algorithm}\label{sbsc:lstar}
In this section, we review the famous active learning framework for deterministic finite automata (DFA), named $L^*$ algorithm~\cite{Angluin87}. Before that, we recall the classic Myhill-Nerode Theorem on which it is based.

\newtheorem{thm}{Theorem}

\begin{definition}[Right congruence relation $\sim_{\mathcal{L}}$]\label{def:Rightcongruence}
Given a language $\mathcal{L}$ over $\Sigma$, for $u,u'\in\Sigma^*$, we say $u$ and $u'$ are equivalent under $\mathcal{L}$, denoted by $u\sim_{\mathcal{L}}u'$, iff $\forall v\in\Sigma^*, uv\in \mathcal{L} \Longleftrightarrow u'v\in \mathcal{L}$.
\end{definition}

\begin{theorem}[Myhill-Nerode Theorem]
A language $\mathcal{L}\subseteq \Sigma^*$ is regular if and only if $\sim_{\mathcal{L}}$ has a finite number of equivalence classes, and furthermore, this number is equal to the number of states in the minimal DFA.
\end{theorem}

$L^*$ algorithm gives a way to find the finite number of equivalence classes and build the minimal DFA recognizing $\mathcal{L}$. It can be described as a learning process between a learner and a teacher\iffalse in Fig.~\ref{fig:learnerteacher}\fi. 

The learner wants to learn an unknown regular language $\mathcal{L}$ from the teacher who knows the language $\mathcal{L}$ and holds two oracles to answer two kinds of queries from the learner. One kind of query is named \emph{membership query}. In such a query, the learner asks whether a word $u\in\Sigma^*$ is in the language $\mathcal{L}$. The teacher can answer ``yes ($+$)" or ``no" ($-$). The other is named \emph{equivalence query}. The learner asks whether the language $\mathcal{L}(H)$ of a hypothesis DFA $H$ is equal to $\mathcal{L}$. The teacher answers ``yes" if $H$ is correct, otherwise answers ``no" with a counterexample $w\in\Sigma^*$ to show a difference between $\mathcal{L}(H)$ and $\mathcal{L}$.

The learner maintains an \emph{observation table} $O=(U,V,f)$ to record the membership queries and query results, where $U\subseteq \Sigma^*$ is a prefix-closed set of words that identify different states in the current hypothesis DFA $H$, and $V\subseteq \Sigma^*$ is a suffix-closed set of words used to distinguish the states and $f$ maps $(U\cup U\cdot\Sigma)\cdot V$ to $\{+,-\}$. Thus we can find that such design is directed by the right congruence relation. For the rows in the table, a function $\textit{row}: U\cup U\Sigma\rightarrow (V\rightarrow \{+,-\})$ returns a list of the query results given by $\textit{row}(u)(v)=f(uv)$. The learner always wants to keep the observation table \emph{closed} and \emph{consistent} in the learning process. 
\begin{itemize}
    \item An observation table is closed if for each $u\in U$ and $\sigma\in\Sigma$, there always exists $u'\in U$ such that $\textit{row}(u')=\textit{row}(u\sigma)$.
    \item An observation table is consistent if for all elements $u,u'\in U$ such that $\textit{row}(u)=\textit{row}(u')$, then $\textit{row}(u\sigma)=\textit{row}(u'\sigma)$ for all $\sigma\in\Sigma$.
\end{itemize}

At the beginning, $U$ and $V$ only contain the empty word $\epsilon$, respectively. The learner makes membership queries by asking whether the words of the set $(U\cup U\cdot\Sigma)\cdot V$ are in $\mathcal{L}$ and collects the words and the query results in the table. If the current table is not closed, then the learner finds $u'\in U\Sigma$ such that $\textit{row}(u')\neq \textit{row}(u)$ for all $u\in U$ and moves such $u'$ from $U\Sigma$ to $U$. After that, the learner makes membership queries for every $u'\sigma v$ where $\sigma\in\Sigma$ and $v\in V$. If the table is not consistent, the learner finds $u,u'\in U$, $\sigma\in\Sigma$ and $v\in V$ such that $\textit{row}(u')=\textit{row}(u)$ and $f(u'\sigma v)\neq f(u\sigma v)$ and moves such word $\sigma v$ to $V$. After that, the learner makes membership for every $u\sigma v$ where $u\in U\cup U\Sigma$. Once the observation table $O$ is closed and consistent, the learner can construct a hypothesis DFA $H=(\Sigma_{H}, L_{H}, l_{H}^0, F_{H}, \delta_{H})$ where the alphabet $\Sigma_{H}=\Sigma$, the set of states $L_{H}=\{l_{\textit{row}(u)}\mid u\in U\}$, the initial state $ l_{H}^0=l_{\textit{row}(\epsilon)}$, the set of accepting states $F_{H}=\{l_{\textit{row}(u)} \mid u\in U \wedge f(u\cdot\epsilon)=+\}$ and the transition function $\delta_{H}(l_{\textit{row}(u)},\sigma)=l_{\textit{row}(u\sigma)}$. Subsequently, the learner makes an equivalence query. If the answer is no, the teacher gives a counterexample $w$. The learner adds every prefix of $w$ to $U$ and the corresponding membership queries and table operations are performed until the table is closed and consistent. Then a new hypothesis can be constructed from the table for an equivalence query. The learning loop terminates until the answer to an equivalence query is true and the current hypothesis $H$ is a correct DFA such that $\mathcal{L}(H)=\mathcal{L}$.

\section{Equivalence relation on reset-clocked words}\label{sec:myhill}
 
In this section, we aim to build an equivalence relation on reset-clocked words introduced in Section~\ref{sbsc:timedautomata} and prove that it has a finite number of equivalence classes. We first introduce the \emph{region word}, which is helpful in presenting the equivalence relation.

\begin{definition}[Region word]
Given a timed automaton $\mathcal{A}$, a region word $\xi$ is a finite sequence of pairs $(\sigma,\llbracket \nu \rrbracket)$, where $\sigma\in \Sigma$ is an action of $\mathcal{A}$ and $\llbracket \nu \rrbracket$ represents a region, i.e. $\llbracket \nu \rrbracket\in\mathcal{R}$.
\end{definition}

For each clocked word $\gamma$ of $\mathcal{A}$, there must be a unique region word $\xi$ corresponding to $\gamma$. 
We denote it by $\xi=\llbracket \gamma \rrbracket$.
For example, if $\gamma=(a,\{0,0\})(a,\{1.5,1.5\})$, 
then $\xi=\llbracket\gamma\rrbracket=(a,c_{1}=0\wedge c_{2}=0)(a,1< c_{1}< 2\wedge 1< c_{2}< 2\wedge c_{1}=c_{2})$.

\begin{lemma}\label{lemma:logictoresetlogic}
    Given a CTA $\mathcal{A}$, for all valid clocked words $\gamma$ and $\gamma'$ such that $ \llbracket\gamma\rrbracket=\llbracket\gamma'\rrbracket$, they have the same transition sequence of $\mathcal{A}$ and reach the same symbolic state. 
    Thus, the following two conclusions hold:
  \begin{itemize}
      \item $\gamma\in \mathscr{L}(\mathcal{A})$ iff $\gamma'\in \mathscr{L}(\mathcal{A})$;
      \item $resets(\gamma_r)=resets(\gamma_r')$, where $\gamma_r$ and $\gamma_r'$ are the reset-clocked words corresponding to $\gamma$ and $\gamma'$, respectively.
  \end{itemize}
\end{lemma}

The first conclusion shows that given a region word $\xi$, all valid clocked words $\gamma$ with $\xi=\llbracket \gamma \rrbracket$ have the same accepting result over $\mathcal{A}$. The second conclusion tells that such clocked words witness the same reset information over the transitions of $\mathcal{A}$. 

\begin{definition}[Valid successor]\label{def:validsuccessor}
    Given a CTA $\mathcal{A}$, a reset-clocked word $\gamma_{r}$ and a region word $\xi$, we say a reset-clocked word $\gamma_r'$ is a valid successor of $\gamma_{r}$ corresponding to $\xi$ if it satisfies two conditions:
\begin{itemize}
    \item $\llbracket vw(\gamma_r') \rrbracket=\xi$;
    \item $\gamma_{r}\cdot\gamma_r'$ is a valid reset-clocked word of $\mathcal{A}$, i.e. there exists a run $\rho$ in $\mathcal{A}$ such that $\gamma_{r}\cdot\gamma_{r}'=\Gamma(\textit{trace}_r(\rho))$.
\end{itemize}
We denote the set of \emph{valid successors} of $\gamma_r$ corresponding to $\xi$ by $\textit{vs}_{\mathcal{A}}(\gamma_{r},\xi)$.
\end{definition}

\begin{example}
Consider the CTA $\mathcal{A}$ in Fig.~\ref{fig:timedautomata}. Given a valid reset-clocked word $\gamma_{r}=(a, \{1.1,$ $1.1\}, \{\bot,\top\})$ and a region word $\xi=(b,1<c_{1}< 2\wedge c_{2}=0)$, $\gamma_r'=(b, \{1.1,0\}, \{\bot,\top\})$ is a valid successor of $\gamma_r$ corresponding to $\xi$ since $\llbracket vw(\gamma_r')\rrbracket=\llbracket (b, \{1.1,0\})\rrbracket = \xi$ and $\gamma_r\cdot\gamma_r'$ corresponds to the run $\rho=(l_0,\{0,0\})\xrightarrow{a,1.1} (l_1,\{1.1,0\})$ $\xrightarrow{b,0}(l_1,\{1.1,0\})$. 
We denote it by $\gamma_r'\in \textit{vs}_{\mathcal{A}}(\gamma_r,\xi)$.
\end{example}

\begin{lemma}\label{lemma:alltssame}
    Given a valid reset-clocked word $\gamma_{r}$ of a CTA $\mathcal{A}$ and a region word $\xi$, for all $\gamma_r',\gamma_r''\in \textit{vs}_{\mathcal{A}}(\gamma_r,\xi)$, $\gamma_{r}\gamma_r'$ and $\gamma_{r}\gamma_r''$ witness the same transition sequence of $\mathcal{A}$ and reach the same symbolic state in the end.
\end{lemma}

\begin{lemma}\label{lemma:samesymbolicstate}
Given two valid reset-clocked words $\gamma_{r1}$ and $\gamma_{r2}$ of $\mathcal{A}$ and a region word $\xi$, if $\gamma_{r1}$ and $\gamma_{r2}$ reach the same symbolic state of $\mathcal{A}$, then for all $\gamma_{r1}'\in\textit{vs}_{\mathcal{A}}(\gamma_{r1},\xi)$ and $\gamma_{r2}'\in\textit{vs}_{\mathcal{A}}(\gamma_{r2},\xi)$, 
the following two conditions hold:
\begin{itemize}
    \item $resets(\gamma_{r1}')=resets(\gamma_{r2}')$;
    \item $\gamma_{r1}\gamma_{r1}'$ and $\gamma_{r2}\gamma_{r2}'$ reach some same symbolic state.
\end{itemize}
\end{lemma}

Given a DTA $\mathcal{A}$, let $\Sigma$ be the finite alphabet and $\mathcal{R}$ the finite set of regions of $\mathcal{A}$. We let $\bm{\Sigma}=\Sigma\times\mathbb{R}_{\geq 0}^{|\mathcal{C}|}$ be the infinite set of clocked actions, $\bm{\Sigma_{G}}=\Sigma\times\mathcal{R}$ be the finite set of region actions, and $\bm{\Sigma_{r}}=\Sigma\times\mathbb{R}_{\geq 0}^{|\mathcal{C}|}\times\mathbb{B}^{|\mathcal{C}|}$ be the infinite set of reset-clocked actions. 
%{\color{red}For a reset-clocked action $\bm{\sigma_{r}}=(\sigma, \mathbf{v},\mathbf{b})\in \bm{\Sigma_{r}}$, we also use $vw(\bm{\sigma_{r}})$ to represent the clocked action $\bm{\sigma}=(\sigma,\mathbf{v})$, and let $resets(\bm{\sigma_{r}})=\mathbf{b}$.} 
An equivalence relation over the reset-clocked language $\mathscr{L}_{r}(\mathcal{A})$ is defined as follows.

\begin{definition}[Equivalence relation]\label{def:equivalencerelation}
    Let $\mathscr{L}_{r}(\mathcal{A})$ be the reset-clocked language of a DTA $\mathcal{A}$,
two reset-clocked words $ \gamma_{r1}, \gamma_{r2} \in  \bm{\Sigma_{r}}^*$ are indistinguishable by $\mathscr{L}_{r}(\mathcal{A})$, denoted by $\gamma_{r1}\sim_{\mathscr{L}_{r}(\mathcal{A})}\gamma_{r2}$,
if for all $\gamma_{r1}'\in \textit{vs}_{\mathcal{A}}(\gamma_{r1},\xi)$ and $\gamma_{r2}' \in \textit{vs}_{\mathcal{A}}(\gamma_{r2},\xi)$, 
% it holds that $\gamma_{r1}\gamma_{r1}'\in\mathscr{L}_r(\mathcal{A})$ iff $\gamma_{r2}\gamma_{r2}'\in\mathscr{L}_r(\mathcal{A})$ for all $\xi\in \bm{\Sigma_{G}^*}$.
the following conditions holds for all $\xi\in \bm{\Sigma_{G}}^*$:
\begin{itemize}
\item $\gamma_{r1}\gamma_{r1}'\in\mathscr{L}_r(\mathcal{A})$ iff $\gamma_{r2}\gamma_{r2}'\in\mathscr{L}_r(\mathcal{A})$;
\item $resets(\gamma_{r1}')=resets(\gamma_{r2}')$.
\end{itemize}
\end{definition}

\begin{lemma}\label{lemma:finitevalid}
    If two valid reset-clocked words $\gamma_{r1}, \gamma_{r2}$ of $\mathcal{A}$ reach the same symbolic state of $\mathcal{A}$, then $\gamma_{r1}\sim_{\mathscr{L}_{r}(\mathcal{A})}\gamma_{r2}$.
\end{lemma}

\begin{lemma}\label{lemma:finiteinvalid}
    All invalid reset-clocked words belong to the same equivalence class of $\sim_{\mathscr{L}_{r}(\mathcal{A})}$.
\end{lemma}

\begin{theorem}\label{thm:equivalenceclass}
    Let $\mathscr{L}_{r}(\mathcal{A})$ be the reset-clocked language of a DTA $\mathcal{A}$, and then $\sim_{\mathscr{L}_{r}(\mathcal{A})}$ has a finite number of equivalence classes.
\end{theorem}

\section{Learning DTA from a powerful teacher}\label{sec:powerfulteacher}

In this section, we present the algorithm for learning a DTA $\mathcal{A}$ from a powerful teacher. The basic idea is that we transform the learning problem for the timed language of an underlying DTA $\mathcal{A}$ into the learning problem for the reset-clocked language of~$\mathcal{A}$. We first introduce the basic idea and then describe the settings of the powerful teacher. After presenting the details of the learning algorithm, we show its correctness, termination and complexity.

\begin{theorem}\label{theorem:basicidea}
  Given two DTAs $\mathcal{A}_1$ and $\mathcal{A}_2$, if $\mathscr{L}_{r}(\mathcal{A}_1)=\mathscr{L}_{r}(\mathcal{A}_2)$, then $\mathcal{L}(\mathcal{A}_1)=\mathcal{L}(\mathcal{A}_2)$.
\end{theorem}

According to Theorem~\ref{theorem:basicidea}, to construct a DTA that recognizes a target timed language $\mathcal{L} = \mathcal{L}(\mathcal{A})$, it suffices to learn a hypothesis DTA $\mathcal{H}$ such that $\mathscr{L}_{r}(\mathcal{A})=\mathscr{L}_{r}(\mathcal{H})$. Instead of checking directly if $\mathscr{L}_{r}(\mathcal{A})=\mathscr{L}_{r}(\mathcal{H})$, the contraposition of the theorem guarantees that we can still perform equivalence queries over their timed languages: if $\mathcal{L}(\mathcal{A})=\mathcal{L}(\mathcal{H})$, then $\mathcal{H}$ recognizes the target language already; otherwise, a counterexample making $\mathcal{L}(\mathcal{A})\neq\mathcal{L}(\mathcal{H})$ also yields an evidence for $\mathscr{L}_{r}(\mathcal{A})\neq\mathscr{L}_{r}(\mathcal{H})$.

There are two roles in the learning algorithm: a learner and a powerful teacher. The teacher holds the underlying target DTA $\mathcal{A}$ and three oracles to answer three kinds of queries: \emph{reset information query}, membership query and equivalence query from the learner. For a reset information query, the learner asks the reset information along the run of a valid clocked word $\gamma$. The teacher returns the reset information of the clocks, and the learner gets a reset-clocked word. For a membership query, the teacher receives a reset-clocked word $\gamma_r$ from the learner and answers whether $\gamma_{r}\in\mathscr{L}_{r}(\mathcal{A})$. For an equivalence query, the learner submits a hypothesis DTA $\mathcal{H}$ and the teacher answers whether $\mathcal{L}(\mathcal{A})=\mathcal{L}(\mathcal{H})$. If the answer is no, the powerful teacher returns a reset-delay-timed word $\omega_r$ as a counterexample in addition. 
We denote the three oracles for reset information query, membership query and equivalence query by $\textsf{RQ}$, $\textsf{MQ}$ and $\textsf{EQ}$, respectively.

Suppose that the DTA $\mathcal{A}$ is a complete DTA and the learner knows the number of the clocks $|\mathcal{C}|$ in advance. 
% {\color{red}Furthermore, to avoid the assumption that for all $c_j\in\mathcal{C}$ and $1\leq j\leq\lvert\mathcal{C}\rvert$, its corresponding $\kappa(c_j)$ in $\mathcal{A}$ needs to be known, we take an integer $k_{max}$ greater than or equal to the maximum  of $\kappa(c_j)$ for all $c_j\in\mathcal{C}$ and $1\leq j\leq\lvert\mathcal{C}\rvert$ and use $k_{max}$ to replace all  $\kappa(c_j)$.} 
We show more details as follows.

\subsection{Timed observation table, reset information query, and membership query}\label{sbsc:membership}

Similar to the classic $L^*$ algorithm, we design an observation table to collect the query results.

\begin{definition}[Timed observation table for DTA]
A timed observation table for DTA is a tuple $\mathbf{T}=(\Sigma, \bm{\Sigma_{G}}, \bm{\Sigma_{r}}, \bm{S}, \bm{R}, \bm{E}, f, g)$, where $\bm{S},\bm{R}\subset\bm{\Sigma_r}^*$ are finite sets of reset-clocked words, named prefix set and boundary, respectively, and $\bm{E}\subset\bm{\Sigma_G}^*$ is a finite set of region words, named suffix set. Specifically,
\begin{itemize}
\item$\bm{S}$ and $\bm{R}$ are disjoint, i.e., $\bm{S}\cup\bm{R}=\bm{S}\uplus\bm{R}$, and $\bm{S}\cup\bm{R}$ is prefix-closed.
\item The empty word is, by default, both a prefix and a suffix, i.e., $\epsilon\in \bm{E}$ and $\epsilon\in \bm{S}$;
\item $f:\bm{S}\cup\bm{R}\times\bm{E}\rightarrow \{+,-\}$ is the classification function such that for a reset-clocked word $\gamma_r\in\bm{S}\cup\bm{R}$ and a region word $e\in\bm{E}$, $f(\gamma_r,e)=-$ if $\textit{vs}_{\mathcal{A}}(\gamma_r, e)=\emptyset$ or $\mathsf{MQ}(\gamma_r e_r)=-$ i.e. $\gamma_r e_r \notin \mathscr{L}_{r}(\mathcal{A})$ where $e_r \in \textit{vs}_{\mathcal{A}}(\gamma_r,e)$, and $f(\gamma_r,e)=+$ if $\mathsf{MQ}(\gamma_r e_r)=+$, i.e. $\gamma_r e_r \in \mathscr{L}_{r}(\mathcal{A})$.
\item $g:\bm{S}\cup\bm{R}\times\bm{E}\rightarrow \{\bot,\top\}^{|e\in\bm{E}|\times|\mathcal{C}|}$ is a labelling function such that for a reset-clocked word $\gamma_r\in \bm{S}\cup\bm{R}$ and a region word $e\in\bm{E}\backslash \{\epsilon\}$, $g(\gamma_r,e)=resets(e_r)$ where $e_r\in \textit{vs}_{\mathcal{A}}(\gamma_r,e)$ if $\textit{vs}_{\mathcal{A}}(\gamma_r, e)\neq\emptyset$, otherwise $g(\gamma_r,e)=\{\top\}^{|e\in\bm{E}|\times|\mathcal{C}|}$.

\end{itemize}
\end{definition}

Consider the classification function $f$. It records the results of the membership queries i.e., whether $\gamma_r e_r$ is in $\mathscr{L}_{r}(\mathcal{A})$. Depending on Lemma~\ref{lemma:alltssame} and Lemma~\ref{lemma:finitevalid}, $e_r$ can represent all other valid successors $e_r'$ according to the region word $e$ such that $\gamma_r e_r\sim_{\mathscr{L}_{r}(\mathcal{A})}\gamma_r e_r'$. 
The function $g$ records the reset information in the valid successor $e_r$ and assigns all resets to $\top$ if there is no valid successor $e_r$.
However, according to Definition~\ref{def:validsuccessor}, the valid successor $e_r$ depends on a known DTA $\mathcal{A}$. We claim that the learner can compute the valid successors using reset information queries instead of knowing DTA $\mathcal{A}$.

\begin{algorithm}[!t]
%\scriptsize
%   \SetKwData{Passive}{passive}\SetKwData{Read}{read}\SetKwData{Active}{active}\SetKwData{Write}{write}
  \caption{FindValidSuccessor}
  \label{alg:findvalidsuccessor}
  \SetKwInOut{Input}{input}\SetKwInOut{Output}{output}

  \Input{A reset-clocked word $\gamma_{r}=(\sigma_1,\mathbf{v}_1,\mathbf{b}_1)(\sigma_2,\mathbf{v}_2,\mathbf{b}_2)\cdots(\sigma_n,\mathbf{v}_n,\mathbf{b}_n)\in\bm{S}\cup\bm{R}$; \\a region word $e=(\sigma_1',\llbracket \nu_1 \rrbracket)(\sigma_2',\llbracket \nu_2 \rrbracket)\cdots(\sigma_k',\llbracket \nu_k \rrbracket)\in\bm{E}$.}
  \Output{A valid successor $e_r\in\textit{vs}_{\mathcal{A}}(\gamma_r, e)$.}
  
  $e_r\gets \epsilon$ \;
  \If{$\gamma_{r}$ is doomed}{\Return $\bot$} \label{line:doomed}
  \Else{
  \If{$\gamma_r=\epsilon$}{$\mathbf{v}_n\gets 0; \mathbf{b}_n\gets \top$ \tcp*{Set the initial clock valuation.}}
  \For{$i \gets 1$ \KwTo $k$}{
  \For{$j\gets 1$ \KwTo $|\mathcal{C}|$}{
  \If{$\mathbf{b}_{n+i-1,j}=\top$}{$\mathbf{v}_{n+i,j}\gets 0$;}
  \Else {$\mathbf{v}_{n+i,j}\gets \mathbf{v}_{n+i-1,j}$;}
  }
  $\textit{flag}, d \gets \textsf{Solve}(\llbracket\mathbf{v}_{n+i}+d\rrbracket=\llbracket \nu_i \rrbracket)$ \label{line:SMT} 
  
  \If{$\textit{flag}=\bot$}{\Return $\bot$;}
  \Else{$\mathbf{v}_{n+i}\gets\mathbf{v}_{n+i}+d$;}
  % $\mathbf{b}_{n+i}\gets \textbf{RQ}(\gamma_r\cdot (\sigma_i',\mathbf{v}_{n+i}))$\;
    $\mathbf{b}_{n+i}\gets \textbf{RQ}(vw(\gamma_r)\cdot (\sigma_i',\mathbf{v}_{n+i}))$\;
    \label{line:RQ}
 $e_r\gets e_r\cdot (\sigma_{i}',\mathbf{v}_{n+i},\mathbf{b}_{n+i})$\;
 $\gamma_r\gets \gamma_r\cdot (\sigma_i',\mathbf{v}_{n+i},\mathbf{b}_{n+i})$\;
  }
\Return $e_r$\;
  }
\end{algorithm}

Algorithm~\ref{alg:findvalidsuccessor} shows the computation for a valid successor $e_r$. Given a reset-clocked word $\gamma_r\in \bm{S}\cup \bm{R}$, we can use Equation~\ref{eq:transform} to determine whether $\gamma_r$ is doomed. If it is doomed, then we stop the computation as we know that there is no valid successor (Line~\ref{line:doomed}). Otherwise, we compute $e_r$ by considering the region actions in the region word $e\in\bm{E}$ one by one. For the first region action $(\sigma_1',\llbracket \nu_1 \rrbracket)$, depending on the last reset information in $\gamma_r$, we can compute the clock values $\mathbf{v}_{n+1}$. 
% Then we use an SMT solver to find a delay time $d\in\mathbb{R}_{\geq 0}$ such that $\mathbf{v}_{n+1} + d\in \llbracket \nu_1 \rrbracket$ (Line~\ref{line:SMT}). 
Then we find a delay time $d\in\mathbb{R}_{\geq 0}$ such that $\mathbf{v}_{n+1} + d\in \llbracket \nu_1 \rrbracket$ by the function $Solve()$ and use $flag$ to represent whether there is such $d$ or not. 
If there is no such $d$, it means that there is no clock valuation in $\llbracket \nu_1 \rrbracket$ that can form a valid successor. Otherwise, we update $\mathbf{v}_{n+1}$ and get a clocked word $(\sigma_1',\mathbf{v}_{n+1})$. Then we make a reset information query for $vw(\gamma_r)\cdot(\sigma_1',\mathbf{v}_{n+1})$ to get the last reset information $\mathbf{b}_{n+1}$ (Line~\ref{line:RQ}). After that, we update $e_r$ and $\gamma_r$. The computation for other region actions in $e$ is similar.

\begin{example}
Fig.~\ref{fig:observationtable} shows an example of a timed observation table. Suppose that the CTA $\mathcal{A}$ in Fig.~\ref{fig:timedautomata} is the underlying target automaton. Consider $\gamma_{r}=(a,\{0,0\},\{\top,\top\})$ and $e=(a,1< c_{1}< 2\wedge 1< c_{2}< 2\wedge c_{1}=c_{2})$, $n=1$ for the length of $\gamma_r$ and there is only one region action in $e$. 
According to Algorithm~\ref{alg:findvalidsuccessor}, we get $e_{r}=(a,\{1.05,1.05\},\{\top,\top\})\in \textit{vs}_{\mathcal{A}}(\gamma_r, e)$. 
The detailed process of finding $e_{r}$ can be found in Appendix~\ref{appendix:findvalidsucc}.
\end{example}

\begin{figure}[!t]
\begin{center}
\resizebox{0.49\textwidth}{!}{
\begin{tabular}{ P{5.5cm}|Y{0.5cm} Y{5cm}  }
 \hline
  &$\epsilon$&$(a,1< c_{1}< 2\wedge  1<c_{2}<2\wedge c_1=c_2)$\\
 \hline
 $\epsilon$&$-$&$+,\bot,\top$\\
  $(a,\{1.05,1.05\},\{\bot,\top\})$   & $+$    &$-,\top,\top$\\
   $(a,\{0,0\},\{\top,\top\})$   & $-$    &$-,\top,\top$\\
 \hline
  $(b,\{0,0\},\{\top,\top\})$   & $-$    &$-,\top,\top$\\
  $(a,\{1.05,1.05\},\{\bot,\top\})(a,\{0,0\},\{\top,\top\})$   & $-$    &$-,\top,\top$\\
  $(a,\{1.05,1.05\},\{\bot,\top\})(b,\{0,0\},\{\top,\top\})$   & $-$    &$-,\top,\top$\\
  $(a,\{1.05,1.05\},\{\bot,\top\})(b,\{1.05,0\},\{\bot,\top\})$   & $+$    &$-,\top,\top$\\
  $(a,\{1.05,1.05\},\{\bot,\top\})(b,\{2.05,1\},\{\top,\top\})$   & $-$    &$-,\top,\top$\\
  $(a,\{0,0\},\{\top,\top\})(a,\{1.05,1.05\},\{\top,\top\})$   & $-$    &$-,\top,\top$\\
  $(a,\{0,0\},\{\top,\top\})(a,\{0,0\},\{\top,\top\})$   & $-$    &$-,\top,\top$\\
  $(a,\{0,0\},\{\top,\top\})(b,\{0,0\},\{\top,\top\})$   & $-$    &$-,\top,\top$\\
 \hline
\end{tabular}
}
\end{center}
\vspace{-0.2cm}
\caption{An example of the observation table.}
\label{fig:observationtable}
\end{figure}

Similar to the $L^*$ algorithm, for the observation table, we also define a function $\textit{row}$ mapping each
$\gamma_{r}\in \bm{S}\cup\bm{R}$ to a vector indexed by every $e\in \bm{E}$. Each element of the vector is defined as $\{f(\gamma_{r}, e), g(\gamma_r, e)\}$.
% {\color{red}$\{f(\gamma_{r}, e),g(\gamma_{r}, e)\}$ if $e\neq\epsilon$ and otherwise $\{f(\gamma_{r}, \epsilon)\}$}. 
For example, in Fig.~\ref{fig:observationtable}, consider $\gamma_r=(a,\{1.05,$ $1.05\},\{\bot,\top\})$, we have
% $\textit{row}(\gamma_r)={\color{red}(\{+\},\{-,\top,\top\})}$.
$\textit{row}(\gamma_r)=\{\{+\},\{-,\top,\top\}\}$.

Before constructing a hypothesis $\mathcal{H}$ based on the timed observation table $\mathbf{T}$, the
learner has to ensure that $\mathbf{T}$ satisfies the following conditions:
\begin{itemize}
	\item \emph{Reduced:} $\forall \emph{s}, \emph{s}'\in \bm{S}\colon \emph{s}\neq \emph{s}' \text{ implies }\mathit{row}(\emph{s})\neq \mathit{row}(\emph{s}')$;
	
	\item \emph{Closed:} $\forall r \in \bm{R}, \exists s \in \bm{S}\colon \mathit{row}(s) = \mathit{row}(r)$;

\item \emph{Consistent:} $\forall \gamma_{r}, \gamma_{r}' \in \bm{S}\cup\bm{R}$, $\mathit{row}(\gamma_{r})=\mathit{row}(\gamma_{r}')$ implies that: for all $\bm{\sigma_{r}},\bm{\sigma_{r}}'\in \bm{\Sigma_{r}}$ such that $\llbracket vw(\bm{\sigma_{r}})\rrbracket=\llbracket vw(\bm{\sigma_{r}}')\rrbracket$,  if $\gamma_{r}\bm{\sigma_{r}}, \gamma_{r}'\bm{\sigma_{r}}'\in \bm{S}\cup\bm{R}$,  the following conditions are satisfied:
	\begin{itemize}
	    \item 	$\mathit{row}(\gamma_{r}\bm{\sigma_{r}})=\mathit{row}(\gamma_{r}'\bm{\sigma_{r}}')$;
	    \item $resets(\bm{\sigma_{r}})=resets(\bm{\sigma_{r}}')$.
	    
	\end{itemize}

	\item \emph{Evidence-closed:} $\forall \emph{s}\in\bm{S}$ and $\forall \emph{e}\in\bm{E}$, if $\textit{vs}_{\mathcal{A}}(\emph{s}, e)$ is not empty, there is $e_{r}\in \textit{vs}_{\mathcal{A}}(\emph{s}, e)$ such that $\emph{s}e_{r}\in \bm{S}\cup\bm{R}$;
	
	%\item \emph{Prefix-closed:} $\bm{S}\cup\bm{R}$ is prefix-closed.
\end{itemize}

A timed observation table $\mathbf{T}$ is prepared if it satisfies the above four conditions. To get the table prepared, the learner can perform the following operations.

\begin{figure}[!t]
\begin{center}
\resizebox{0.35\textwidth}{!}{

\begin{tabular}{ P{6.5cm}|Y{0.8cm}  }
 \hline
  $\textbf{T1}$&$\epsilon$\\
 \hline
 $\epsilon$&$-$\\
  $(a,\{1.05,1.05\},\{\bot,\top\})$   & $+$ \\
  \hline
   $(a,\{0,0\},\{\top,\top\})$   & $-$    \\
  $(b,\{0,0\},\{\top,\top\})$   & $-$    \\
  $(a,\{1.05,1.05\},\{\bot,\top\})(a,\{0,0\},\{\top,\top\})$   & $-$    \\
  $(a,\{1.05,1.05\},\{\bot,\top\})(b,\{0,0\},\{\top,\top\})$   & $-$    \\
  $(a,\{1.05,1.05\},\{\bot,\top\})(b,\{1.05,0\},\{\bot,\top\})$   & $+$    \\
  $(a,\{1.05,1.05\},\{\bot,\top\})(b,\{2.05,1\},\{\top,\top\})$   & $-$    \\
  $(a,\{0,0\},\{\top,\top\})(a,\{1.05,1.05\},\{\top,\top\})$   & $-$    \\
 \hline
\end{tabular}
}%\\[0.2cm]
\hfill
\resizebox{0.49\textwidth}{!}{
\begin{tabular}{ P{5.5cm}|Y{0.2cm} Y{4.8cm}  }
 \hline
  $\textbf{T2}$&$\epsilon$&$(a,1< c_{1}< 2\wedge  1<c_{2}<2\wedge c_1=c_2 )$\\
 \hline
 $\epsilon$&$-$&$+,\bot,\top$\\
  $(a,\{1.05,1.05\},\{\bot,\top\})$   & $+$    &$-,\top,\top$\\
   \hline
   $(a,\{0,0\},\{\top,\top\})$   & $-$    &$-,\top,\top$\\
  $(b,\{0,0\},\{\top,\top\})$   & $-$    &$-,\top,\top$\\
  $(a,\{1.05,1.05\},\{\bot,\top\})(a,\{0,0\},\{\top,\top\})$   & $-$    &$-,\top,\top$\\
  $(a,\{1.05,1.05\},\{\bot,\top\})(b,\{0,0\},\{\top,\top\})$   & $-$    &$-,\top,\top$\\
  $(a,\{1.05,1.05\},\{\bot,\top\})(b,\{1.05,0\},\{\bot,\top\})$   & $+$    &$-,\top,\top$\\
  $(a,\{1.05,1.05\},\{\bot,\top\})(b,\{2.05,1\},\{\top,\top\})$   & $-$    &$-,\top,\top$\\
    $(a,\{0,0\},\{\top,\top\})(a,\{1.05,1.05\},\{\top,\top\})$   & $-$    &$-,\top,\top$\\
 \hline
\end{tabular}
}
\end{center}
\vspace{-0.2cm}
\caption{Example of making the observation table consistent.}
\label{fig:consistent}
\vspace{-0.1cm}
\end{figure}

% \begin{description}[itemsep=5pt, align=left, leftmargin=0pt, font=\normalfont\itshape]

\textbf{Making \bm{$\mathbf{T}$} closed.}  If there is $r\in\bm{R}$ such that for all $s\in\bm{S}$, $\mathit{row}(r)\neq \mathit{row}(s)$, the learner will move $r$ from $\bm{R}$ to $\bm{S}$. Moreover, for each $\sigma\in\Sigma$, a reset-clocked word $\pi(vw(r)\cdot \bm{\sigma})$ needs to be added to $\bm{R}$, where $\bm{\sigma}=(\sigma,\{0\}^{|\mathcal{C}|})$.
%$\Pi_{1:1}\bm{\sigma}=\sigma$ and $\Pi_{2:2}\bm{\sigma}=\{0\}^{\lvert\mathcal{C}\rvert}$. 
The mapping $\pi$ defined in Section~\ref{sec:preliminaries} can be implemented by a reset information query in the powerful teacher situation.

\textbf{Making \bm{$\mathbf{T}$} consistent.} Considering the first case, there exist two reset-clocked words $\gamma_{r},\gamma_{r}'\in \bm{S}\cup\bm{R}$ at least, for  $\bm{\sigma_{r}},\bm{\sigma_{r}}'\in \bm{\Sigma_{r}}$ such that $\llbracket vw(\bm{\sigma_{r}})\rrbracket=\llbracket vw(\bm{\sigma_{r}}')\rrbracket$,  both $\gamma_{r}\bm{\sigma_{r}}$ and $\gamma_{r}'\bm{\sigma_{r}}'\in \bm{S}\cup\bm{R}$. However, $\mathit{row}(\gamma_{r})=\mathit{row}(\gamma_{r}')$ but $\mathit{row}(\gamma_{r}\bm{\sigma_{r}})\neq\mathit{row}(\gamma_{r}'\bm{\sigma_{r}'})$. In this case, there must be $\emph{e}\in\bm{E}$ such that $f(\gamma_{r}\bm{\sigma_{r}},\emph{e})\neq f(\gamma_{r}'\bm{\sigma_{r}'},\emph{e})$ or $g(\gamma_{r}\bm{\sigma_{r}},\emph{e})\neq g(\gamma_{r}'\bm{\sigma_{r}'},\emph{e})$. Thus, the learner adds $\llbracket vw(\bm{\sigma_{r}})\rrbracket\cdot\emph{e}$ to $\bm{E}$, so that $\gamma_{r}$ and $\gamma_{r}'$ are separated. Consider the second case that $resets(\bm{\sigma_{r}})\neq resets(\bm{\sigma_{r}}')$. In this case, $\llbracket vw(\bm{\sigma_{r}})\rrbracket$ 
is enough to distinguish $\gamma_{r}$ from $\gamma_{r}'$. So, the learner adds $\llbracket vw(\bm{\sigma_{r}})\rrbracket$ to $\bm{E}$. After extending $\bm{E}$, the learner fills the table by making reset information queries and membership queries.

\textbf{Making \bm{$\mathbf{T}$} evidence-closed.} In this case, $\exists \emph{s}\in\bm{S}$ and $\exists \emph{e}\in\bm{E}$, $\textit{vs}_{\mathcal{A}}(\emph{s}, e)$ is not empty, and there is no $e_{r}\in \textit{vs}_{\mathcal{A}}(\emph{s}, e)$ such that $\emph{s} e_{r}\in \bm{S}\cup\bm{R}$. The learner finds $e_{r}\in \textit{vs}_{\mathcal{A}}(\emph{s}, e)$ and then adds all prefixes of $\emph{s} e_{r}$ to $\bm{R}$, except for those already in $\bm{S}\cup\bm{R}$. 
% In the second case, $\exists \emph{s}\in\bm{S}$ and $\exists \emph{e}\in\bm{E}$, if there is no $\left\{\gamma\right\}=\emph{e}$ causing $\Pi_{\{1,2,\cdots,m+1\}}\emph{s}\cdot \gamma$ to be valid, there is no $\left\{\gamma'\right\}=e$ such that $\pi(\Pi_{\{1,2,\cdots,m+1\}}\emph{s}\cdot\gamma')\in \bm{S}\cup\bm{R}$. We find $\gamma'$ satisfying that $\left\{\gamma'\right\}=\emph{e}$. Next, all prefixes of $\pi(\Pi_{\{1,2,\cdots,m+1\}}\emph{s}\cdot\gamma')$ also need to be added to $\bm{R}$. 
Similarly, the learner needs to fill the table.
% \end{description}

The condition that an observation table is reduced is inherently preserved by the aforementioned operations, together with the counterexample
processing described later in Section~\ref{sbsc:eqctx}. Furthermore, a table may need several rounds of
these operations before being prepared since certain conditions may
be violated by different interleaved operations.

\begin{example}\label{example:consistent}
Fig.~\ref{fig:consistent} shows a process for making the table consistent. Assuming that the CTA in Fig.~\ref{fig:timedautomata} recognizes the target timed language and an observation table \textbf{T1} is generated in the learning process. We show the inconsistency of $\textbf{T1}$ as follows. Let $\gamma_r=\epsilon$ and $\gamma_r'=(a,\{0,0\},\{\top,\top\})$, we have $\mathit{row}(\gamma_r)=\mathit{row}(\gamma_r')$. Choose $\bm{\sigma_r}=(a,\{1.05,$ $1.05\},\{\bot,\top\})$ and $\bm{\sigma_r}'=(a,\{1.05,$ $1.05\},\{\top,\top\})$ such that $\llbracket vw(\bm{\sigma_{r}})\rrbracket=\llbracket vw(\bm{\sigma_{r}}')\rrbracket$. We have $\gamma_r\bm{\sigma_r}=(a,\{1.05,1.05\},\{\bot,\top\})$ and $\gamma_r'\bm{\sigma_r}'=(a,\{0,0\},\{\top,\top\})$ $(a,\{1.05,1.05\},\{\top,\top\})$ in $\bm{S}\cup\bm{R}$. However, we find $\mathit{row}(\gamma_r\bm{\sigma_r})\neq\mathit{row}(\gamma_r'\bm{\sigma_r}')$ and $resets(\bm{\sigma_r})\neq resets(\bm{\sigma_r}')$. Thus, $\textbf{T1}$ violates the two conditions of consistency. In order to make the table consistent, we add $\llbracket(a,\{1.05,1.05\})\rrbracket\cdot\epsilon$, which equals $(a,1\textless c_{1}\textless 2\wedge  1<c_{2}< 2\wedge c_1=c_2)$, to $\bm{E}$ to distinguish $\epsilon$ and $(a,\{0,0\},\{\top,\top\})$, and then fill the table and get \textbf{T2}.% shown in Fig.~\ref{fig:consistent}.
\end{example}

\subsection{Hypothesis construction}\label{sbsc:hypo}
When the observation table $\mathbf{T}$ is prepared, a hypothesis $\mathcal{H}$ can be constructed based on $\mathbf{T}$ in two steps: 1. Build a DFA $\text{M}$ from $\mathbf{T}$; 2. Convert $\text{M}$ into DTA $\mathcal{H}$. 

Given a prepared table $\mathbf{T}=(\Sigma, \bm{\Sigma_{G}}, \bm{\Sigma_{r}}, \bm{S}, \bm{R}, \bm{E}, f, g)$, the learner first builds a DFA $\text{M}=(L_{M},l_M^0, F_{M},$ $\Sigma_{M}, \Delta_{M})$ as follows. 
\begin{itemize}
    \item the finite set of locations $L_{M}=\{l_{\mathit{row}(s)} \mid s\in\bm{S}\}$;
    \item the initial location $l_M^0=l_{\mathit{row}(\epsilon)}$ for $\epsilon\in\bm{S}$;
    \item the finite set of accepting locations $F_{M}=\{l_{\mathit{row}(s)} \mid f(s\cdot\epsilon)=+ \text{ for } s\in\bm{S} \text{ and } \epsilon\in\bm{E}\}$;
    \item the finite abstract alphabet $\Sigma_{M}=\{\bm{\sigma_r}\in\bm{\Sigma_r} \mid \gamma_r\bm{\sigma_r}\in\bm{S}\cup\bm{R} \text{ for } \gamma_r\in\bm{\Sigma_r}^*\}$;
    \item the finite set of transitions $\Delta_M = \{(l_{\mathit{row}(\gamma_r)},\bm{\sigma_r},l_{\mathit{row}(\gamma_r\bm{\sigma_r})}) \mid $ $ \gamma_r,\gamma_r\bm{\sigma_r}\in\bm{S}\cup\bm{R} \text{ for } \bm{\sigma_r}\in\bm{\Sigma_r}\}$.
\end{itemize}

The basic idea in the above construction is that we view $\bm{\sigma_r}\in\bm{\Sigma_r}$ as an abstract action if $\bm{\sigma_r}$ occurs in the table and consider each row in the table as the location which is reached by the row.  We label the locations by function $\mathit{row}$. The constructed DFA $\text{M}$ is compatible with the timed observation table $\mathbf{T}$ in a sense captured by the following lemma. 

\begin{lemma}\label{lemma:tabledfa}
	Given a prepared observation table $\mathbf{T}=(\Sigma, \bm{\Sigma_{G}}, \bm{\Sigma_{r}}, $ $\bm{S}, \bm{R}, \bm{E}, f, g)$,  for all $\gamma_{r}\cdot e\in(\bm{S}\cup\bm{R})\cdot\bm{E}$, we have
	\begin{itemize}
	    \item if $f(\gamma_{r}, e)=+$, there is $\gamma_r'\in vs_{\mathcal{A}}(\gamma_{r}, e)$ such that $resets(\gamma_r')=g(\gamma_{r}, e)$ and $\gamma_{r}\gamma_r'$ is accepted in the constructed DFA $\text{M}$;
	    \item if $f(\gamma_{r}, e)=-$ and $vs_{\mathcal{A}}(\gamma_{r}, e)\neq\emptyset$, there is $\gamma_r'\in vs_{\mathcal{A}}(\gamma_{r}, e)$ such that $resets(\gamma_r')=g(\gamma_{r}, e)$ and $\gamma_{r}\gamma_r'$ is not accepted in the constructed DFA $\text{M}$. 
	    %\item if $vs_{\mathcal{A}}(\gamma_{r},e)=\emptyset$, there is no accepting word in DFA $\text{M}$ in the form of $\gamma_r\gamma_r'$ such that $\gamma_r'\in vs_{\mathcal{A}}(\gamma_r'',e)$ for any $\gamma_r''\in\bm{S}\cup\bm{R}$ with $\mathit{row}(\gamma_r'')=\mathit{row}(\gamma_r)$.
	    \item if $f(\gamma_{r}, e)=-$ and $vs_{\mathcal{A}}(\gamma_{r},e)=\emptyset$, DFA $\text{M}$ has no accepting word in the form of $\gamma_r\gamma_r'$ such that $\llbracket vw(\gamma_r')\rrbracket=e$.
	\end{itemize}
\end{lemma}

After the DFA $\text{M}$ is constructed, the learner constructs a hypothesis $\mathcal{H}=(L, l_0, F, \mathcal{C},$ $\Sigma, \Delta)$ from $\text{M}$. In the hypothesis $\mathcal{H}$, $\mathcal{C}$ is the set of clocks known before learning and $\Sigma$ is the given alphabet as in~$\mathbf{T}$. In addition, $L=L_{M}$, $l_{0}=l_M^0$, and $F=F_{M}$.
The set of transitions $\Delta$ in $\mathcal{H}$ can be constructed as follows: For any $l\in L_{M}$ and $\sigma\in\Sigma$, let
$\Psi_{l,\sigma}=\{\mathbf{v}_i\mid (l,(\sigma,\mathbf{v}_i,\mathbf{b}_i),l')\in \Delta_{M}\}$,
then applying the partition function $P^c(\cdot)$ to $\Psi_{l,\sigma}$ returns $n$ disjoint clock constraints, written as $I_{1}, \cdots, I_{n}$, satisfying $\mathbf{v}_{i}\in I_{i}$ for any $1\leq i\leq n$, where $n=|\Psi_{l,\sigma}|$. Then for each transition $(l,(\sigma,\mathbf{v}_i,\mathbf{b}_i),l')\in \Delta_{M}$, we add a transition $(l, \sigma, I_i, \{c_j \vert \mathbf{b}_{i,j}=\top\}, l')$ in $\Delta$. After that, the construction of hypothesis $\mathcal{H}$ is finished.

\begin{definition}[Partition function]\label{def:partition} Given a set of clock valuations $\Psi_{l,\sigma}$, the partition function $P(\cdot)$ operates on it in the five steps:
\begin{itemize}

\item For all clock valuations in $\Psi_{l,\sigma}$, we list them as $\mathbf{v}_{1}, \mathbf{v}_{2},\cdots,$ $\mathbf{v}_{n}$, with $\{0\}^{|\mathcal{C}|}=\mathbf{v}_{1}\preceq\mathbf{v}_{2}\preceq\cdots\preceq\mathbf{v}_{n}$. Then for $1\leq i \leq n$, we map $ \mathbf{v}_{i}$ to a constraint $A_i$, which is defined as
\begin{equation*}
\quad\quad A_i = \begin{cases}
	\llbracket \mathbf{v}_{i} \rrbracket, & \text{if}\ \  \exists 1\leq j \leq |\mathcal{C}| \text{ satisfying that }\mathbf{v}_{i,j}>\kappa(c_j); \\
	\emptyset, & \text{otherwise}.
	\end{cases}
\end{equation*}
We denote the union of $A_{i}$ for all $1\leq i \leq n$ by $U_0$.

\item After $U_0$ is constructed, for $1\leq i \leq n$, we map clock valuation $\mathbf{v}_i$ to a new constraint $U_i=\bigwedge_{1\leq j \leq |\mathcal{C}|}ic_{i,j}$ where 
\begin{equation*}
 \quad\quad ic_{i,j} = \begin{cases}
	c_{j}\geq \mathbf{v}_{i,j}, & \text{if}\ \  \mathbf{v}_{i,j}\in \mathbb{N}; \\
	c_{j}\textgreater\lfloor \mathbf{v}_{i,j} \rfloor, & \text{otherwise}.
	\end{cases}
\end{equation*}
\item Based on $U_0$ constructed in the first step and $U_1,\cdots, U_n$ constructed in the second step, map the $n$ clock valuations $\mathbf{v}_{1}, \mathbf{v}_{2},\cdots,\mathbf{v}_{n}$ to the $n$ constraints $W_{1}, W_{2}, \cdots,$ $W_{n}$, respectively, where $W_i$ is defined as
\begin{equation*}
 \quad\quad \begin{cases}
	U_{n}-(U_{n}\cap U_{0}), & \text{if}\ \  i=n; \\
	U_{i}-(U_{i}\cap(W_{i+1}\cup W_{i+2}\cup\cdots\cup W_{n}\cup U_0)), & \text{otherwise}.
\end{cases}
\end{equation*}

\item Based on $W_1,\cdots, W_n$ and $A_1,\cdots, A_n$, map the $n$ clock valuations $\mathbf{v}_{1}, \mathbf{v}_{2},\cdots,\mathbf{v}_{n}$ to the $n$ constraints $I_{1}, I_{2}, \cdots,$ $I_{n}$, respectively, which is defined as
\begin{equation*}
 \quad\quad I_{i} = W_{i}+A_i.
\end{equation*}

\item Finally, find whether there are $\mathbf{v}_{i}, \mathbf{v}_{j}$ such that $A_i=A_j=\emptyset$ and $U_{i}=U_{j}$, but $\llbracket\mathbf{v}_{i}\rrbracket\neq\llbracket\mathbf{v}_{j}\rrbracket$.
% $\lfloor\mathbf{v}_{i,k}\rfloor=\lfloor\mathbf{v}_{j,k}\rfloor$ {\color{red}and $\mathbf{v}_{i,k}\notin \mathbb{N}$ iff $\mathbf{v}_{j,k}\notin \mathbb{N}$} for all $1\leq k \leq |\mathcal{C}|$.
% and both of $\{\vec{\nu_{i}}\}$ and $\{\vec{\nu_{j}}\}$ are open regions.
If so, we support that $\mathbf{v}_i\preceq\mathbf{v}_j$, then after the previous three steps, $W_{i}=\emptyset$ and $W_{j}=U_{j}-(U_{j}\cap(W_{j+1}\cup W_{j+2}\cup\cdots\cup W_{n}\cup U_0))$ since $U_i=U_j$. Next, in the fourth step, $I_i=\emptyset$ and $I_j=W_j$ since $A_i=A_j=\emptyset$. To map $\mathbf{v}_{i}$ to the correct partition, let $I_{i}=\llbracket\mathbf{v}_{i}\rrbracket$ and thus the updated $I_{j}=W_j-I_{i}$. Note that after this step, the union of the updated $I_{i}$ and $I_{j}$ is not changed.

\end{itemize}
\end{definition}

The basic idea of the partition function is to generate disjoint clock constraints $I_1,\cdots,I_n$ containing the clock valuations. Specially, we need to generate that for each $\mathbf{v}_{i}$, there is one and only one final constraint $I_i$ containing it. In the first step, to avoid dividing a region $\llbracket\mathbf{v}_{i}\rrbracket$ into more than one final constraint, we select $\llbracket\mathbf{v}_{i}\rrbracket$ which may lead to this situation, add it to $U_0$ and guarantee that $U_0$ does not participate in the division operations in the second step and third step. In the second step, we compute a coarse constraints $U_i$ containing the corresponding clock valuations $\mathbf{v}_i$. But they are not disjoint. In the third step, we guarantee the disjoint property and get the constraints $W_i$. In the fourth step, since $U_0$ does not participate in the division operations in the second step and third step, we select $A_i$ from $U_0$ and add it to $W_i$. In the second step and third step, we do not consider the relation between clocks in one clock valuation. But they may form different regions (constraints). Hence, for two clock valuations such that $\llbracket\mathbf{v}_{i}\rrbracket\neq\llbracket\mathbf{v}_{j}\rrbracket$ but $U_{i}=U_{j}$ and $A_i=A_j=\emptyset$, if $\mathbf{v}_i\preceq\mathbf{v}_j$, then $I_i=\emptyset$, which means that $\mathbf{v}_{i}$ has no corresponding final constraint containing it. Therefore, in the final step,  we divide the constraint $I_j$ into two constraints $I_i$ and $I_j$. 

\begin{lemma}\label{lemma:partitiondeterministic} The clock constraints $I_{1}, \cdots, I_{n}$ generated by partition function $P(\cdot)$ are disjoint, i.e. $I_{i}\cap I_{j}=\emptyset$ for all $i\neq j$ and $1\leq i,j\leq n$.
\end{lemma}

\begin{theorem}
    \label{theorem:partitioncomplete} $P(\cdot)$ returns
a partition on
$\mathbb{R}_{\geq0}^{\lvert \mathcal{C} \rvert}$.

\end{theorem}

\begin{lemma}\label{lemma:dfatocta}
Given a DFA $\text{M}=(L_{M},l^0_M, F_{M}, \Sigma_{M}, \Delta_{M})$, which is generated from a prepared timed observation table $\mathbf{T}$, the hypothesis $\mathcal{H}=(L, l_0, F, \mathcal{C}, \Sigma, \Delta)$ is transformed from $\text{M}$. For all $\gamma_{r}\cdot e\in(\bm{S}\cup\bm{R})\cdot\bm{E}$,
we have
	\begin{itemize}
	    \item if $f(\gamma_{r}, e)=+$, there is $\gamma_r'\in vs_{\mathcal{A}}(\gamma_{r}, e)$ such that $resets(\gamma_r')=g(\gamma_{r}, e)$ and $\gamma_{r}\gamma_r'$ is accepted in the constructed hypothesis $\mathcal{H}$;
	    \item if $f(\gamma_{r}, e)=-$ and $vs_{\mathcal{A}}(\gamma_{r}, e)\neq\emptyset$, there is $\gamma_r'\in vs_{\mathcal{A}}(\gamma_{r}, e)$ such that $resets(\gamma_r')=g(\gamma_{r}, e)$ and $\gamma_{r}\gamma_r'$ is not accepted in the constructed hypothesis $\mathcal{H}$. 
	    \item if $f(\gamma_{r}, e)=-$ and $vs_{\mathcal{A}}(\gamma_{r},e)=\emptyset$,  $\mathcal{H}$ has no accepting word in the form of $\gamma_r\gamma_r'$ such that $\llbracket vw(\gamma_r')\rrbracket=e$.
	\end{itemize}
\end{lemma}

\begin{theorem}
    \label{theorem:cta}
    The hypothesis $\mathcal{H}$ is a CTA.
\end{theorem}

\begin{theorem}
\label{theorem:7} For all $\gamma_{r}\cdot
e\in(\bm{S}\cup\bm{R}\cdot\bm{E})$, for every valid clocked word $\gamma'$ such that  $\llbracket\gamma'\rrbracket=\llbracket vw(\gamma_r)\rrbracket\cdot e$, $\mathcal{H}$ accepts  $\pi(\gamma')=\gamma_{r}'$ if and only if $f(\gamma_{r}, e)=+$.
\end{theorem}

\subsection{Equivalence query and counterexample processing}\label{sbsc:eqctx}
The equivalence problem determining whether $\mathcal{L}(\mathcal{H})=\mathcal{L}(\mathcal{A})$ can be decomposed into the two-side inclusion problem, i.e., whether $\mathcal{L}(\mathcal{H})\subseteq\mathcal{L}(\mathcal{A})$ and $\mathcal{L}(\mathcal{A})\subseteq\mathcal{L}(\mathcal{H})$. Since the deterministic timed automata are closed under complementation~\cite{AlurD94}, we can implement the equivalence oracle as follows: $\mathcal{L}(\mathcal{H})\subseteq\mathcal{L}(\mathcal{A})$ iff $\mathcal{L}(\mathcal{H})\cap\overline{\mathcal{L}(\mathcal{A})}=\emptyset$, and $\mathcal{L}(\mathcal{A})\subseteq\mathcal{L}(\mathcal{H})$ iff $\mathcal{L}(\mathcal{A})\cap\overline{\mathcal{L}(\mathcal{H})}=\emptyset$. If $\mathcal{L}(\mathcal{H})\not\subseteq\mathcal{L}(\mathcal{A})$, there is a delay-timed word $\omega$ where the corresponding run $\rho$ in $\mathcal{H}$ reaches an accepting location and the corresponding run $\rho'$ in $\mathcal{A}$ reaches an unaccepting location. In this case, the teacher provides the corresponding reset-delay-timed word of $\omega_r$ in the target automaton $\mathcal{A}$ as a negative counterexample $\mathit{ctx}_{-}=(\omega_{r},-)$. Similarly, if $\mathcal{L}(\mathcal{A})\not\subseteq\mathcal{L}(\mathcal{H})$, the teacher provides a positive counterexample $\mathit{ctx}_{+}=(\omega_{r},+)$. When receiving a counterexample $\mathit{ctx}=(\omega_{r},+/-)$, the learner converts the counterexample to a reset-clocked word $\gamma_{r}=(\sigma_1,\mathbf{v}_1,\mathbf{b}_1)(\sigma_2,\mathbf{v}_2,\mathbf{b}_2)\cdots(\sigma_n,\mathbf{v}_n,\mathbf{b}_n)$ and add all prefixes of $\gamma_r$ to $\bm{R}$ except for those already in $\bm{S}\cup\bm{R}$. According to the definition of the reset-clocked word, $\omega_{r}$ and $\gamma_{r}$ trigger the same sequence of transitions in $\mathcal{A}$. According to the  contraposition of Theorem~\ref{theorem:basicidea}, if $\omega$ is an evidence of $\mathcal{L}(\mathcal{H})\neq\mathcal{L}(\mathcal{A})$, $\gamma_{r}$ is an evidence of $\mathscr{L}_{r}(\mathcal{H})\neq \mathscr{L}_{r}(\mathcal{A})$.

\subsection{Learning algorithm}\label{sbsc:algorithm}

\begin{algorithm}[!t]
%\scriptsize
 \caption{Learning deterministic timed automata from a powerful teacher}
	\label{alg:learning}
	\SetKwInOut{Input}{input}
	\SetKwInOut{Output}{output}
	%\SetKwRepeat{Do}{do}{while}
	\Input{
 % the timed observation table $\mathbf{T} = (\Sigma, \bm{\Sigma}, \bm{\Sigma_r}, \bm{S}, \bm{R}, \bm{E}, f, g)$;
  the timed observation table  $\mathbf{T}=(\Sigma, \bm{\Sigma_{G}}, \bm{\Sigma_{r}}, \bm{S}, \bm{R}, \bm{E}, f, g)$;
 the number of clocks $|\mathcal{C}|$.}
	\Output{the hypothesis $\mathcal{H}$ recognizing the target language $\mathcal{L}$.}
	$\bm{S}\leftarrow\{\epsilon\}$;
	$\bm{R}\leftarrow\{\pi(\gamma) \mid \gamma=(\sigma,\{0\}^{|\mathcal{C}|}), \forall \sigma \in \Sigma \}$;
	$\bm{E}\leftarrow\{\epsilon\}$ \tcp*{initialization}
	fill $\mathbf{T}$ by membership queries\;
	$\mathit{equivalent}$ $\leftarrow$ $\bot$\;
	\While{$\mathit{equivalent}$ = $\bot$}{
		$\mathit{prepared}$ $\leftarrow$ is\_prepared($\mathbf{T}$) \tcp*{whether the table is prepared}
		\While{$\mathit{prepared}$ = $\bot$}
		{
			\uIf{$\mathbf{T}$ is not closed}{
				$\!$make\_closed($\mathbf{T}$)
			}
			\uIf{$\mathbf{T}$ is not consistent}{
				$\!$make\_consistent($\mathbf{T}$)
			}
			\uIf{$\mathbf{T}$ is not evidence-closed}{
				$\!$make\_evidence\_closed($\mathbf{T}$)
			}
			%\lIf{$\mathbf{T}$ is not prefix-closed}{
				%$\!$make\_prefix\_closed($\mathbf{T}$)
			%}
			$\mathit{prepared}$ $\leftarrow$ is\_prepared($\mathbf{T}$)\;
		}
		$\text{M} \leftarrow$ build\_DFA($\mathbf{T}$) \tcp*{transforming $\mathbf{T}$ to a DFA $\text{M}$}
		$\mathcal{H} \leftarrow$ build\_hypothesis($\text{M}$) \tcp*{constructing a hypothesis $\mathcal{H}$ from $\text{M}$}
		$\mathit{equivalent}$, $\mathit{ctx}$ $\leftarrow$ equivalence\_query($\mathcal{H}$)\;
		\If{$\mathit{equivalent}$ = $\bot$}{
			ctx\_processing($\mathbf{T}$, $\mathit{ctx}$) \tcp*{counterexample processing}
		}
	}
	\Return $\mathcal{H}$\;
\end{algorithm}

The learning process that integrates all the previously stated
ingredients is presented in Algorithm~\ref{alg:learning}. 
First, the learner initializes a timed observation table $\mathbf{T}$, where $\bm{S}=\{\epsilon\}$, $\bm{E}=\{\epsilon\}$. For each $\sigma\in\Sigma$, the learner obtains the corresponding reset-clocked word $\pi(\gamma)$ and the acceptance status in $\mathcal{A}$ of the clocked word $\gamma=(\sigma,\{0\}^{\lvert \mathcal{C} \rvert})$ through membership queries and adds the results to $\bm{R}$. Before constructing a hypothesis, the learner performs several rounds of operations described in Section~\ref{sbsc:membership} until $\mathbf{T}$ is prepared. Then, a hypothesis $\mathcal{H}$ is constructed leveraging an intermediate DFA $\text{M}$ and submitted to the teacher for an equivalence query. If the answer is positive, $\mathcal{H}$ recognizes the target language. Otherwise, the learner receives a counterexample $\mathit{ctx}$ to update $\mathbf{T}$.
The whole procedure repeats until the teacher gives a positive answer to an equivalence query. 
An illustrative case on learning DTA $\mathcal{A}$ in Figure~\ref{fig:timedautomata} is presented in Appendix~\ref{appendix:illustrative_case}. 
We have the following result.

\begin{theorem}\label{theorem:powerfulcorrectness}
Algorithm~\ref{alg:learning} terminates and returns a CTA $\mathcal{H}$ which recognizes the target timed language $\mathcal{L}$.
\end{theorem}

\subsection{Complexity}\label{sbsc:complexity}
Suppose that $\mathcal{L}$ is a target timed language accepted by a DTA $\mathcal{A}$, let $n=\lvert L\rvert$ represent the number of locations of $\mathcal{A}$, $m=\lvert \Sigma\rvert$ represent the size of the alphabet, 
% {\color{red}$\kappa_{ass}$ be a function mapping each clock $c\in\mathcal{C}$ to $k_{max}$ instead of mapping $c$ to the largest integer appearing in the guards over $c$,}
$\kappa$ be a function mapping each clock $c\in\mathcal{C}$ to the largest integer appearing in the guards over $c$ 
and $\Lambda=|\mathcal{C}|!\cdot 2^{|\mathcal{C}|}\cdot\prod_{c\in\mathcal{C}}(2\kappa(c)+2)$ represent the bound of the number of clock regions. 
We measure the number of queries.

\noindent\textbf{Equivalence query.} According to the proof of Theorem~\ref{theorem:powerfulcorrectness}, $\mathcal{H}$ has at most $n\cdot \Lambda$ locations. The number of transitions from a location for an action is at most $\Lambda$.
Therefore, the number of transitions in $\mathcal{H}$ is at most $mn\Lambda^{2}$. As every counterexample adds at least one fresh transition to $\mathcal{H}$ if considering each final timing constraint of the partition corresponds to a transition, the number of equivalence queries is bounded by $mn\Lambda^{2}$.

\noindent\textbf{Membership query.} In order to count the number of membership queries, $(\lvert\bm{S}\rvert+\lvert\bm{R}\rvert)\cdot\lvert\bm{E}\rvert$ needs to be counted. 
Since in the observation table, $\bm{E}$ is used to distinguish the elements in $\bm{S}$, we have $|\bm{E}| \leq n\cdot \Lambda$. 
Next, $\lvert\bm{R}\rvert$ is analyzed as follows. Let $h$ denote the maximum length of all counterexamples obtained by the learner. The algorithm adds elements to $\bm{R}$ in three cases: processing of counterexamples, making $\mathbf{T}$ closed, and making $\mathbf{T}$ evidence-closed. 
In the first case, because the number of counterexamples equals to the number of equivalence queries, which is bounded by $mn\Lambda^{2}$, the number of prefixes added to $\bm{R}$ is at most $hmn\Lambda^{2}$ in this case. In the second case, since for each prefix in $\bm{S}$ and each action, there is a corresponding prefix in $\bm{R}$, the number of prefixes added to $\bm{R}$ in this step is at most $mn\cdot \Lambda$. In the last case, the number of prefixes added to $\bm{R}$ depends on $\lvert\bm{S}\rvert\cdot\lvert\bm{E}\rvert$, which is at most $(n\cdot \Lambda)^2$. Therefore, the number of membership queries is bounded by $(n\cdot \Lambda+hmn\Lambda^{2}+mn\cdot \Lambda+(n\cdot \Lambda)^2)\cdot n\cdot \Lambda$.

\noindent\textbf{Reset information query.} In each membership query, the number of reset information queries depends on the length of the suffix. Let $p$ be the length of the longest suffix in $\bm{E}$. During the learning process, each new suffix added to $\bm{E}$ is composed of a region word with length one and a suffix that has already been in $\bm{E}$. Therefore, $p$ is at most $\lvert\bm{E}\rvert-1$. Hence, the number of reset information queries to fill the observation table is bounded by $(\lvert\bm{E}\rvert-1)\cdot(\lvert\bm{S}\rvert+\lvert\bm{R}\rvert)\cdot\lvert\bm{E}\rvert$. Moreover, to make the observation table closed, for each $s\in \bm{S}$, $m$ reset information queries are needed. As a result, the number in total is bounded by $(\lvert\bm{E}\rvert-1)\cdot(\lvert\bm{S}\rvert+\lvert\bm{R}\rvert)\cdot\lvert\bm{E}\rvert+m\cdot\lvert\bm{S}\rvert$.

\section{Learning DTA from a normal teacher}\label{sec:normalteacher}

In this section, we extend the learning algorithm for learning a DTA $\mathcal{A}$ from a normal teacher. In this situation, the teacher can only answer the standard membership query and equivalence query. In a standard membership query, the input is delay-timed words, and the teacher returns whether the timed word is in the target timed language or not. In a standard equivalence query, if $\mathcal{L}(\mathcal{A})\neq\mathcal{L}(\mathcal{H})$, the teacher provides a delay-timed word as a counterexample. 

The algorithm in the normal teacher situation is based on the learning procedure in Section~\ref{sec:powerfulteacher}.
In this algorithm, the structure of the observation table is unchanged. $\bm{S}\cup\bm{R}$ still stores reset-clocked words, and $\bm{E}$ still stores region words. In order to obtain reset-clocked words from delay-timed words of the membership queries, the learner \emph{guesses} clock reset information in the table. 

Since guessing the reset information, for each reset case, it will result in a table instance. A list $\mathit{ToExplore}$ is maintained to collect all these table instances that need to be explored later. We use breadth-first search (BFS) to explore the tables in $\mathit{ToExplore}$, which is crucial for ensuring the termination. Since the algorithm cannot detect incorrect guesses of reset information, depth-first search (DFS) may not terminate if it follow a branch that never finds the correct hypothesis. Moreover, our BFS is not based on a strictly FIFO queue. In each round, we take out the table instance with the minimal number of guessed reset. Since the table instance contains reset information, the operations on the table instance are the same as those in the powerful teacher situation. 
% Details are as follows.

%The detailed process is given in {\color{red}Appendix~\ref{appendix:learnfromnormal}}. 
% Therefore, the learning process can be thought of as exploring a search tree, where branching is caused by guessing resets. 
The algorithmic summary (Algorithm~\ref{alg:learningnormal}) is provided in Appendix~\ref{appendix:learnfromnormal}, as its structure is similar to Algorithm~\ref{alg:learning}. 
Details are as follows.

The initial tables in $\mathit{ToExplore}$ are as follows. 
For each action $\sigma\in\Sigma$, there is one row in $\bm{R}$ corresponding to the clocked word $\gamma=(\sigma, \{0\}^{\lvert \mathcal{C} \rvert})$. 
To transform it to a reset-clocked word, the learner needs to guess the reset information of each clock in $\gamma$, get a reset-clocked word $\gamma_r=(\sigma, \{0\}^{\lvert \mathcal{C} \rvert}, \mathbf{b})$ and add $\gamma_r$ into $\bm{R}$. For each action $\sigma$, there are $2^{\lvert \mathcal{C} \rvert}$ possible combinations of guesses. So, in this step, $2^{\lvert \mathcal{C} \rvert\lvert\Sigma\rvert}$ table instances are added to $\mathit{ToExplore}$.
% (Line~\ref{line:initaltables}).

In each iteration of the algorithm, one table instance is taken out of
$\mathit{ToExplore}$. The learner checks whether the table is closed,
consistent, and evidence-closed. If the table is not closed,
i.e. there exists $r\in\bm{R}$ such that $row(r)\neq row(s)$ for all
$s\in \bm{S}$. In this case, the learner moves $r$ from $\bm{R}$ to $\bm{S}$. Then
for each $\sigma\in \Sigma$, a word $\gamma=vw(r)\cdot (\sigma,\{0\}^{\lvert \mathcal{C} \rvert})$ is created, and the learner needs to guess the reset information of $(\sigma,\{0\}^{\lvert \mathcal{C} \rvert})$. Similar to the initialization of the table, there are $2^{\lvert \mathcal{C} \rvert}$ possible combinations of guesses for action $\sigma$. Therefore, $2^{\lvert \mathcal{C} \rvert\lvert\Sigma\rvert}$ unfilled table
instances will be generated. Next, these $2^{\lvert \mathcal{C} \rvert\lvert\Sigma\rvert}$ unfilled table instances are filled by membership queries. Because the learner cannot initiate the reset information query, for each added entry in $\bm{R}$, it is necessary to guess
reset information for all transitions in
$e\in\bm{E}$. For an element $e\in\bm{E}$, the reset information of $\lvert e\rvert\cdot \lvert\mathcal{C}\rvert$ clocks needs to be guessed. Therefore, for each new row in $\bm{R}$, a total of $(\sum_{e_i \in \bm{E}\setminus\{\epsilon\}}{\lvert e_i \rvert}) \cdot \lvert \mathcal{C} \rvert$ clock reset information needs to be guessed. Therefore, this procedure inserts at most $2^{(\sum_{e_i \in \bm{E}\setminus\{\epsilon\}}{\lvert e_i \rvert}) \cdot \lvert \mathcal{C} \rvert^2\cdot \lvert \Sigma \rvert}$ table instances to $\mathit{ToExplore}$.

If the table is not consistent, i.e. there exist two reset-clocked words $\gamma_{r},\gamma_{r}'\in \bm{S}\cup\bm{R}$ at least, for  $\bm{\sigma_{r}},\bm{\sigma_{r}}'\in \bm{\Sigma_{r}}$ such that $\llbracket vw(\bm{\sigma_{r}})\rrbracket=\llbracket vw(\bm{\sigma_{r}}')\rrbracket$,  both $\gamma_{r}\bm{\sigma_{r}}$ and $\gamma_{r}'\bm{\sigma_{r}}'\in \bm{S}\cup\bm{R}$. However, $\mathit{row}(\gamma_{r})=\mathit{row}(\gamma_{r}')$ but $\mathit{row}(\gamma_{r}\bm{\sigma_{r}})\neq\mathit{row}(\gamma_{r}'\bm{\sigma_{r}'})$ or $resets(\bm{\sigma_{r}})\neq resets(\bm{\sigma_{r}}')$.
Let $e\in\bm{E}$ be one place where they are different. Then
$\llbracket vw(\bm{\sigma_{r}})\rrbracket\cdot e$ needs to be added to $\bm{E}$. For each entry in
$\bm{S}\cup\bm{R}$, all transitions in
% $\bm{\sigma_r}\cdot e$
$\llbracket vw(\bm{\sigma_{r}})\rrbracket\cdot e$
need to be guessed, and then the table can be
filled.  Hence, $2^{\lvert \llbracket vw(\bm{\sigma_{r}})\rrbracket \cdot e \rvert \cdot\lvert \mathcal{C} \rvert\cdot (\lvert \bm{S}
  \rvert + \lvert \bm{R} \rvert)}$
filled table instances will be generated and inserted into
$\mathit{ToExplore}$. The operation for making tables evidence-closed
is analogous.

If the current table is prepared, the learner builds a hypothesis $\mathcal{H}$ and makes an equivalence query to the teacher. If the answer is negative, the teacher gives a delay-timed word
$\omega$ as a counterexample. The learner first finds the longest reset-clocked word $\gamma_r$ in $\bm{S}\cup\bm{R}$, which agrees with a prefix of $\omega$ after converting $\gamma_r$ to a delay-timed word. Then the reset information of  $\gamma_r$ is copied to $\omega$, and the reset information of the remainder of $\omega$ needs to be guessed. Hence, at most $2^{\lvert\omega\rvert\cdot \lvert \mathcal{C} \rvert}$ unfilled table instances are generated. To fill these table instances, the reset information of $\bm{E}$ also needs to be guessed. Therefore, there are at most $2^{\lvert\omega\rvert\cdot(\sum_{e_i \in \bm{E}\setminus\{\epsilon\}}{\lvert e_i \rvert}) \cdot \lvert \mathcal{C} \rvert^2}$ filled table instances need to be added in $\mathit{ToExplore}$.

%Therefore, we have the following result.
\begin{theorem} \label{theorem:normalterminate}
  \iffalse Algorithm~\ref{alg:learningnormal}\fi The algorithm learning DTA from a normal teacher terminates and returns a correct CTA $\mathcal{H}$.
  % which recognizes the target timed language $\mathcal{L}$.
\end{theorem}

\noindent\emph{Complexity analysis.} 
If $\mathbf{T}=(\Sigma, \bm{\Sigma_{G}}, \bm{\Sigma_{r}}, \bm{S}, \bm{R}, \bm{E}, f, g)$ is the final table corresponding to the correct hypothesis, 
the number of guessed resets in $\bm{S}\cup\bm{R}$ is $(|\bm{S}|+|\bm{R}|) \cdot \lvert \mathcal{C} \rvert$, and the
number of guessed resets for entries in each row of the table is $(\sum_{e_i \in \bm{E}\setminus\{\epsilon\}}{\lvert e_i \rvert}) \cdot \lvert \mathcal{C} \rvert$. 
Since we use BFS to explore the tables in $ToExplore$ w.r.t. the size of the tables, a smaller table is always handled before any larger ones. Hence, before finding the final table $\mathbf{T}$ corresponding to the correct hypothesis, the number of the handled tables is bounded by 
$\mathcal{O}(2^{(|\bm{S}|+|\bm{R}|) \cdot \lvert \mathcal{C} \rvert^2 \cdot(\sum_{e_i \in \bm{E}\setminus\{\epsilon\}}{\lvert e_i \rvert})})$.

\section{Experiments}\label{sec:experiment}

We have implemented the a prototype named \textsc{DTAL} in JAVA for our learning algorithms of deterministic timed automata with multiple clocks. Specially, we denote the implementations for the powerful and the normal teachers by \textsc{DTAL\_powerful} and by \textsc{DTAL\_normal}, respectively.  
We compared our methods with \textsc{LearnTA} proposed in~\cite{waga2023active} for learning DTAs, and with \textsc{LearnRTA} proposed in~\cite{AnWZZZ21} for learning real-time automata (RTAs). These two algorithms are implemented in C++ and Python, respectively. Note that we omitted checking of evidence-closed condition in our experiments to avoid adding too many extra rows in $\bm{R}$. %This does not affect the correctness of the algorithms.
%For the algorithm learning from the powerful teacher \textsc{DTAL\_powerful}, we evaluate it on two benchmarks in~\cite{waga2023active,AnWZZZ21}, respectively. For the algorithm learning from the normal teacher \textsc{DTAL\_normal}, we compared it with \textsc{LearnTA} by $5$ DTAs with different complexity. 
%, which can be found in Appendix~\ref{appendix:targetDTAs}. 
The evaluations have been carried out on an Intel Core-i7 processor with 16GB RAM running 64-bit Ubuntu. 

\subsection{Evaluation on \textsc{DTAL\_powerful}}

\noindent\textbf{Experiment on RTAs.} We first compared the performance of the learning algorithm \textsc{DTAL\_powerful} with the algorithm \textsc{LearnRTA} proposed in~\cite{AnWZZZ21} by learning real-time automata (RTAs). Since RTA can be viewed as a TA with only one clock which resets at every transition, we can directly use \textsc{DTAL\_powerful} to learn real-time automata instead of \textsc{DTAL\_normal} by trivially answering reset information queries. We used the $240$ randomly generated RTAs from~\cite{AnWZZZ21}, which are divided into $12$ groups depending on different numbers of locations, size of alphabet, and
maximum number of time intervals that can be generated by a single division. These three numbers make up case ID. Table~\ref{tab:RTAresults} summarizes the experimental results. It can be seen that the number of equivalence queries needed in \textsc{DTAL\_powerful}  is always very close to that in \textsc{LearnRTA}. However, the number of membership queries needed in \textsc{DTAL\_powerful} is sightly smaller than that in \textsc{LearnRTA}. This is because that we omitted checking of evidence-closed condition. Besides, the runtime needed in \textsc{DTAL\_powerful} is shorter than that in \textsc{LearnRTA}. This may be due to the fact that these two algorithms are written in JAVA and Python, respectively.

% \subsubsection{Experiment on DOTAs} We also learning randomly generated $50$ DOTAs from~\cite{XuAZ22}  by the learning algorithms \textsc{DTAL\_smart} and \textsc{LearnDOTA}, respectively. The learnt $50$ DOTAs are divided into $5$ groups according to the
% number of locations, number of actions, and maximum clock value in guards. Table~\ref{tab:results} summarizes the experimental results.

\noindent\textbf{Experiment on DTAs.} We first compared the performance of our learning algorithm \textsc{DTAL\_powerful} with the algorithm \textsc{LearnTA} proposed in~\cite{waga2023active} on learning the DTA ``unbalanced:2" in~\cite{waga2023active} as an example. The target automaton contains $5$ locations, $1$ action and $5$ transitions. Our algorithm \textsc{DTAL\_powerful} yields a timed automaton with only $8$ locations (including the ``sink" location) and $2$ clocks, whereas \textsc{LearnTA} yields a timed automaton with $25$ locations and $7$ clocks. It gives an evidence that our algorithm can return a much smaller DTA with fewer number of clocks. Since when learning a DTA by using the algorithm \textsc{DTAL\_normal}, the final table will never be larger than the final table in the powerful teacher situation. Therefore, the experimental result gives us highly confidence that \textsc{DTAL\_normal} will return a much smaller automaton than \textsc{LearnTA} once it terminates.

\begin{table}\centering
  \caption{Experimental results on learning RTAs. \#Reset, \#Membership, \#Equivalence represent the number of reset information queries, membership queries and equivalence queries needed during the learning process, respectively. $n$ is the number of locations of the learned automaton.} 
  \label{tab:RTAresults}
  \resizebox{1\columnwidth}{!}{
  \begin{tabular}{ccccccc}
    \toprule
    Case ID & Algorithm & \#Reset & \#Membership &\#Equivalence &$n$ &time(s) \\
    \midrule
    \multirow{2}{*}{4-4-4} & \textsc{DTAL\_powerful} & $228.0$ & $267.6$ & $28.7$ & $5$ &$2.9$ \\ 
                           &        \textsc{LearnRTA} & $-$ & $298.4$&$28.3$ &$5$& $4.5$\\
    \midrule
    \multirow{2}{*}{6-4-4} & \textsc{DTAL\_powerful} & $529.3$ &$578.8$ &$46.8$ & $7$ &$16.4$\\ 
                   &       \textsc{LearnRTA} & $-$ & $696.0$ & $45.6$ & $7$ & $37.1$\\
    \midrule
    \multirow{2}{*}{7-2-4} & \textsc{DTAL\_powerful} & $388.2$ & $383.3$ & $32.6$ & $8$ & $3.5$\\ 
                   &    \textsc{LearnRTA} & $-$ & $554.0$ & $31.2$ & 8 & $9.0$\\
    \midrule
    \multirow{2}{*}{7-3-4} &\textsc{DTAL\_powerful} & $458.2$ & $456.9$ & $37.6$ & $8$ & $6.5$\\ 
                   &        \textsc{LearnRTA} & $-$ & $630.4$ & $36.7$ & 8 & $17.0$\\
    \midrule
    \multirow{2}{*}{7-4-2}& \textsc{DTAL\_powerful} & $469.9$ & $436.9$ & $23.6$ & $8$ & $2.9$\\ 
                   &    \textsc{LearnRTA} & $-$ & $630.0$ & $22.8$ & 8 & $11.0$\\
    \midrule
    \multirow{2}{*}{7-4-3}& \textsc{DTAL\_powerful} & $470.9$ & $496.6$ & $32.0$ & $8$ & $5.8$\\ 
                   &    \textsc{LearnRTA} &$-$ &$680.0$ & $31.4$ & 8 &$16.1$\\
    \midrule
    \multirow{2}{*}{7-4-4}& \textsc{DTAL\_powerful} & $529.4$ & $571.5$ & $46.6$ & $8$ & $14.0$\\ 
                   &    \textsc{LearnRTA} & $-$ & $793.6$ & $45.2$ & 8 & $34.7$\\
    \midrule
    \multirow{2}{*}{7-4-5}& \textsc{DTAL\_powerful}&$613.6$ &$639.3$ &$50.0$ & $8$ & $19.7$\\ 
                   &    \textsc{LearnRTA}  & $-$ & $831.3$ & $48.5$ & 8 & $46.0$ \\
    \midrule
    \multirow{2}{*}{7-4-6}& \textsc{DTAL\_powerful} & $907.8$ & $986.4$ &$85.2$ & $8$&$110.3$\\ 
                   &    \textsc{LearnRTA} & $-$ &$1149.9$ & $83.4$ & 8 & $142.5$ \\
    \midrule
    \multirow{2}{*}{7-4-7}& \textsc{DTAL\_powerful} & $1029.8$ & $1143.1$ & $95.0$ & $8$&$176.1$\\ 
                   &    \textsc{LearnRTA}  &$-$ &$1304.3$ & $93.1$& 8 &$219.9$\\
    \midrule
    \multirow{2}{*}{8-4-4}& \textsc{DTAL\_powerful} & $804.0$ & $841.4$ &$56.3$ & $9$&$33.9$\\ 
                   &    \textsc{LearnTA}  &$-$ &$1156.5$ & $54.3$& 9 &$55.7$\\
    \midrule
    \multirow{2}{*}{10-4-4}& \textsc{DTAL\_powerful}&$1268.2$& $1291.1$ & $70.2$ & $11$ & $94.1$ \\ 
                   &    \textsc{LearnTA}  &$-$ &$1843.3$ & $67.9$& 11 &$167.7$\\
%                        \midrule
%                    \multirow{2}{*}{12-4-4}& \textsc{DTAL\_smart}&
% $1761.6$ &$1735.2$
% & $82.6$& &$148.8$\\ 
%                    &    \textsc{LearnTA}  &$-$ &$2476.8$ & $79.9$& &$280$\\
  \bottomrule
\end{tabular}}
\vspace{-0.2cm}
\end{table}

\begin{figure*}[!t]
\begin{center}
\begin{minipage}[b]{0.35\textwidth}
\begin{center}
\resizebox{!}{0.9\textwidth}{
\begin{tikzpicture}[scale=0.65, ->, >=stealth', shorten >=1pt, auto, node distance=2cm, semithick, every node/.style={scale=0.65}]
% \centering
        \node[initial,accepting, state]  (1) {$l_1$};
          \node[accepting, state](2) at (3,0) {$l_2$};
        \node[accepting,state](3) at (6,0) {$l_3$};
         \node[accepting, state](4) at (6,-6) {$l_4$};
          \node[accepting, state](5) at (0,-6) {$l_5$};
          
        \path  
    
        (1) edge[ color=black] node[ align=center] {$a,$ \\ $ 1<c_{1}<2,$\\ $\{c_1\}$} (2)

        (2) edge[color=black] node[ align=center] {$a,$ \\ $ 1<c_{1}<2\wedge 2<c_{2}<4,$\\ $\{c_1,c_2\}$} (3)
        
        (3) edge[ color=black] node[left, align=center] {$a,$ \\ $ 1<c_{1}<2,$\\ $\{c_1\}$} (4)
        
        (4) edge[color=black] node[ align=center] {$a,$ \\ $ 1<c_{1}<2\wedge 2<c_{2}<4,$\\ $\{c_1,c_2\}$} (5)

        (5) edge[ color=black] node[ align=center] {$a,$ \\ $ 1<c_{1}<2\wedge c_{2}> 1,$\\ $\{c_1\}$} (1);
\end{tikzpicture}
}
\end{center}
\end{minipage}
\begin{minipage}[b]{0.49\textwidth}
\begin{center}
\resizebox{!}{0.8\textwidth}{
\begin{tikzpicture}[scale=0.65, ->, >=stealth', shorten >=1pt, auto, node distance=2cm, semithick, every node/.style={scale=0.65}]
% \centering
        \node[initial,accepting, state]  (1) {$l_1$};
        % \node[state](2) at (3,0) {$l_2$};
         \node[accepting, state](3) at (3,0) {$l_2$};
          \node[accepting, state](4) at (6,-6) {$l_6$};
          \node[accepting, state](5) at (9,-6) {$l_5$};
        \node[accepting,state](6) at (3,-6) {$l_7$};
         \node[accepting, state](7) at (9,0) {$l_4$};
          \node[accepting, state](8) at (6,0) {$l_3$};
          
        \path  
        
        % (1) edge[color=black] node[ black, align=center] {$a,$ \\ $  c_{1} \geq 2\wedge c_{2} \geq 2,$ \\ $\{c_1,c_2\}$} (2)

        % (1) edge[in=270,out=270, color=black] node[ black, align=center] {$a,$ \\ $  c_{1} \leq 1\wedge c_{2} \geq 0,$ \\ $\{c_1,c_2\}$} (2)

        % (1) edge[in=90,out=90, color=black] node[ black, align=center] {$a,$ \\ $  c_{1} > 1\wedge c_{2} < 1,$ \\ $\{c_1,c_2\}$} (2)
        
        (1) edge[in=100,out=90, color=black] node[ align=center] {$a,$ \\ $ 1<c_{1}<2\wedge c_{2}> 1,$\\ $\{c_1\}$} (3)

        (1) edge[in=260,out=270, color=black] node[ align=center] {$a,$ \\ $ c_{1}\geq 2\wedge 1<c_{2}<2,$\\ $\{c_1\}$} (3)
        
        % (2) edge[loop right] node[,align=center] {$a,$ \\ $ c_{1}\geq 2\wedge 1<c_{2}<2,$\\ $\{c_1\}$} (2)
        
        (3) edge[in=260,out=270,color=black] node[ align=center] {$a,$ \\ $ 1<c_{1}<2\wedge c_{2}>2,$\\ $\{c_1,c_2\}$} (8)
        
        (3) edge[in=100,out=90, color=black] node[ align=center] {$a,$ \\ $ c_{1}\geq 2\wedge 2<c_{2}\leq 3,$\\ $\{c_1,c_2\}$} (8)
        
        (8) edge[in=90,out=90,color=black] node[ align=center] {$a,$ \\ $ c_{1}\geq 2\wedge 1<c_{2}<2,$\\ $\{c_1\}$} (7)

        (8) edge[in=270,out=280, color=black] node[ align=center] {$a,$ \\ $ 1<c_{1}<2\wedge c_{2}> 1,$\\ $\{c_1\}$} (7)
        
        (7) edge[in=50,out=310,color=black] node[ left,align=center] {$a,$ \\ $ (1<c_{1}<2\wedge c_{2}>2)\vee$\\$(c_{1}\geq 2\wedge 2<c_{2}\leq 3),$\\ $\{c_1,c_2\}$} (5)
        
        % (7) edge[color=black] node[right, align=center] {$a,$ \\ $ c_{1}\geq 2\wedge 2<c_{2}\leq 3,$\\ $\{c_1,c_2\}$} (5)
        
        (5) edge[in=90,out=90, color=black] node[above, align=center] {$a,$ \\ $ 1<c_{1}<2\wedge c_{2}> 1,$\\ $\{c_1\}$} (4)
        
        (5) edge[in=270,out=270,color=black] node[above, align=center] {$a,$ \\ $ c_{1}\geq 2\wedge 1<c_{2}<2,$\\ $\{c_1\}$} (4)
        
        (4) edge[in=270,out=270,color=black] node[ align=center] {$a,$ \\ $ 1<c_{1}<2\wedge c_{2}>2,$\\ $\{c_1\}$} (6)
        
        (4) edge[ color=black] node[ align=center] {$a,$ \\ $ c_{1}\geq 2\wedge 2<c_{2}\leq 3,$\\ $\{c_1\}$} (6)
        
        (6) edge[in=270, color=black] node[ align=center] {$a,$ \\ $ c_{1}> 1\wedge 3<c_{2}<4,$\\ $\{c_1,c_2\}$} (8);
\end{tikzpicture}
}
\end{center}
\end{minipage}
\end{center}
\caption{Target DTA unbalanced:2 in~\cite{waga2023active} (left) and the learnt DTA by \textsc{DTAL\_powerful} (right). The ``sink" location is omitted.}
\label{fig:unbalanced:2}
\end{figure*}
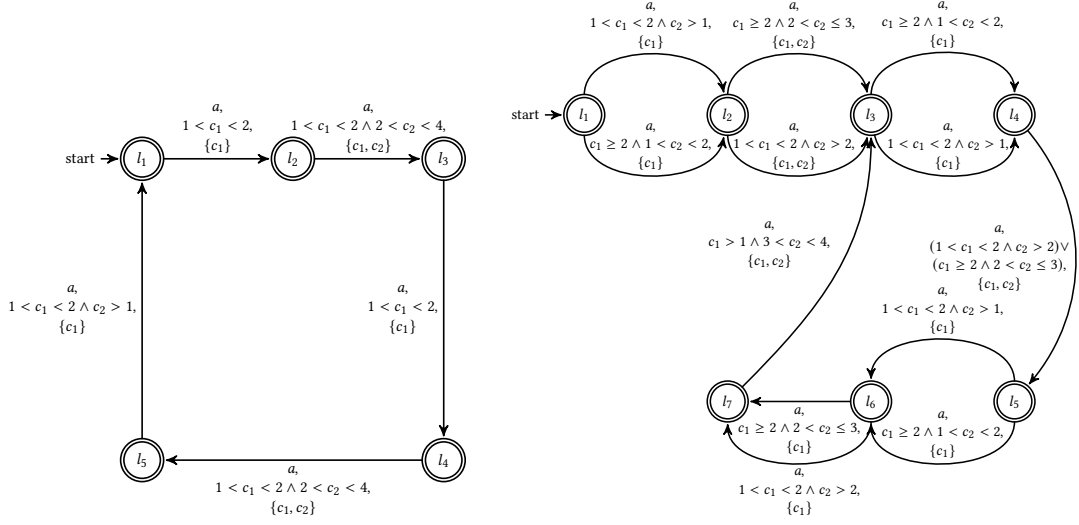

\subsection{Evaluation on \textsc{DTAL\_normal}}
\begin{table}\centering
  \caption{Experimental results on learning DTAs. Each row represents one DTA where $\lvert L\rvert$ is the number of locations, $\lvert \Sigma\rvert$ is the alphabet size, $\lvert \mathcal{C}\rvert$ is the number of clock variables, and $\kappa(\mathcal{C})$ is the maximum constant in the guards in the DTA. $n$ is the number of locations of the learned automaton. $-$ represents a timeout (1.5 hours).} 
  \label{tab:DTAresults}
  \resizebox{1\columnwidth}{!}{
  \begin{tabular}{ccccccccc}
    \toprule
    $\lvert L\rvert$&$\lvert \Sigma\rvert$&$\lvert \mathcal{C}\rvert$&$\kappa(\mathcal{C})$&Algorithm &\#Membership&\#Equivalence&$n$&time(s)\\
    \midrule
    \multirow{2}{*}{2} & \multirow{2}{*}{1}  & \multirow{2}{*}{2} & \multirow{2}{*}{3}&\textsc{DTAL\_normal} & $763145$ &$272734$ & $3$&$2440$\\ 
                   &      &     &     &    \textsc{LearnTA}&$1189$ &$6$ & $3$&$0.035$\\
    \midrule
    \multirow{2}{*}{2} & \multirow{2}{*}{1} & \multirow{2}{*}{2} & \multirow{2}{*}{30} & \textsc{DTAL\_normal}&$891513$ &$320108$ & $3$&$5093$\\ 
                   &      &     &     &  \textsc{LearnTA}  & $793983$&$33$ & $29$&$631$\\
    \midrule
    \multirow{2}{*}{3} & \multirow{2}{*}{1} & \multirow{2}{*}{2} & \multirow{2}{*}{3}& \textsc{DTAL\_normal}&$1020044$ &$102024$ & $3$ &$1269$\\ 
                   &      &     &     &      \textsc{LearnTA}& $1251$&$7$ & $4$&$0.009$\\
    \midrule
    \multirow{2}{*}{3} & \multirow{2}{*}{1} & \multirow{2}{*}{2} & \multirow{2}{*}{30}& \textsc{DTAL\_normal}& $914554$&$122622$ & $3$ & $2588 $\\ 
                   &      &     &     &  \textsc{LearnTA} & $1874410$&$49$ & $39$&$475$\\
    \midrule
    \multirow{2}{*}{4} & \multirow{2}{*}{1}  & \multirow{2}{*}{2} & \multirow{2}{*}{50}&\textsc{DTAL\_normal} & $-$ &$-$ & $-$&$-$\\ 
                   &      &     &     &    \textsc{LearnTA}&$-$ &$-$ & $-$&$-$\\
  \bottomrule
\end{tabular}}
\vspace{-0.2cm}
\end{table}

Table~\ref{tab:DTAresults} summarizes the experimental results. It shows that, as long as the maximum constant $\kappa(\mathcal{C})$ increases, the number of locations in the timed automata learned by \textsc{LearnTA} increases dramatically. This may be due to the fact that in the timed observation table of \textsc{LearnTA}, prefixes consist of a series of simple elementary languages. Hence, as the value of $\kappa(\mathcal{C})$ increases, the number of prefixes in the timed observation table increases significantly, thus affecting the number of locations of the hypothesis. Besides, in \textsc{LearnTA}, the number
of the membership queries blows up exponentially to the value of $\kappa(\mathcal{C})$. This is because that when learning DTA from a normal teacher, symbolic membership queries is forbidden, a symbolic membership query needs to be reduced to finitely many membership queries. However, the size of this decomposition exponentially blows up with respect to $\kappa(\mathcal{C})$. Similarly, in \textsc{LearnTA}, the number of equivalence queries also grows as $\kappa(\mathcal{C})$ increases.

The experimental results show that \textsc{DTAL} in the normal teacher situation requires more membership queries and equivalence queries. This is because when guessing the reset information of the target DTA, a large number of observation tables are generated, and more membership queries and equivalence queries are needed to maintain these observation tables. But it can return a much smaller automaton, which benefits from the equivalence relation presented in Section~\ref{sec:myhill}. The results also show that both algorithms, i.e., \textsc{DTAL} and \textsc{LearnTA}, need improvements before being applied in practice.

\section{Related work}\label{sec:relatedwork}

\noindent\textbf{Active learning.} In~\cite{GrinchteinJP06,GrinchteinJL10}, Grinchtein et al. propose learning algorithms for deterministic event-recording automata (ERA)~\cite{AlurFH99}. ERA is a kind of timed automata with several restrictions on clocks such that for every action $a$, a clock $x_a$ is associated to record the time length from the last occurrence of $a$ to the current. Recently, in~\cite{HenryJM20}, Henry et al. consider the learning problem of reset-free ERA in which some transitions may reset no clocks. In~\cite{AnCZZZ20}, An et al. provide an active learning method for deterministic one-clock timed automata (DOTA). After that, several improvements have been proposed, e.g., using PAC learning~\cite{ShenAZZ0Z20}, mutation testing~\cite{TangSZAZZ22}, SMT solving~\cite{XuAZ22}, etc. Real-time automata can be viewed as a kind of DOTA such that the single clock resets at every transition. The efficient learning algorithms have been designed in both the deterministic case~\cite{AnWZZZ21} and the nondeterministic case~\cite{AnZZZ21}. Instead of using clocks, the ``timer'' is another kind of model to record the time length. It can be assigned a value, and if it counts down to zero, a visible timeout signal will be triggered. In~\cite{CaldwellCF16}, Caldwell et al. proposed an algorithm for learning the Mealy machine with timers from programmable logic controllers. In~\cite{VaandragerB021}, Vaandrager et al. presented an efficient learning algorithm for such models with one timer. Recently, an active learning algorithm of DTA with Myhill-Nerode style characterization is proposed in~\cite{waga2023active}. However, with a large maximum constant in the guards in the target DTA, the number of queries required during learning and the number of locations in the learned DTA grows significantly in this algorithm. Our algorithm can return a much smaller DTA.

\noindent\textbf{Passive learning.} In~\cite{Verwer07,VerwerWW12}, an algorithm is proposed to learn deterministic RTA~\cite{Dima01} from labelled time-stamped event sequences. A passive learning algorithm for timed automata with one clock is further
proposed in~\cite{VerwerWW09,VerwerWW11}. In~\cite{TapplerALL19}, Tappler et al. present a passive learning algorithm for deterministic timed automata based on genetic programming (GP). In further, by introducing a testing process on the system under test (SUT), Aichernig et al. add an active manner in the GP-based learning procedure~\cite{AichernigPT20}. Recently, in~\cite{TapplerAL22}, Tappler et al. present an SMT-based passive learning algorithm for deterministic timed automata. There are also some works on passive learning of timed automata in practical cases~\cite{PastoreMM17,CornanguerLRT22} and dynamical systems~\cite{JinAZZZ21}. 
As mentioned before, passive learning algorithms cannot guarantee the correctness of the learned model.

\section{Conclusion and Future work}\label{sec:conclusion}
We propose an active learning approach for deterministic multi-clock timed automata. We transform the problem into learning the reset-clocked language of the target automaton. We first introduce an equivalence relation for reset-clocked languages. Based on it, we present a learning algorithm for deterministic timed automata in the powerful teacher situation by designing the membership query and the reset information query over region words. In further, we drop the reset information query and present the active learning algorithm in the normal teacher situation where the learner can only make standard membership queries and equivalence queries. 

In future work, we intend to improve the algorithm by using the technique reported in~\cite{XuAZ22} so that it can be used in learning timed models with large sizes of practical systems. 

%\section*{Acknowledgement}
\begin{acks}
We thank the anonymous reviewers for their valuable feedback and suggestions for improving the quality of this paper. Y. Teng and M. Zhang are supported by Shanghai 2023 ``Science and Technology Innovation Action Plan": Special Project for Key Technical Breakthrough of Blockchain (No. 23511100800), and the National Natural Science Foundation of China (No. 61972284, 62032019). J. An is supported by ERATO HASUO Metamathematics for Systems Design Project (No. JPMJER1603), JST.
\end{acks}
\bibliographystyle{ACM-Reference-Format}
\bibliography{ref}

\newpage
\appendix

\section{Proofs of Lemmas and Theorems}
\label{appendix:proofs}
\begin{proof}[Proof of Lemma~\ref{lemma:logictoresetlogic}]
    The two conclusions are based on the fact that $\gamma$ and $\gamma'$ witness the same transition sequence of $\mathcal{A}$. We show the fact and prove the two conclusions as follows.
Since $\llbracket \gamma \rrbracket = \llbracket \gamma' \rrbracket$, suppose $\gamma=(\sigma_1, \mathbf{v}_1)(\sigma_2,\mathbf{v}_2)\cdots(\sigma_n,\mathbf{v}_n)$ and $\gamma'=(\sigma_1, \mathbf{v}_1')(\sigma_2,\mathbf{v}_2')\cdots(\sigma_n,\mathbf{v}_n')$, we have $\llbracket\mathbf{v}_i \rrbracket = \llbracket\mathbf{v}_i'\rrbracket$ for all $1\leq i\leq n$. Depending on the definitions of clock constraint and region, given a transition guard $\phi_i$, we have $\llbracket\mathbf{v}_i \rrbracket\subseteq \phi_i$ iff $\llbracket\mathbf{v}_i' \rrbracket\subseteq \phi_i$. What's more, since $\mathcal{A}$ is deterministic and complete, for every clocked action $(\sigma_i,\mathbf{v}_i)$, it triggers a unique transition $(l_{i-1},\sigma_i,\phi_i,\mathcal{B}_i,l_i)$ which is the same as the one triggered by $(\sigma_i,\mathbf{v}_i')$ for all $1\leq i \leq n$. Hence, $\gamma$ and $\gamma'$ witness the same transition sequence of $\mathcal{A}$, and thus the reset information in each step of $\gamma_{r}$ and $\gamma_{r}'$ are the same.
So, the clock valuations after passing $\gamma$ and $\gamma'$ in $\mathcal{A}$ also must belong to the same region. Hence, $\gamma$ and $\gamma'$ reach the same symbolic state. Therefore, the two conclusions hold.
\end{proof}

\begin{proof}[Proof of Lemma~\ref{lemma:alltssame}]
   According to the definition, $\gamma_{r}\gamma_r'$ and $\gamma_{r}\gamma_r''$ are two valid reset-clocked words of $\mathcal{A}$ and $\llbracket vw(\gamma_{r}\gamma_r')\rrbracket=\llbracket vw(\gamma_{r}\gamma_r'')\rrbracket$. 
% Therefore, we have the conclusion based on the Lemma~\ref{lemma:logictoresetlogic}.
Therefore, based on Lemma~\ref{lemma:logictoresetlogic}, $vw(\gamma_{r}\gamma_r')$ and $vw(\gamma_{r}\gamma_r'')$ witness the same transition sequence of $\mathcal{A}$. Thus, $\gamma_{r}\gamma_r$ and $\gamma_{r}\gamma_r''$ reach the same location. According to the lemma, we also have $resets(\gamma_r')=resets(\gamma_r'')$. Hence, the reset information in each step of $\gamma_{r}\gamma_r'$ and $\gamma_{r}\gamma_r''$ are the same. So, the clock valuations after passing $\gamma_{r}\gamma_r'$ and $\gamma_{r}\gamma_r''$ in $\mathcal{A}$ also must belong to the same region. Therefore, $\gamma_{r}\gamma_r'$ and $\gamma_{r}\gamma_r''$ reach the same symbolic state in $\mathcal{A}$.
\end{proof}

\begin{proof}[Proof of Lemma~\ref{lemma:samesymbolicstate}]
    Depending on the definition of the symbolic states, the clock valuations after running $\gamma_{r1}$ and $\gamma_{r2}$ on $\mathcal{A}$ must belong to the same region. Because by letting time pass from any two points in the same region, the next visited region is the same~\cite{MalerP04}, there are only two cases about the valid successors of $\gamma_{r1}$ and $\gamma_{r2}$ according to $\xi$. The first case is $\textit{vs}_{\mathcal{A}}(\gamma_{r1},\xi)=\textit{vs}_{\mathcal{A}}(\gamma_{r2},\xi)=\emptyset$ and the second case is $\textit{vs}_{\mathcal{A}}(\gamma_{r1},\xi)\neq\emptyset$ and $\textit{vs}_{\mathcal{A}}(\gamma_{r2},\xi)\neq\emptyset$. 

It is easy to show the correctness of the first case. Consider the second case, according to Definition~\ref{def:validsuccessor}, $\llbracket vw(\gamma_{r1}') \rrbracket = \llbracket vw(\gamma_{r2}') \rrbracket=\xi$. 
According to Lemma~\ref{lemma:logictoresetlogic}, $\gamma_{r1}'$ and $\gamma_{r2}'$ witness the same transition sequence reaching the same location $l'$ from some location $l$ of $\mathcal{A}$. What's more, the reset information over the transition sequences of $\gamma_{r1}'$ and $\gamma_{r2}'$ from $l$ is equal since the transition sequences are the same, i.e. $resets(\gamma_{r1}')=resets(\gamma_{r2}')$. Hence, the reset information at each step of $\gamma_{r1}'$ and $\gamma_{r2}'$ from $l$ are the same. Therefore, the clock valuations after passing $\gamma_{r1}\gamma_{r1}'$ and $\gamma_{r2}\gamma_{r2}'$ in $\mathcal{A}$ also must belong to the same region. So $\gamma_{r1}\gamma_{r1}'$ and $\gamma_{r2}\gamma_{r2}'$ also reach another same symbolic state in $\mathcal{A}$. 
\end{proof}

\begin{proof}[Proof of Lemma~\ref{lemma:finitevalid}]
    Since two valid reset-clocked words $\gamma_{r1}$ and $\gamma_{r2}$ reach the same symbolic state of $\mathcal{A}$, according to Lemma~\ref{lemma:samesymbolicstate}, given a $\xi\in\bm{\Sigma_{G}}^*$, for all $\gamma_{r1}'\in \textit{vs}_{\mathcal{A}}(\gamma_{r1},\xi)$ and $\gamma_{r2}' \in \textit{vs}_{\mathcal{A}}(\gamma_{r2},\xi)$, we have $\gamma_{r1}\gamma_{r1}'$ and $\gamma_{r2}\gamma_{r2}'$ also reach some same symbolic state. Hence, $\gamma_{r1}\gamma_{r1}'\in\mathscr{L}_{r}(\mathcal{A})$ iff $\gamma_{r2}\gamma_{r2}'\in\mathscr{L}_{r}(\mathcal{A})$. Besides, Lemma~\ref{lemma:samesymbolicstate} also shows that $resets(\gamma_{r1}')=resets(\gamma_{r2}')$. According to Definition~\ref{def:equivalencerelation}, we have $\gamma_{r1}\sim_{\mathscr{L}_{r}(\mathcal{A})}\gamma_{r2}$. 
\end{proof}

\begin{proof}[Proof of Lemma~\ref{lemma:finiteinvalid}]
    If $\gamma_{r1}$ and $\gamma_{r2}$ are two invalid reset-clocked words of $\mathcal{A}$, then for all $\xi\in\bm{\Sigma_{G}}^*$, we have $\textit{vs}_{\mathcal{A}}(\gamma_{r1},\xi)=\textit{vs}_{\mathcal{A}}(\gamma_{r2},\xi)=\emptyset$. Then, according to Definition~\ref{def:equivalencerelation}, $\gamma_{r1}\sim_{\mathscr{L}_{r}(\mathcal{A})}\gamma_{r2}$. 
\end{proof}

\begin{proof}[Proof of Theorem~\ref{thm:equivalenceclass}]
    Depending on the definitions of the region and symbolic state, each valid reset-clocked word $\gamma_r$ of $\mathcal{A}$ reaches a symbolic state. In Section~\ref{sbsc:timedautomata}, we know that the number of symbolic states of $\mathcal{A}$ is bounded by $|L|\cdot |\mathcal{C}|!\cdot 2^{|\mathcal{C}|} \cdot\prod_{c\in\mathcal{C}}(2\kappa(c)+2)$. Lemma~\ref{lemma:finitevalid} shows that all valid reset-clocked words form at most as many equivalence classes as the number of symbolic states. And Lemma~\ref{lemma:finiteinvalid} tells that all invalid reset-clocked words of $\mathcal{A}$ belong to the same equivalence class. Hence we have the conclusion that $\sim_{\mathscr{L}_{r}(\mathcal{A})}$ has a finite number of equivalence classes. 
\end{proof}

\begin{proof}[Proof of Theorem~\ref{theorem:basicidea}]
    By the transforming method of the reset-clocked word and the reset-delay-timed word, $\mathscr{L}_{r}(\mathcal{A}_1)=\mathscr{L}_{r}(\mathcal{A}_2)$ iff $\mathcal{L}_{r}(\mathcal{A}_1)=\mathcal{L}_{r}(\mathcal{A}_2)$. By the definitions of the delay-timed word and the reset-delay-timed word,  we conclude that $\mathcal{L}(\mathcal{A}_1)=\mathcal{L}(\mathcal{A}_2)$ if $\mathcal{L}_{r}(\mathcal{A}_1)=\mathcal{L}_{r}(\mathcal{A}_2)$.  Hence, if $\mathscr{L}_{r}(\mathcal{A}_1)=\mathscr{L}_{r}(\mathcal{A}_2)$, then $\mathcal{L}(\mathcal{A}_1)=\mathcal{L}(\mathcal{A}_2)$.
\end{proof}

\begin{proof}[Proof of Lemma~\ref{lemma:tabledfa}]
    First, consider the case that $\gamma_{r}\in\bm{S}$. According to the definition of $ f(\gamma_{r}, e)$, if $f(\gamma_{r}, e)=+$, then $vs_{\mathcal{A}}(\gamma_{r}, e)\neq\emptyset$. Because the observation table $\mathbf{T}$ is evidence-closed, there must be $\gamma_r'\in vs_{\mathcal{A}}(\gamma_{r}, e)$ and $\gamma_{r}\gamma_r'\in \bm{S}\cup\bm{R}$. According to the definition of observation table, $g(\gamma_r,e)=resets(\gamma_r')$. Since $f(\gamma_{r}, e)=+$, then $f(\gamma_r\gamma_r', \epsilon)=+$. According to the construction of $\text{M}$, $\gamma_{r}\gamma_r'$ ends in location $l_{\mathit{row}(\gamma_{r}\gamma_r')}\in F_{M}$, namely the constructed DFA $\text{M}$ accepts $\gamma_{r}\gamma_r'$. When $f(\gamma_{r}, e)=-$ and $vs_{\mathcal{A}}(\gamma_{r}, e)\neq\emptyset$, a similar analysis can be performed. If $f(\gamma_{r}, e)=-$ and $vs_{\mathcal{A}}(\gamma_{r}, e)=\emptyset$, it means that there is no valid successor. Suppose that $\text{M}$ has an accepting word in the form of $\gamma_r\gamma_r'$ such that $\llbracket vw(\gamma_r')\rrbracket = e$. According to the construction of $\text{M}$, there must be some word $\gamma_r''$ such that $\gamma_r''\gamma_r'\in\bm{S}\cup\bm{R}$ with $\mathit{row}(\gamma_r'')=\mathit{row}(\gamma_r)$,  $\gamma_r'\in vs_{\mathcal{A}}(\gamma_r'',e)\neq\emptyset$, and $f(\gamma_r'',e)=+$. It is impossible that $f(\gamma_r,e)\neq f(\gamma_r'',e)$ while $\mathit{row}(\gamma_r'')=\mathit{row}(\gamma_r)$. Thus, DFA $\text{M}$ has no accepting word in such a form.

For the case that $\gamma_{r}\in\bm{R}$, there exists $\gamma_r'''\in\bm{S}$ such that $\mathit{row}(\gamma_{r})= \mathit{row}(\gamma_r''')$ since $\mathbf{T}$ is closed, which implies that $ f(\gamma_{r} ,e)= f(\gamma_r''',e)$ and $ g(\gamma_{r} ,e)= g(\gamma_r''',e)$ for all $ e\in\bm{E}$. What's more, $\mathit{row}(\gamma_{r})= \mathit{row}(\gamma_r''')$ indicates that $\gamma_r'''$ and $\gamma_r$ reach the same location in $\text{M}$. Thus, after appending any suffixes, the behaviours must be the same. Thus, it is reduced to the first case.
\end{proof}

\begin{proof}[Proof of Lemma~\ref{lemma:partitiondeterministic}]
 Suppose that $i<j$, according to the fourth step, $I_i=W_i+A_i$ and $I_j=W_j+A_j$. Now we  proof that $I_i\cap I_j=\emptyset$. According to the third step, we have $W_i\cap W_j=\emptyset$.  In the first step, if $A_i$ or $A_j$ equals $\emptyset$, $A_i\cap A_j=\emptyset$. Otherwise, if neither $A_i$ nor $A_j$ equals $\emptyset$, $A_i$ and $A_j$ must belong to different regions, which also leads to the fact that $A_i\cap A_j=\emptyset$. Besides, since $A_i, A_j\in U_0$, according to the third step, $W_i\cap A_j=\emptyset$ and $W_j\cap A_i=\emptyset$. In addition, the last step does not influence $I_i\cap I_j$. Therefore, $I_i\cap I_j=\emptyset$.
   % It is straightforward according to the steps in the definition of $P$. 
\end{proof}

\begin{proof}[Proof of Theorem~\ref{theorem:partitioncomplete}]
According to the fourth step, $I_{1}\cup I_{2}\cup\cdots\cup I_{n}=(W_1\cup W_2\cup\cdots\cup W_n)\cup (A_1\cup A_2\cup\cdots\cup A_n)=(W_1\cup W_2\cup\cdots\cup W_n)\cup U_0$. 
According to the definition of $U_{i}$, $U_{1}$ is a constraint as a conjunction of $|\mathcal{C}|$ constraints $ic_{1,j}$ which are defined according to $\mathbf{v}_{1}=\{0\}^{|\mathcal{C}|}$.
Therefore, $U_{1}$ is equal to $\mathbb{R}_{\geq0}^{\lvert \mathcal{C} \rvert}$. Hence, according to the third step, $W_{1}\cup W_{2}\cup\cdots\cup W_{n}=(U_{1}-(U_{1}\cap(W_{2}\cup\cdots\cup  W_{n}\cup U_0))\cup(W_{2}\cup\cdots\cup W_{n})=U_{1}-(U_1\cap U_0)=U_1-U_0$. Hence, $I_{1}\cup I_{2}\cup\cdots\cup I_{n}=(U_1-U_0)\cup U_0=U_1=\mathbb{R}_{\geq0}^{\lvert \mathcal{C} \rvert}$. Lemma~\ref{lemma:partitiondeterministic} tells that $I_{1},\cdots,I_{n}$ are disjoint. Hence, it is concluded that $P(\cdot)$ returns a partition on $\mathbb{R}_{\geq0}^{\lvert \mathcal{C} \rvert}$.
\end{proof}

\begin{proof}[Proof of Lemma~\ref{lemma:dfatocta}]
    According to the definition of the partition function, for each transition $(l,(\sigma,\mathbf{v},\mathbf{b}),l')\in\Delta_{M}$ in $\text{M}$, there is a transition $(l, \sigma, I, \mathcal{B}, l')$ in the hypothesis $\mathcal{H}$, where $\mathbf{v}\in I\in P(\Psi_{l,\sigma})$ and $\mathcal{B}$ is the set of clocks $c_i=\top$ in $\mathbf{b}$. Hence, for every $\gamma_r\gamma_r'$, $\mathcal{H}$ accepts it iff $M$ accepts it. According to Lemma~\ref{lemma:tabledfa}, $M$ has the three properties. Thus, $\mathcal{H}$ also keep these. 
\end{proof}

\begin{proof}[Proof of Theorem~\ref{theorem:cta}]
    First, for any $l\in L_{M}$ and $\sigma\in\Sigma$, the partition function $P(\Psi_{l,\sigma})$ maps  $\Psi_{l,\sigma}$ to $n$ districts, $I_{1}, \cdots, I_{n}$, where $n=|\Psi_{l,\sigma}|$. According to Lemma~\ref{lemma:partitiondeterministic}, $I_{i}\cap I_{j}=\emptyset$, $i\neq j$ for all $1\leq i,j\leq n$. Thus, for each pair of transitions of the form $(l,\sigma,\phi_{1},-,-)$ and 
$(l,\sigma,\phi_{2},-,-)$, the clock constraints $\phi_{1}$ and $\phi_{2}$ are mutually exclusive. Therefore, $\mathcal{H}$ is a DTA. 
Second, to make the observation table closed, when the learner moves an element $\emph{r}\in\bm{R}$ to $\bm{S}$, a reset-clocked word $\pi(vw(r)\cdot \bm{\sigma})$ is added to $\bm{R}$, where $\bm{\sigma}=(\sigma,0^{\lvert \mathcal{C} \rvert})$ for all $\sigma\in\Sigma$. 
% Thus, the condition $\mathbf{v}_1=\{0\}^{|\mathcal{C}|}$ is satisfied. 
Thus, the condition $\mathbf{v}_1=\{0\}^{|\mathcal{C}|}\in\Psi_{l,\sigma}$ is satisfied. 
Therefore, $P(\Psi_{l,\sigma})$ returns
a partition on $\mathbb{R}_{\geq0}^{\lvert \mathcal{C} \rvert}$ according to Theorem~\ref{theorem:partitioncomplete}. Hence, $\mathcal{H}$ is a CTA.
\end{proof}

\begin{proof}[Proof of Theorem~\ref{theorem:7}]
    By Lemma~\ref{lemma:tabledfa}, Lemma~\ref{lemma:dfatocta} and Theorem~\ref{theorem:cta}, for every $\gamma_{r}\cdot e\in(\bm{S}\cup\bm{R})\cdot\bm{E}$, there is $\gamma_r''\in vs_{\mathcal{A}}(\gamma_{r},e)$ and there exists a unique accepting run $\rho$ that admits $\gamma_r\gamma_r''$ in $\mathcal{H}$ if $f(\gamma_{r}, e)=+$. 
Hence, every reset-clocked action $(\sigma_{i}, \mathbf{v}_i,\mathbf{b}_i)$ in $\gamma_r\gamma_r''$ triggers a transition 
$\delta_{i} = (l_{i-1}, \sigma_{i}, \phi_{i}, \mathcal{B}_i,$ $l_{i})$,
where $\mathbf{v}_i\in\phi_{i}$ and $\mathcal{B}_i=\{c_j \vert \mathbf{b}_{i,j}=\top\}$ for $1\leq i\leq n$ and $1\leq j \leq |\mathcal{C}|$. According to the definitions of the clock constraint of TA and the region, $\llbracket\mathbf{v}_i\rrbracket\subseteq\phi_{i}$. 
Therefore, for each $\llbracket\mathbf{v}_i'\rrbracket=\llbracket\mathbf{v}_i\rrbracket$, the reset-clocked action $(\sigma_{i}, \mathbf{v}_i',\mathbf{b}_i)$ also triggers the transition $\delta_{i}$ as long as it is valid. Hence, there exists a unique accepting run $\rho'$ that admits $\gamma_{r}'=(\sigma_{1},\mathbf{v}_1',\mathbf{b}_1)(\sigma_{2}, \mathbf{v}_2',\mathbf{b}_2) $ $\cdots
(\sigma_{n}, \mathbf{v}_n',\mathbf{b}_n)$ in $\mathcal{H}$. 
% and $\Pi_{\{m+2,m+3,\cdots,2m+1\}}\gamma_{r}''=\Pi_{\{2\}} f(\gamma_{r}\cdot e)$.
Suppose it is the case when $f(\gamma_{r}, e)=-$ and $vs_{\mathcal{A}}(\gamma_{r}, e)\neq\emptyset$, it is easy to follow the case $f(\gamma_{r}, e)=+$ by Lemma~\ref{lemma:dfatocta}. In the case where $f(\gamma_{r}, e)=-$ and $vs_{\mathcal{A}}(\gamma_{r}, e)=\emptyset$, there is no $\gamma_r''\in vs_{\mathcal{A}}(\gamma_{r},e)$ causing $\gamma_r'=\gamma_r\gamma_r''$ to be
valid. Lemma~\ref{lemma:tabledfa} and Lemma~\ref{lemma:dfatocta} guarantee that there are no such accepting words in both DFA $\text{M}$ and $\mathcal{H}$.
\end{proof}

\begin{proof}[Proof of Theorem~\ref{theorem:powerfulcorrectness}]
    First, the correctness of the algorithm is guaranteed by the equivalence query. Now we prove termination. 
    % According to the method of constructing $\mathcal{H}$ in Section~\ref{sbsc:membership} and Theorem~\ref{thm:equivalenceclass}, 
% Since we avoid the assumption that for all $c_j\in\mathcal{C}$ and $1\leq j\leq\lvert\mathcal{C}\rvert$, its corresponding $\kappa(c_j)$ in $\mathcal{A}$ needs to be known,  we let $\kappa_{ass}$ be a function mapping each clock $c\in\mathcal{C}$ to $k_{max}$ instead of mapping $c$ to the largest integer appearing in the guards over $c$,
% and $\Lambda=|\mathcal{C}|!\cdot 2^{|\mathcal{C}|}\cdot\prod_{c\in\mathcal{C}}(2\kappa_{ass}(c)+2)$ represent the bound of the number of clock regions. 
According to Section~\ref{sec:myhill} and Section~\ref{sbsc:membership}, we know that each reset-clocked word $\emph{s}\in\bm{S}$ actually corresponds to a symbolic state in $\mathcal{A}$. Because the observation table $\mathbf{T}$ is reduced, for any two reset-clocked words $s_{1}$ and $s_{2}$ in $\bm{S}$, there must exist $\emph{e}\in\bm{E}$ such that $f(s_{1}, \emph{e})\neq f(s_{2}, \emph{e})$ or $g(s_{1}, \emph{e})\neq g(s_{2}, \emph{e})$. So, $s_{1}$ and $s_{2}$ belong to different equivalence classes, i.e., $s_{1}$ and $s_{2}$ correspond to different symbolic states in $\mathcal{A}$. 
Therefore, the algorithm constructs an injection from $\bm{S}$ to the set of symbolic states of $\mathcal{A}$. 
So the size of $\bm{S}$ does not exceed 
% $|L|\cdot\Lambda$. 
$|L|\cdot |\mathcal{C}|!\cdot 2^{|\mathcal{C}|} \cdot\prod_{c\in\mathcal{C}}(2\kappa(c)+2)$ according to Theorem~\ref{thm:equivalenceclass}. 
        % {\color{red}So the size of $\bm{S}$ does not exceed the size of equivalence classes of $\mathcal{A}$, which is bounded according to Theorem~\ref{thm:equivalenceclass}. }

    % We consider that in the DFA $\text{M}$, $\exists l\in L_{M}, \sigma\in\Sigma$, there are  $(\mu_{i,1},\mu_{i,2},\cdots,\mu_{i,m}),$ $ (\mu_{i+1,1},\mu_{i+1,2},\cdots,\mu_{i+1,m})\in \Psi_{l,\sigma}$ such that $\{(\mu_{i,1},$ $\mu_{i,2},\cdots,\mu_{i,m})\}=\{(\mu_{i+1,1},\mu_{i+1,2},$ $\cdots,\mu_{i+1,m})\}$. 
% The algorithm adds elements to the table in three cases: processing of counterexamples, making $\mathbf{T}$ closed,  and making $\mathbf{T}$ evidence-closed. In these three cases, only the case making $\mathbf{T}$ closed may add invalid reset-clocked words in the table. The number of invalid words is finite and corresponds to the same equivalence class according to Lemma~\ref{lemma:finiteinvalid}. In the other two cases, all reset-clocked words added in the table are valid. 
% In this paper, whether $\mathbf{T}$ is consistent depends on two reset-clocked actions $\bm{\sigma_{r}}, \bm{\sigma_{r}}'$ satisfying $\left\{\Pi_{\{1,2,\cdots,m+1\}}\bm{\sigma_{r}}\right\}=\left\{\Pi_{\{1,2,\cdots,m+1\}}\bm{\sigma_{r}}'\right\}$. Therefore, 
Now let's consider the number of transitions. Given a prepared table, since the table is consistent, $\forall \gamma_{r}, \gamma_{r}' \in \bm{S}\cup\bm{R}$ such that $\mathit{row}(\gamma_{r})=\mathit{row}(\gamma_{r}')$ and $\forall \bm{\sigma_{r}},\bm{\sigma_{r}}'\in \bm{\Sigma_{r}}$ such that $\llbracket vw(\bm{\sigma_{r}})\rrbracket=\llbracket vw(\bm{\sigma_{r}}')\rrbracket$,  if $\gamma_{r}\bm{\sigma_{r}}, \gamma_{r}'\bm{\sigma_{r}}'\in \bm{S}\cup\bm{R}$, $\mathit{row}(\gamma_{r}\bm{\sigma_{r}})=\mathit{row}(\gamma_{r}'\bm{\sigma_{r}}')$ and $resets(\bm{\sigma_{r}})=resets(\bm{\sigma_{r}}')$.
So, for each $l\in L_{M}$ and $\sigma\in\Sigma$ in DFA $\text{M}$, if there are two transitions $(l,(\sigma,\mathbf{v},\mathbf{b}),l')$ and $(l,(\sigma,\mathbf{v}',\mathbf{b}')),l'')$ such that $\llbracket\mathbf{v}\rrbracket=\llbracket\mathbf{v}'\rrbracket$, $l'$ and $l''$ must be the same location and $\mathbf{b}=\mathbf{b}'$.
% According to the definition of the partition function, if there is more than one clock valuation in the same region, there is only one clock valuation will be mapped to a final timing constraint $I$ which is not $\emptyset$.
So, the number of transitions from a location triggered by an action in $\mathcal{H}$ is bounded by $ |\mathcal{C}|!\cdot 2^{|\mathcal{C}|} \cdot\prod_{c\in\mathcal{C}}(2\kappa(c)+2)$. Hence, the number of transitions for one location is $\lvert \Sigma\rvert\cdot |\mathcal{C}|!\cdot 2^{|\mathcal{C}|} \cdot\prod_{c\in\mathcal{C}}(2\kappa(c)+2)$. Since each iteration of the algorithm either
adds at least one element to $\bm{S}$ or refines at least one of the partitions along with transitions of $\mathcal{H}$, or both, the algorithm is guaranteed to terminate. 
\end{proof}

\begin{proof}[Proof of Theorem~\ref{theorem:normalterminate}]
%     The algorithm for learning a DTA $\mathcal{A}$ from a normal teacher can be viewed as a BFS searching on a finite multi-way tree, each of whose nodes is a filled table instance. The depth of a node is
% the number of guessed resets of its corresponding table instance. 
% Let $\mathbf{T}=(\Sigma, \bm{\Sigma_{G}}, \bm{\Sigma_{r}}, \bm{S}, \bm{R}, $ $\bm{E}, f, g)$ be the final table in the learning process with a smart teacher. The number of guesses in prefixes is $(|\bm{S}| + |\bm{R}|)\cdot |\mathcal{C}|$, as the table is prefix-closed.
% Since when a new prefix is added to $\bm{S}\cup\bm{R}$, there are at most $(\sum_{e_i \in \bm{E}\setminus\{\epsilon\}}{\lvert e_i \rvert}) \cdot \lvert \mathcal{C} \rvert$ guessed resets. 
% Hence, $\mathbf{T}$ can be found at depth
% at most $(|\bm{S}|+|\bm{R}|) \cdot \lvert \mathcal{C} \rvert + (|\bm{S}|+|\bm{R}|)\cdot (\sum_{e_i \in \bm{E}\setminus\{\epsilon\}}{\lvert e_i \rvert}) \cdot \lvert \mathcal{C} \rvert$ of the tree.
% Therefore, in the worst case, the learner can find $\mathbf{T}$ after checking all table instances before that depth and then finds $\mathbf{T}$.  Consequently, the algorithm terminates and returns a correct $\mathcal{H}$ which recognizes the target timed
% language.

Let $\mathbf{T}=(\Sigma, \bm{\Sigma_{G}}, \bm{\Sigma_{r}}, \bm{S}, \bm{R}, $ $\bm{E}, f, g)$ be the final table  in the learning process with a powerful teacher. The number of guesses in prefixes is $(|\bm{S}| + |\bm{R}|)\cdot |\mathcal{C}|$, as the table is prefix-closed.
Since when a new prefix is added to $\bm{S}\cup\bm{R}$, there are at most $(\sum_{e_i \in \bm{E}\setminus\{\epsilon\}}{\lvert e_i \rvert}) \cdot \lvert \mathcal{C} \rvert$ guessed resets. 
Since we use BFS to explore the tables in $ToExplore$ and BFS is based on a queue in terms of the size of the table, a smaller table is always handled before any larger ones. Hence, before finding the final table $\mathbf{T}$, all handled tables is smaller than $\mathbf{T}$. Since all table instances are produced by guessing reset information, the number of tables smaller than $\mathbf{T}$ is bounded by 
$\mathcal{O}(2^{(|\bm{S}|+|\bm{R}|) \cdot \lvert \mathcal{C} \rvert^2 \cdot(\sum_{e_i \in \bm{E}\setminus\{\epsilon\}}{\lvert e_i \rvert})})$. Therefore, in the worst case, the learner can find $\mathbf{T}$ after checking $\mathcal{O}(2^{(|\bm{S}|+|\bm{R}|) \cdot \lvert \mathcal{C} \rvert^2 \cdot(\sum_{e_i \in \bm{E}\setminus\{\epsilon\}}{\lvert e_i \rvert})})$ table instances.
\end{proof}

\section{Example of a transformation from DTA to CTA}
\label{appendix:transitionCTA}

Figure~\ref{fig:timedautomata} depicts the transformation of DTA $\mathcal{A}$ with $\Sigma=\{a,b\}$ and $\mathcal{C}=\{c_1,c_2\}$ on the left into the CTA on the right. First, a non-accepting ``sink” location $l_{2}$ is introduced. Second, we introduce three fresh transitions (marked in blue) from $l_{0}$ to $l_{2}$ as well as transitions from $l_{2}$ to itself. In the last step, five new transitions (marked in red) are introduced. For location $l_{0}$ and action $a$, the outgoing transition has a guard cover $c_{1}\textgreater 1\wedge c_{2}\textgreater 1$. Then we have the complement $\textit{Compl}_{l_0,a}=(0 \leq c_{1}\leq 1\wedge c_{2}\geq 0)\vee (c_{1} \textgreater 1\wedge 0 \leq c_{2}\leq 1)$. Hence, we introduce transitions $(l_{0}, a,0 \leq c_{1}\leq 1 \wedge c_{2}\geq 0,\{c_{1}, c_{2}\},l_{2})$ and $(l_{0}, a,c_{1} \textgreater 1\wedge 0 \leq c_{2}\leq 1,\{c_{1}, c_{2}\},l_{2})$. For location $l_1$ and actions $a,b$, three new transitions from $l_{1}$ to $l_{2}$  are constructed similarly.

\section{Example of finding valid successor}
\label{appendix:findvalidsucc}

\begin{figure}[H]
\begin{center}
\resizebox{0.49\textwidth}{!}{
\begin{tabular}{ P{5.5cm}|Y{0.5cm} Y{5cm}  }
 \hline
  &$\epsilon$&$(a,1< c_{1}< 2\wedge 1< c_{2}< 2\wedge c_{1}=c_{2})$\\
 \hline
 $\epsilon$&$-$&$+,\bot,\top$\\
  $(a,\{1.05,1.05\},\{\bot,\top\})$   & $+$    &$-,\top,\top$\\
   $(a,\{0,0\},\{\top,\top\})$   & $-$    &$-,\top,\top$\\
 \hline
  $(b,\{0,0\},\{\top,\top\})$   & $-$    &$-,\top,\top$\\
  $(a,\{1.05,1.05\},\{\bot,\top\})(a,\{0,0\},\{\top,\top\})$   & $-$    &$-,\top,\top$\\
  $(a,\{1.05,1.05\},\{\bot,\top\})(b,\{0,0\},\{\top,\top\})$   & $-$    &$-,\top,\top$\\
  $(a,\{1.05,1.05\},\{\bot,\top\})(b,\{1.05,0\},\{\bot,\top\})$   & $+$    &$-,\top,\top$\\
  $(a,\{1.05,1.05\},\{\bot,\top\})(b,\{2.05,1\},\{\top,\top\})$   & $-$    &$-,\top,\top$\\
  $(a,\{0,0\},\{\top,\top\})(a,\{1.05,1.05\},\{\top,\top\})$   & $-$    &$-,\top,\top$\\
  $(a,\{0,0\},\{\top,\top\})(a,\{0,0\},\{\top,\top\})$   & $-$    &$-,\top,\top$\\
  $(a,\{0,0\},\{\top,\top\})(b,\{0,0\},\{\top,\top\})$   & $-$    &$-,\top,\top$\\
 \hline
\end{tabular}
}
\end{center}
\caption{A copy of Fig.~\ref{fig:observationtable} for the ease of reading.}
\label{fig:observationtableinappendix}
\end{figure}

Figure~\ref{fig:observationtableinappendix} shows an example of a timed observation table. Suppose that the CTA $\mathcal{A}$ in Figure~\ref{fig:timedautomata} is the underlying target automaton. Consider $\gamma_{r}=(a,\{0,0\},\{\top,\top\})$ and $e=(a,1< c_{1}< 2\wedge 1< c_{2}< 2\wedge c_{1}=c_{2})$, $n=1$ for the length of $\gamma_r$ and there is only one region action in $e$. According to Algorithm~\ref{alg:findvalidsuccessor}, we can get $\mathbf{v}_{n+1,1}=0$ and $\mathbf{v}_{n+1,2}=0$ since $\mathbf{b}_{n,1}=\mathbf{b}_{n,2}=\top$. Consider the region $\llbracket \nu_1 \rrbracket= 1< c_{1}< 2\wedge 1< c_{2}< 2\wedge c_{1}=c_{2}$, we build a formula for $\mathbf{v}_{n+1} + d \in \llbracket \nu_1 \rrbracket$ such that $1 < 0 + d < 2 \wedge 1 < 0+d < 2 \wedge 0+d = 0 +d$. 
% A possible outcome by an SMT solver is, for example,  $d = 1.05$, thus we have $\mathbf{v}_{n+1}=\{1.05,1.05\}$. 
Then we find $d$ can be $1.05$, thus we have $\mathbf{v}_{n+1}=\{1.05,1.05\}$.
After a reset information query $\textsf{RQ}((a,\{0,0\})\cdot(a, \mathbf{v}_{n+1}))$, we get $\mathbf{b}_{n+1}=\{\top,\top\}$. Therefore, $e_{r}=(a,\{1.05,1.05\},\{\top,\top\})\in \textit{vs}_{\mathcal{A}}(\gamma_r, e)$. In this case, $f(\gamma_{r},e)=-$ because $\textsf{MQ}(\gamma_{r} e_{r})=-$, i.e. $\gamma_{r} e_{r}\notin\mathscr{L}_{r}(\mathcal{A})$, and $g(\gamma_{r},e)=resets(e_{r})=\{\top,\top\}$.

\section{Illustrative Case} 
\label{appendix:illustrative_case}

In order to describe our algorithm for learning DTA from a powerful teacher more illustratively, we present the process of learning DTA $\mathcal{A}$ in Figure~\ref{fig:timedautomata} where the initial location of $\mathcal{A}$ is $l_{0}$, the accepting location is $l_{1}$, and the alphabet $\Sigma=\{a,b\}$.
Figure~\ref{fig:case-tables} shows the iterations of the timed observation table during the learning process, and Figure~\ref{fig:case-automata} shows the iterations of the intermediate DFA $\mathrm{M}$ and hypothesis $\mathcal{H}$.

The first step in the learning algorithm is the initialization of the timed observation table. The learner obtains $\mathbf{T1}$ in Figure~\ref{fig:case-tables} after that. Since $\mathbf{T1}$ is prepared, the learner then constructs DFA $\mathrm{M_1}$ and  hypothesis $\mathcal{H}_1$ in further in Figure~\ref{fig:case-automata} based on $\mathbf{T1}$. By asking an equivalence query, the learner obtains a counterexample $\mathit{ctx_1}=(a,1.05,\{\bot,\top\})$, and then adds the corresponding reset-clocked word $(a,\{1.05,$ $1.05\},\{\bot,\top\})$ to the timed observation table. 

Since $\mathit{row}(a,\{1.05,1.05\},\{\bot,\top\})$ $=+$ and there is no row in $\bm{S}$ with the same value, the current table $\mathbf{T2}$ is not closed. The learner then makes the observation table closed by moving $(a,\{1.05,1.05\},$ $\{\bot,\top\})$ from $\bm{R}$ to $\bm{S}$ and gets $\mathbf{T3}$. Then, DFA $\mathrm{M}_3$ and hypothesis $\mathcal{H}_3$ in Figure~\ref{fig:case-automata} can be constructed. Similarly, the learner gets new counterexamples $(a,1.05,\{\bot,\top\})(b,0,\{\bot,\top\})$, $(a,1.05,\{\bot,\top\})(b,1,\{\top,$ $\top\})$ and  $(a,0,\{\top,\top\})(a,$ $1.05,\{\top,\top\})$ through the equivalence queries and constructs the tables $\mathbf{T4}$, $\mathbf{T5}$ and $\mathbf{T6}$, respectively. 

Table $\mathbf{T6}$ is not consistent since for the prefixes $\epsilon$ and $(a,\{0,0\},\{\top,$ $\top\})$, the reset information of $(a,\{1.05,1.05\})$ in $\epsilon\cdot(a,\{1.05,1.05\},\{\bot,$ $\top\})$ is not equal to that of $(a,\{1.05,1.05\})$ in $(a,\{0,0\},\{\top,\top\})\cdot(a,\{1.05,1.05\},\{\top,\top\})$, and $\mathit{row}((a,\{1.05,$ $1.05\},\{\bot,\top\}))\neq \mathit{row}$ $((a,\{0,0\},\{\top,\top\})\cdot(a,\{1.05,1.05\},\{\top,$ $\top\}))$, the learner adds $(a,1\textless c_{1}\textless 2\wedge 1\textless c_{2}\textless 2\wedge c_{1}=c_{2})$ to the $\bm{E}$ to make the table consistent. 

Table $\mathbf{T7}$ is not closed since there is no row in $\bm{S}$ having the same value with $(a,\{0,0\},\{\top,\top\})$. Therefore, the learner moves $(a,\{0,0\},\{\top,\top\})$ from $\bm{R}$ to $\bm{S}$ to make the timed observation table closed and get a prepared table $\mathbf{T8}$. 
Similarly, whenever the learner obtains a new counterexample through the equivalence query, the learner adds all prefixes of the counterexample to the timed observation table. When the timed observation table is prepared, the learner constructs a new DFA and hypothesis according to the observation table. After two such steps, the CTA $\mathcal{H}_{10}$ in Figure~\ref{fig:case-automata} is obtained, and the positive answer of an equivalence query shows that $\mathcal{L}(\mathcal{H}_{10})=\mathcal{L}(\mathcal{A})$. Thus, $\mathcal{H}_{10}$ is the correct model, and the learning process terminates. 

Note that the learned model may have tighter timing constraints on transitions than those of the target DTA. Compared to the target DTA $\mathcal{A}$ in Figure~\ref{fig:timedautomata}, the learned model $\mathcal{H}_{10}$ is a CTA with a sink location $l_{--}$, and the other two locations $l_{-+}$ and $l_{+-}$ correspond to the locations $l_0$ and $l_1$ of $\mathcal{A}$, respectively. We can find that two self transitions on $l_{+-}$ can be merged and form a new transition $(l_{+-}, b, c_1 > 1 \wedge 0 \leq c_2 < 1, \{c_2\}, l_{+-})$. Compared to the transition $(l_1, b, c_1 \geq 0 \wedge 0 \leq c_2 < 1, \{c_2\}, l_1)$ in $\mathcal{A}$, $c_1 >1$ is a correct and tighter constraint than $c_1 \geq 0$, since when reaching location $l_1$ the clock $c_1$ does never reset, its value must be greater than $1$. Similarly, the constraint of $c_1$ on the transition from $l_{+-}$ to $l_{-+}$ in $\mathcal{H}_{10}$ is tighter than the original constraint of $c_1$ on the transition from $l_1$ to $l_0$ in $\mathcal{A}$.

\section{Algorithm for learning DTA from a normal teacher}
\label{appendix:learnfromnormal}
\begin{algorithm}
%\scriptsize
	% \caption{Learning deterministic timed automata with multi-clock from a normal teacher}
 \caption{Learning deterministic timed automata with multiple clocks from a normal teacher}
	\label{alg:learningnormal}
	\SetKwInOut{Input}{input}
	\SetKwInOut{Output}{output}
	%\SetKwRepeat{Do}{do}{while}
	\Input{
 % the timed observation table $\mathbf{T} = (\Sigma, \bm{\Sigma}, \bm{\Sigma_r}, \bm{S}, \bm{R}, \bm{E}, f)$;
 the timed observation table  $\mathbf{T}=(\Sigma, \bm{\Sigma_{G}}, \bm{\Sigma_{r}}, \bm{S}, \bm{R}, \bm{E}, f, g)$;
 the number of clocks $|\mathcal{C}|$}
	\Output{the hypothesis $\mathcal{H}$ recognizing the target language $\mathcal{L}$.}
	$\mathit{ToExplore}\leftarrow\emptyset$;
	$\bm{S}\leftarrow\{\epsilon\}$;
	$\bm{R}\leftarrow\{\pi(\gamma) \mid \gamma=(\sigma,\{0\}^{|\mathcal{C}|}), \forall \sigma \in \Sigma \}$;
	$\bm{E}\leftarrow\{\epsilon\}$\; %\tcp*{initialization}
	$\mathit{currentTable} \leftarrow (\Sigma, \bm{\Sigma_{G}}, \bm{\Sigma_r}, \bm{S}, \bm{R}, \bm{E}, f, g)$\;
	$\mathit{tables}$ $\leftarrow$ guess\_and\_fill($\mathit{currentTable}$)\tcp*{guess resets and fill all table instances}
	$\mathit{ToExplore}$.insert($\mathit{tables}$)\tcp*{insert table instances $\mathit{tables}$ into $\mathit{ToExplore}$}  \label{line:initaltables}
	
	%\While{$\mathit{To\_explore}$ is not empty}{
	$\mathit{currentTable}$ $\leftarrow$ $\mathit{ToExplore}$.pop()\tcp*{pop out head instance as the current table}
	$\mathit{equivalent}$ $\leftarrow$ $\bot$\;
	\While{$\mathit{equivalent}$ = $\bot$}{
		$prepared$ $\leftarrow$ is\_prepared($\mathit{currentTable}$)\tcp*{whether the current table is prepared}
		\While{$prepared$ = $\bot$}
		{
			\If{$\mathit{currentTable}$ is not closed}{
				$\mathit{tables}$ $\leftarrow$ guess\_and\_make\_closed($\mathit{currentTable}$);
				$\mathit{ToExplore}$.insert($\mathit{tables}$)\;
				$\mathit{currentTable}$ $\leftarrow$ $\mathit{ToExplore}$.pop()\; 
			}
			\If{$\mathit{currentTable}$ is not consistent}{
				$\mathit{tables}$ $\leftarrow$ guess\_and\_make\_consistent($\mathit{currentTable}$);
				$\mathit{ToExplore}$.insert($\mathit{tables}$)\;
				$\mathit{currentTable}$ $\leftarrow$ $\mathit{ToExplore}$.pop()\; 
			}
			\If{$\mathit{currentTable}$ is not evidence-closed}{
				$\mathit{tables}$ $\leftarrow$ guess\_and\_make\_evidence\_closed($\mathit{currentTable}$);
				$\mathit{ToExplore}$.insert($\mathit{tables}$)\;
				$\mathit{currentTable}$ $\leftarrow$ $\mathit{ToExplore}$.pop()\;
			}
			%\If{$\mathit{currentTable}$ is not prefix-closed}{
				%$\mathit{tables}$ $\leftarrow$ guess\_and\_make\_prefix\_closed($\mathit{currentTable}$);
				%$\mathit{ToExplore}$.insert($\mathit{tables}$)\;
				%$\mathit{currentTable}$ $\leftarrow$ $\mathit{ToExplore}$.pop()\;
			%}
			$prepared$ $\leftarrow$ is\_prepared($\mathit{currentTable}$)\;
		}
		$\text{M} \leftarrow$ build\_DFA($\mathit{currentTable}$) \tcp*{transforming $\mathit{currentTable}$ to a DFA $\text{M}$}
		$\mathcal{H} \leftarrow$ build\_hypothesis($\text{M}$) \tcp*{constructing a hypothesis $\mathcal{H}$ from $\text{M}$}
		$\mathit{equivalent}$, $\mathit{ctx}$ $\leftarrow$ equivalence\_query($\mathcal{H}$)\tcp*{$\mathit{ctx}$ is a delay-timed word}
		\If{$\mathit{equivalent}$ = $\bot$}{
			$\mathit{tables}$ $\leftarrow$ guess\_and\_ctx\_processing($\mathit{currentTable}$, $\mathit{ctx}$) \tcp*{counterexample processing}
			$\mathit{ToExplore}$.insert($\mathit{tables}$)\;
			$\mathit{currentTable}$ $\leftarrow$ $\mathit{ToExplore}$.pop()\;
		}
	}
	\Return $\mathcal{H}$\;
	%}
 \end{algorithm}

\section{Timed automata in Section~\ref{sec:experiment}}
\label{appendix:targetDTAs}

Figure~\ref{fig:target_dta_experiment} presents the target DTAs used in our experiment (Section~\ref{sec:experiment}). For each target DTA, a case ID is in the form of $\lvert L\rvert\_\lvert \Sigma\rvert\_\lvert \mathcal{C}\rvert\_\lvert \kappa(\mathcal{C})\rvert$, which is designed to consist of the number of locations, the alphabet size, the number of clock variables and the maximum constant in the guards in the DTA.

\begin{figure}
\vspace{1cm}
% \begin{center}
\begin{center}
%\resizebox{0.21\textwidth}{!}{
\resizebox{0.18\textwidth}{!}{
\begin{tabular}{ P{2.5cm}|Y{0.3cm}}
 \hline
  $\textbf{T1}$&$\epsilon$\\
 \hline
 $\epsilon$&$-$\\
 \hline
  $(a,\{0,0\},\{\top,\top\})$   & $-$ \\
   $(b,\{0,0\},\{\top,\top\})$   & $-$    \\
 \hline
\end{tabular}
}
% $\xrightarrow[g(ctx_1)=(a,\{1.05,1.05\},\{\bot,\top\})]{ctx_1=(a,1.05,\{\bot,\top\}),+}$
$\xrightarrow[\Gamma(ctx_1)=(a,\{1.05,1.05\},\{\bot,\top\})]{ctx_1=(a,1.05,\{\bot,\top\}),+}$
% $\xrightarrow{(a,1.05,\{\bot,\top\}),+}$
 \\[5pt]
%\resizebox{0.25\textwidth}{!}{
\resizebox{0.22\textwidth}{!}{
\begin{tabular}{ P{3.4cm}|Y{0.3cm} }
 \hline
  $\textbf{T2}$&$\epsilon$\\
 \hline
 $\epsilon$&$-$\\
 \hline
  $(a,\{0,0\},\{\top,\top\})$   & $-$ \\
   $(b,\{0,0\},\{\top,\top\})$   & $-$    \\
   $(a,\{1.05,1.05\},\{\bot,\top\})$   & $+$    \\
 \hline
\end{tabular}
}
$\xrightarrow{\ \ \ \ \ \ \ \ \ \ \ \ closed\ \ \ \ \ \ \ \ \ \ \ \ }$
 \\[5pt]

%\resizebox{0.42\textwidth}{!}{
\resizebox{0.38\textwidth}{!}{
\begin{tabular}{ P{6.5cm}|Y{0.3cm} }
 \hline
  $\textbf{T3}$&$\epsilon$\\
 \hline
 $\epsilon$&$-$\\
 $(a,\{1.05,1.05\},\{\bot,\top\})$   & $+$    \\
 \hline
  $(a,\{0,0\},\{\top,\top\})$   & $-$ \\
   $(b,\{0,0\},\{\top,\top\})$   & $-$    \\
   $(a,\{1.05,1.05\},\{\bot,\top\})(a,\{0,0\},\{\top,\top\})$   & $-$    \\
   $(a,\{1.05,1.05\},\{\bot,\top\})(b,\{0,0\},\{\top,\top\})$   & $-$    \\
 \hline
\end{tabular}
}
\\[5pt]
\quad
% $\xrightarrow[g(ctx_2)=(a,\{1.05,1.05\},\{\bot,\top\})(b,\{1.05,0\},\{\bot,\top\})]{ctx_2=(a,1.05,\{\bot,\top\})(b,0,\{\bot,\top\}),+}$
$\xrightarrow[\Gamma(ctx_2)=(a,\{1.05,1.05\},\{\bot,\top\})(b,\{1.05,0\},\{\bot,\top\})]{ctx_2=(a,1.05,\{\bot,\top\})(b,0,\{\bot,\top\}),+}$

% $\xrightarrow{(a,1.05,\{\bot,\top\})(b,0,\{\bot,\top\}),+}$}
%\resizebox{0.42\textwidth}{!}{

\resizebox{0.38\textwidth}{!}{
\begin{tabular}{ P{6.5cm}|Y{0.3cm} }
 \hline
  $\textbf{T4}$&$\epsilon$\\
 \hline
 $\epsilon$&$-$\\
 $(a,\{1.05,1.05\},\{\bot,\top\})$   & $+$    \\
 \hline
  $(a,\{0,0\},\{\top,\top\})$   & $-$ \\
   $(b,\{0,0\},\{\top,\top\})$   & $-$    \\
   $(a,\{1.05,1.05\},\{\bot,\top\})(a,\{0,0\},\{\top,\top\})$   & $-$    \\
   $(a,\{1.05,1.05\},\{\bot,\top\})(b,\{0,0\},\{\top,\top\})$   & $-$    \\
    $(a,\{1.05,1.05\},\{\bot,\top\})(b,\{1.05,0\},\{\bot,\top\})$   & $+$    \\
 \hline
\end{tabular}
}
\\[5pt]
\quad
% $\xrightarrow[g(ctx_3)=(a,\{1.05,1.05\},\{\bot,\top\})(b,\{2.05,1\},\{\top,\top\})]{ctx_3=(a,1.05,\{\bot,\top\})(b,1,\{\top,\top\}),-}$
$\xrightarrow[\Gamma(ctx_3)=(a,\{1.05,1.05\},\{\bot,\top\})(b,\{2.05,1\},\{\top,\top\})]{ctx_3=(a,1.05,\{\bot,\top\})(b,1,\{\top,\top\}),-}$

% $\xrightarrow{\   (a,1.05,\{\bot,\top\})(b,1,\{\top,\top\}),-\   }$

%\resizebox{0.42\textwidth}{!}{
\resizebox{0.38\textwidth}{!}{
\begin{tabular}{ P{6.5cm}|Y{0.3cm} }
 \hline
  $\textbf{T5}$&$\epsilon$\\
 \hline
 $\epsilon$&$-$\\
 $(a,\{1.05,1.05\},\{\bot,\top\})$   & $+$    \\
 \hline
  $(a,\{0,0\},\{\top,\top\})$   & $-$ \\
   $(b,\{0,0\},\{\top,\top\})$   & $-$    \\
   $(a,\{1.05,1.05\},\{\bot,\top\})(a,\{0,0\},\{\top,\top\})$   & $-$    \\
   $(a,\{1.05,1.05\},\{\bot,\top\})(b,\{0,0\},\{\top,\top\})$   & $-$    \\
    $(a,\{1.05,1.05\},\{\bot,\top\})(b,\{1.05,0\},\{\bot,\top\})$   & $+$    \\
    $(a,\{1.05,1.05\},\{\bot,\top\})(b,\{2.05,1\},\{\top,\top\})$   & $-$    \\
 \hline
\end{tabular}
}
\\[5pt]
\quad
% $\xrightarrow[g(ctx_4)=(a,\{0,0\},\{\top,\top\})(a,\{1.05,1.05\},\{\top,\top\})]{ctx_4=(a,0,\{\top,\top\})(a,1.05,\{\top,\top\}),-}$
$\xrightarrow[\Gamma(ctx_4)=(a,\{0,0\},\{\top,\top\})(a,\{1.05,1.05\},\{\top,\top\})]{ctx_4=(a,0,\{\top,\top\})(a,1.05,\{\top,\top\}),-}$
% $ \xrightarrow{(a,0,\{\top,\top\})(a,1.05,\{\top,\top\}),-}$
%\resizebox{0.42\textwidth}{!}{
\resizebox{0.38\textwidth}{!}{
\begin{tabular}{ P{6.5cm}|Y{0.3cm}  }
 \hline
  $\textbf{T6}$&$\epsilon$\\
 \hline
 $\epsilon$&$-$\\
 $(a,\{1.05,1.05\},\{\bot,\top\})$   & $+$    \\
 \hline
  $(a,\{0,0\},\{\top,\top\})$   & $-$ \\
   $(b,\{0,0\},\{\top,\top\})$   & $-$    \\
   $(a,\{1.05,1.05\},\{\bot,\top\})(a,\{0,0\},\{\top,\top\})$   & $-$    \\
   $(a,\{1.05,1.05\},\{\bot,\top\})(b,\{0,0\},\{\top,\top\})$   & $-$    \\
    $(a,\{1.05,1.05\},\{\bot,\top\})(b,\{1.05,0\},\{\bot,\top\})$   & $+$    \\
    $(a,\{1.05,1.05\},\{\bot,\top\})(b,\{2.05,1\},\{\top,\top\})$   & $-$    \\
    $(a,\{0,0\},\{\top,\top\})(a,\{1.05,1.05\},\{\top,\top\})$   & $-$    \\
 \hline
\end{tabular}
}
\end{center}
% \end{center}
\caption{Iterations of the timed observation table during the learning process w.r.t $\mathcal{A}$ in Figure~\ref{fig:timedautomata}.}
\label{fig:case-tables}
\end{figure}

\begin{figure*}
\ContinuedFloat
\vspace{1cm}
% \begin{center}
\begin{center}
%\resizebox{0.21\textwidth}{!}{
\resizebox{0.3\textwidth}{!}{
%\resizebox{0.42\textwidth}{!}{
\quad
 $\xrightarrow{\ \ \ \ \ \ \ \ \ \ \ \ \ \ \ \ \ \ \ consistent\ \ \ \ \ \ \ \ \ \ \ \ \ \ \ \ \ \ \ }$}
 \\[5pt]
 \resizebox{0.7\textwidth}{!}{
  %\resizebox{0.7125\textwidth}{!}{
\resizebox{0.66\textwidth}{!}{
\begin{tabular}{ P{6.5cm}|Y{0.3cm} Y{5.4cm}  }
 \hline
  $\textbf{T7}$&$\epsilon$&$(a,1< c_{1}< 2\wedge   1<c_{2}< 2\wedge c_1=c_2)$\\
 \hline
 $\epsilon$&$-$&$+,\bot,\top$\\
 $(a,\{1.05,1.05\},\{\bot,\top\})$   & $+$&$-,\top,\top$    \\
 \hline
  $(a,\{0,0\},\{\top,\top\})$   & $-$&$-,\top,\top$ \\
   $(b,\{0,0\},\{\top,\top\})$   & $-$&$-,\top,\top$    \\
   $(a,\{1.05,1.05\},\{\bot,\top\})(a,\{0,0\},\{\top,\top\})$   & $-$&$-,\top,\top$    \\
   $(a,\{1.05,1.05\},\{\bot,\top\})(b,\{0,0\},\{\top,\top\})$   & $-$&$-,\top,\top$    \\
    $(a,\{1.05,1.05\},\{\bot,\top\})(b,\{1.05,0\},\{\bot,\top\})$   & $+$&$-,\top,\top$    \\
    $(a,\{1.05,1.05\},\{\bot,\top\})(b,\{2.05,1\},\{\top,\top\})$   & $-$&$-,\top,\top$    \\
    $(a,\{0,0\},\{\top,\top\})(a,\{1.05,1.05\},\{\top,\top\})$   & $-$&$-,\top,\top$    \\
 \hline
\end{tabular}
}
\quad
$\xrightarrow{closed}$}
\\[5pt]
\resizebox{0.69\textwidth}{!}{
\begin{tabular}{ P{7cm}|Y{0.3cm} Y{5.4cm}  }
 \hline
  $\textbf{T8}$&$\epsilon$&$(a,1< c_{1}< 2\wedge  1<c_{2}< 2\wedge c_1=c_2)$\\
 \hline
 $\epsilon$&$-$&$+,\bot,\top$\\
 $(a,\{1.05,1.05\},\{\bot,\top\})$   & $+$&$-,\top,\top$    \\
 $(a,\{0,0\},\{\top,\top\})$   & $-$&$-,\top,\top$ \\
 \hline
   $(b,\{0,0\},\{\top,\top\})$   & $-$&$-,\top,\top$    \\
   $(a,\{1.05,1.05\},\{\bot,\top\})(a,\{0,0\},\{\top,\top\})$   & $-$&$-,\top,\top$    \\
   $(a,\{1.05,1.05\},\{\bot,\top\})(b,\{0,0\},\{\top,\top\})$   & $-$&$-,\top,\top$    \\
    $(a,\{1.05,1.05\},\{\bot,\top\})(b,\{1.05,0\},\{\bot,\top\})$   & $+$&$-,\top,\top$    \\
    $(a,\{1.05,1.05\},\{\bot,\top\})(b,\{2.05,1\},\{\top,\top\})$   & $-$&$-,\top,\top$    \\
    $(a,\{0,0\},\{\top,\top\})(a,\{1.05,1.05\},\{\top,\top\})$   & $-$&$-,\top,\top$    \\
    $(a,\{0,0\},\{\top,\top\})(a,\{0,0\},\{\top,\top\})$   & $-$&$-,\top,\top$ \\
    $(a,\{0,0\},\{\top,\top\})(b,\{0,0\},\{\top,\top\})$   & $-$&$-,\top,\top$ \\
 \hline
\end{tabular}
}
\begin{center}
% $\xrightarrow{(a,1.05,\bot,\top)(a,1.05,\top,\top)(a,1.05,\bot,\top),+}$
% $\xrightarrow[g(ctx_5)=(a,\{1.05,1.05\},\bot,\top)(a,\{2.1,1.05\},\top,\top)(a,\{1.05,1.05\},\bot,\top)]{ctx_5=(a,1.05,\bot,\top)(a,1.05,\top,\top)(a,1.05,\bot,\top),+}$
$\xrightarrow[\Gamma(ctx_5)=(a,\{1.05,1.05\},\{\bot,\top\})(a,\{2.1,1.05\},\{\top,\top\})(a,\{1.05,1.05\},\{\bot,\top\})]{ctx_5=(a,1.05,\{\bot,\top\})(a,1.05,\{\top,\top\})(a,1.05,\{\bot,\top\}),+}$
\end{center}
\resizebox{0.85\textwidth}{!}{
\begin{tabular}{ P{10.4cm}|Y{0.3cm} Y{5.4cm}  }
 \hline
  $\textbf{T9}$&$\epsilon$&$(a,1< c_{1}< 2\wedge  1<c_{2}< 2\wedge c_1=c_2)$\\
 \hline
 $\epsilon$&$-$&$+,\bot,\top$\\
 $(a,\{1.05,1.05\},\{\bot,\top\})$   & $+$&$-,\top,\top$    \\
 $(a,\{0,0\},\{\top,\top\})$   & $-$&$-,\top,\top$ \\
 \hline
   $(b,\{0,0\},\{\top,\top\})$   & $-$&$-,\top,\top$    \\
   $(a,\{1.05,1.05\},\{\bot,\top\})(a,\{0,0\},\{\top,\top\})$   & $-$&$-,\top,\top$    \\
   $(a,\{1.05,1.05\},\{\bot,\top\})(b,\{0,0\},\{\top,\top\})$   & $-$&$-,\top,\top$    \\
    $(a,\{1.05,1.05\},\{\bot,\top\})(b,\{1.05,0\},\{\bot,\top\})$   & $+$&$-,\top,\top$    \\
    $(a,\{1.05,1.05\},\{\bot,\top\})(b,\{2.05,1\},\{\top,\top\})$   & $-$&$-,\top,\top$    \\
    $(a,\{0,0\},\{\top,\top\})(a,\{1.05,1.05\},\{\top,\top\})$   & $-$&$-,\top,\top$    \\
    $(a,\{0,0\},\{\top,\top\})(a,\{0,0\},\{\top,\top\})$   & $-$&$-,\top,\top$ \\
    $(a,\{0,0\},\{\top,\top\})(b,\{0,0\},\{\top,\top\})$   & $-$&$-,\top,\top$ \\
    $(a,\{1.05,1.05\},\{\bot,\top\})(a,\{2.1,1.05\},\{\top,\top\})$   & $-$&$+,\bot,\top$ \\
    $(a,\{1.05,1.05\},\{\bot,\top\})(a,\{2.1,1.05\},\{\top,\top\})(a,\{1.05,1.05\},\{\bot,\top\})$   & $+$&$-,\top,\top$ \\
 \hline
\end{tabular}
}
\begin{center}
% $\xrightarrow{(a,1.05,\bot,\top)(a,1.95,\top,\top)(a,1.05,\top,\top),-}$  
% $\xrightarrow[g(ctx_6)=(a,\{1.05,1.05\},\bot,\top)(a,\{3,1.95\},\top,\top)(a,\{1.05,1.05\},\top,\top)]{ctx_6=(a,1.05,\bot,\top)(a,1.95,\top,\top)(a,1.05,\top,\top),-}$
$\xrightarrow[\Gamma(ctx_6)=(a,\{1.05,1.05\},\{\bot,\top\})(a,\{3,1.95\},\{\top,\top\})(a,\{1.05,1.05\},\{\top,\top\})]{ctx_6=(a,1.05,\{\bot,\top\})(a,1.95,\{\top,\top\})(a,1.05,\{\top,\top\}),-}$
\end{center}  
\resizebox{0.85\textwidth}{!}{
\begin{tabular}{ P{10.4cm}|Y{0.3cm} Y{5.4cm}  }
 \hline
  $\textbf{T10}$&$\epsilon$&$(a,1< c_{1}< 2\wedge   1<c_{2}< 2\wedge c_1=c_2)$\\
 \hline
 $\epsilon$&$-$&$+,\bot,\top$\\
 $(a,\{1.05,1.05\},\{\bot,\top\})$   & $+$&$-,\top,\top$    \\
 $(a,\{0,0\},\{\top,\top\})$   & $-$&$-,\top,\top$ \\
 \hline
   $(b,\{0,0\},\{\top,\top\})$   & $-$&$-,\top,\top$    \\
   $(a,\{1.05,1.05\},\{\bot,\top\})(a,\{0,0\},\{\top,\top\})$   & $-$&$-,\top,\top$    \\
   $(a,\{1.05,1.05\},\{\bot,\top\})(b,\{0,0\},\{\top,\top\})$   & $-$&$-,\top,\top$    \\
    $(a,\{1.05,1.05\},\{\bot,\top\})(b,\{1.05,0\},\{\bot,\top\})$   & $+$&$-,\top,\top$    \\
    $(a,\{1.05,1.05\},\{\bot,\top\})(b,\{2.05,1\},\{\top,\top\})$   & $-$&$-,\top,\top$    \\
    $(a,\{0,0\},\{\top,\top\})(a,\{1.05,1.05\},\{\top,\top\})$   & $-$&$-,\top,\top$    \\
    $(a,\{0,0\},\{\top,\top\})(a,\{0,0\},\{\top,\top\})$   & $-$&$-,\top,\top$ \\
    $(a,\{0,0\},\{\top,\top\})(b,\{0,0\},\{\top,\top\})$   & $-$&$-,\top,\top$ \\
    $(a,\{1.05,1.05\},\{\bot,\top\})(a,\{2.1,1.05\},\{\top,\top\})$   & $-$&$+,\bot,\top$ \\
    $(a,\{1.05,1.05\},\{\bot,\top\})(a,\{2.1,1.05\},\{\top,\top\})(a,\{1.05,1.05\},\{\bot,\top\})$   & $+$&$-,\top,\top$ \\
    $(a,\{1.05,1.05\},\{\bot,\top\})(a,\{3,1.95\},\{\top,\top\})$   & $-$&$-,\top,\top$ \\
    $(a,\{1.05,1.05\},\{\bot,\top\})(a,\{3,1.95\},\{\top,\top\})(a,\{1.05,1.05\},\{\top,\top\})$   & $-$&$-,\top,\top$ \\
 \hline
\end{tabular}
}
\end{center}
% \end{center}
\caption{Iterations of the timed observation table during the learning process w.r.t $\mathcal{A}$ in Figure~\ref{fig:timedautomata} (Continued).}
%\label{fig:case-tables}
\end{figure*}

\clearpage
\begin{figure*}
% \flushleft
\begin{center}
\begin{minipage}[b]{0.14\textwidth}
\begin{center}
\resizebox{1\textwidth}{!}{
\begin{tikzpicture}[scale=0.65, ->, >=stealth', shorten >=1pt, auto, node distance=2cm, semithick, every node/.style={scale=0.65}]
        \node[initial, state]  (0) {$l_-$};
        \path  (0) edge[in=240, out=300,loop, color=black] node[below, black, align=center] {$(b,\{0,0\},\{\top, \top\})$} (0)
        
        (0) edge[in= 60, out=120,loop, color=black] node[above, black, align=center] {$(a,\{0,0\},\{\top, \top\})$} (0);
        \node [below=40pt, align=flush center,text width=3cm] at (0) {$\text{M}_1$};
\end{tikzpicture}}
\end{center}
\end{minipage}
\begin{minipage}[b]{0.155\textwidth}
\begin{center}
\resizebox{1\textwidth}{!}{
\begin{tikzpicture}[scale=0.65, ->, >=stealth', shorten >=1pt, auto, node distance=2cm, semithick, every node/.style={scale=0.65}]
\centering
        \node[initial, state]  (0) {$l_-$};
        \path  (0) edge[in=240, out=300,loop, color=black] node[below, black, align=center] {$b,c_{1}\geq 0\wedge c_{2}\geq 0,\{c_1, c_2\}$} (0)
        
        (0) edge[in= 60, out=120,loop, color=black] node[above, black, align=center] {$a,c_{1}\geq 0\wedge c_{2}\geq 0,\{c_1, c_2\}$} (0);

         \node [below=40pt, align=flush center,text width=3cm] at (0) {$\mathcal{H}_1$};
\end{tikzpicture}
}
\end{center}
\end{minipage}
\begin{minipage}[b]{0.32\textwidth}
\begin{center}
\resizebox{1\textwidth}{!}{
\begin{tikzpicture}[scale=0.65, ->, >=stealth', shorten >=1pt, auto, node distance=2cm, semithick, every node/.style={scale=0.65}]
%\centering
        \node[initial, state]  (0) {$l_-$};
        \node[accepting, state]  (1) at (5,0) {$l_+$};
        \path  (0) edge[in=210, out=240,loop, color=black] node[below, black, align=center] {$(b,\{0,0\},\{\top, \top\})$} (0)
        
        (0) edge[in= 150, out=120,loop, color=black] node[above, black, align=center] {$(a,\{0,0\},\{\top, \top\})$} (0)
        
        (1) edge[in= 330, out=210, color=black] node[below, black, align=center] {$(a,\{0,0\},\{\top, \top\})$} (0)
        
        (1) edge[in= 300, out=240, color=black] node[below, black, align=center] {$(b,\{0,0\},\{\top, \top\})$} (0)
        
        (0) edge[color=black] node[above, black, align=center] {$(a,\{1.05,1.05\},\{\bot, \top\})$} (1);

        \node [below right= 45pt and 20pt, align=flush center,text width=3cm] at (0) {$\text{M}_3$};
\end{tikzpicture}
}
\end{center}
\end{minipage}
\begin{minipage}[b]{0.37\textwidth}
\begin{center}
\resizebox{1\textwidth}{!}{
\begin{tikzpicture}[scale=0.65, ->, >=stealth', shorten >=1pt, auto, node distance=2cm, semithick, every node/.style={scale=0.65}]
\centering
        \node[initial, state]  (0) {$l_-$};
        \node[accepting, state]  (1) at (5,0) {$l_+$};
        \path  (0) edge[in=210, out=240,loop, color=black] node[left, black, align=center] {$b,c_{1}\geq 0\wedge c_{2}\geq 0,$ \\$\{c_1, c_2\}$} (0)
        
        (0) edge[in= 150, out=120,loop, color=black] node[left, black, align=center] {$a, c_{1}> 1\wedge c_{2}\leq 1,$ \\$\{c_1,c_2\}$} (0)

        (0) edge[in= 30, out=60,loop, color=black] node[right, black, align=center] {$a, c_{1}\leq 1\wedge c_{2}\geq 0,\{c_1,c_2\}$} (0)
        
        (1) edge[in= 330, out=210, color=black] node[below, black, align=center] {$a,c_{1}\geq 0\wedge c_{2}\geq 0,\{c_1, c_2\}$} (0)
        
        (1) edge[in= 300, out=240, color=black] node[below, black, align=center] {$b,c_{1}\geq 0\wedge c_{2}\geq 0,\{c_1, c_2\}$} (0)
        
        (0) edge[color=black] node[above, black, align=center] {$a,c_{1}> 1\wedge c_{2} > 1,\{c_2\}$} (1);

        \node [below right=45pt and 5pt, align=flush center,text width=3cm] at (0) {$\mathcal{H}_3$};
\end{tikzpicture}
}
\end{center}
\end{minipage}
\begin{minipage}[b]{0.49\textwidth}
\begin{center}
\resizebox{0.75\textwidth}{!}{
\begin{tikzpicture}[scale=0.65, ->, >=stealth', shorten >=1pt, auto, node distance=2cm, semithick, every node/.style={scale=0.65}]
%\centering
        \node[initial, state]  (0) {$l_-$};
        \node[accepting, state]  (1) at (5,0) {$l_+$};
        \path  (0) edge[in=210, out=240,loop, color=black] node[below, black, align=center] {$(b,\{0,0\},\{\top, \top\})$} (0)
        
        (0) edge[in= 150, out=120,loop, color=black] node[above, black, align=center] {$(a,\{0,0\},\{\top, \top\})$} (0)
        
        (1) edge[in= 330, out=210, color=black] node[below, black, align=center] {$(a,\{0,0\},\{\top, \top\})$} (0)
        
        (1) edge[in= 300, out=240, color=black] node[below, black, align=center] {$(b,\{0,0\},\{\top, \top\})$} (0)
        
        (0) edge[color=black] node[above, black, align=center] {$(a,\{1.05,1.05\},\{\bot, \top\})$} (1)
        
        (1) edge[in=75, out=105,loop, color=black] node[above, black, align=center] {$(b,\{1.05,0\},\{\bot, \top\})$} (1);

        \node [below right=45pt and 20pt, align=flush center,text width=3cm] at (0) {$\text{M}_4$};
\end{tikzpicture}
}
\end{center}
\end{minipage}
\begin{minipage}[b]{0.49\textwidth}
\begin{center}
\resizebox{0.9\textwidth}{!}{
\begin{tikzpicture}[scale=0.65, ->, >=stealth', shorten >=1pt, auto, node distance=2cm, semithick, every node/.style={scale=0.65}]
\centering
        \node[initial, state]  (0) {$l_-$};
        \node[accepting, state]  (1) at (5,0) {$l_+$};
        \path  (0) edge[in=210, out=240,loop, color=black] node[left, black, align=center] {$b,c_{1}\geq 0\wedge c_{2}\geq 0,$ \\$\{c_1, c_2\}$} (0)
        
        (0) edge[in= 150, out=120,loop, color=black] node[left, black, align=center] {$a, c_{1}> 1\wedge c_{2}\leq 1, $ \\$\{c_1,c_2\}$} (0)

        (0) edge[in= 30, out=60,loop, color=black] node[right, black, align=center] {$a, c_{1}\leq 1\wedge c_{2}\geq 0,\{c_1,c_2\}$} (0)
        
        (1) edge[in= 330, out=210, color=black] node[below, black, align=center] {$a,c_{1}\geq 0\wedge c_{2}\geq 0,\{c_1, c_2\}$} (0)
        
        (1) edge[in= 300, out=240, color=black] node[below, black, align=center] {$b,0\leq c_{1}\leq 1 \wedge c_{2}\geq 0,\{c_1, c_2\}$} (0)
        
        (0) edge[color=black] node[above, black, align=center] {$a,c_{1}> 1\wedge c_{2} > 1,\{c_2\}$} (1)

        (1) edge[in=75, out=105,loop, color=black] node[above, black, align=center] {$b,c_{1}>1 \wedge c_{2}\geq 0,\{c_2\}$} (1);

        \node [below right=45pt and 5pt, align=flush center,text width=3cm] at (0) {$\mathcal{H}_4$};
\end{tikzpicture}
}
\end{center}
\end{minipage}
\begin{minipage}[b]{0.49\textwidth}
\begin{center}
\resizebox{0.8\textwidth}{!}{
\begin{tikzpicture}[scale=0.65, ->, >=stealth', shorten >=1pt, auto, node distance=2cm, semithick, every node/.style={scale=0.65}]
\centering
        \node[initial, state]  (0) {$l_-$};
        \node[accepting, state]  (1) at (5.5,0) {$l_+$};
        \path  (0) edge[in=210, out=240,loop, color=black] node[below, black, align=center] {$(b,\{0,0\},\{\top, \top\})$} (0)
        
        (0) edge[in= 150, out=120,loop, color=black] node[above, black, align=center] {$(a,\{0,0\},\{\top, \top\})$} (0)
        
        (1) edge[in= 330, out=210, color=black] node[below, black, align=center] {$(a,\{0,0\},\{\top, \top\}$} (0)
        
        (1) edge[in= 300, out=240, color=black] node[below, black, align=center] {$(b,\{0,0\},\{\top, \top\})$} (0)
        
        (0) edge[color=black] node[above, black, align=center] {$(a,\{1.05,1.05\},\{\bot, \top\})$} (1)
        
        (1) edge[in=75, out=105,loop, color=black] node[above, black, align=center] {$(b,\{1.05,0\},\{\bot, \top\})$} (1)
        
        (1) edge[in=30, out=150, color=black] node[above, black, align=center,pos=0.58] {$(b,\{2.05,1\},\{\top, \top\})$} (0);

        \node [below right=45pt and 25pt, align=flush center,text width=3cm] at (0) {$\text{M}_5$};
\end{tikzpicture}
}
\end{center}
\end{minipage}
\begin{minipage}[b]{0.49\textwidth}
\begin{center}
\resizebox{1\textwidth}{!}{
\begin{tikzpicture}[scale=0.65, ->, >=stealth', shorten >=1pt, auto, node distance=2cm, semithick, every node/.style={scale=0.65}]
\centering
        \node[initial, state]  (0) {$l_-$};
        \node[accepting, state]  (1) at (5.5,0) {$l_+$};
        \path  (0) edge[in=210, out=240,loop, color=black] node[left, black, align=center] {$b,c_{1}\geq 0\wedge c_{2}\geq 0,$ \\$\{c_1, c_2\}$} (0)
        
        (0) edge[in= 150, out=120,loop, color=black] node[left, black, align=center] {$a, c_{1}> 1\wedge  c_{2}\leq 1, $ \\$\{c_1,c_2\}$} (0)

        (0) edge[in= 105, out=75,loop, color=black] node[above, black, align=center] {$a,  c_{1}\leq 1\wedge c_{2}\geq 0,$ \\$\{c_1,c_2\}$} (0)
        
        (1) edge[in= 330, out=210, color=black] node[below, black, align=center] {$a,c_{1}\geq 0\wedge c_{2}\geq 0,\{c_1, c_2\}$} (0)
        
        (1) edge[in= 300, out=240, color=black] node[below, black, align=center] {$b,0\leq c_{1}\leq 1 \wedge c_{2}\geq 0,\{c_1, c_2\}$} (0)
        
        (0) edge[color=black] node[above, black, align=center] {$a,c_{1}> 1\wedge c_{2} > 1,\{c_2\}$} (1)

        (1) edge[in=30, out=60,loop, color=black] node[above, black, align=center] {$b,1<c_{1}\leq 2 \wedge c_{2}\geq 0,$ \\$\{c_2\}$} (1)

        (1) edge[in=30, out=150, color=black] node[above, black, align=center,pos=0.45] {$b,c_{1}>2 \wedge c_{2}\geq 1,$ \\$\{c_1, c_2\}$} (0)
        
        (1) edge[in=330, out=300,loop, color=black] node[below, black, align=center] {$b,c_{1}>2\wedge c_{2}< 1,$ \\$\{c_2\}$} (1);

        \node [below right=45pt and 15pt, align=flush center,text width=3cm] at (0) {$\mathcal{H}_5$};
\end{tikzpicture}
}
\end{center}
\end{minipage}
\begin{minipage}[b]{0.49\textwidth}
\begin{center}
\resizebox{0.75\textwidth}{!}{
\begin{tikzpicture}[scale=0.65, ->, >=stealth', shorten >=1pt, auto, node distance=2cm, semithick, every node/.style={scale=0.65}]
\centering
        \node[initial, state]  (0) {$l_{-+}$};
        \node[accepting, state]  (1) at (5,0) {$l_{+-}$};
        \node[state]  (2) at (2.5,-4) {$l_{--}$};
        \path  (0) edge[in=150, out=270, color=black] node[below,sloped, black, align=center] {$(b,\{0,0\},\{\top, \top\})$} (2)
        
        (0) edge[color=black] node[below, sloped, black, align=center] {$(a,\{0,0\},\{\top, \top\})$} (2)
        
        (1) edge[in= 75, out=210, color=black] node[above, sloped, black, align=center] {$(a,\{0,0\},\{\top, \top\})$} (2)
        
        (1) edge[color=black] node[below, sloped, black, align=center] {$(b,\{0,0\},\{\top, \top\})$} (2)
        
        (0) edge[color=black] node[above, black, align=center] {$(a,\{1.05,1.05\},\{\bot, \top\})$} (1)
        
        (1) edge[in=75, out=105,loop, color=black] node[above, black, align=center] {$(b,\{1.05,0\},\{\bot, \top\})$} (1)
        
        (1) edge[in= 30, out=270, color=black] node[below, sloped, black, align=center] {$(b,\{2.05,1\},\{\top, \top\})$} (2)
        
        (2) edge[in=230, out=180,loop, color=black] node[left, black, align=center] {$(a,\{0,0\},\{\top, \top\})$} (2)
        
        (2) edge[in=290, out=240,loop, color=black] node[below, black, align=center] {$(a,\{1.05,1.05\},\{\top, \top\})$} (2)
        
        (2) edge[in=350, out=300,loop, color=black] node[right, black, align=center] {$(b,\{0,0\},\{\top, \top\})$} (2);

        \node [below=40pt, align=flush center,text width=3cm] at (2) {$\text{M}_8$};
\end{tikzpicture}
}
\end{center}
\end{minipage}
\begin{minipage}[b]{0.49\textwidth}
\begin{center}
\resizebox{0.9\textwidth}{!}{
\begin{tikzpicture}[scale=0.65, ->, >=stealth', shorten >=1pt, auto, node distance=2cm, semithick, every node/.style={scale=0.65}]
\centering
        \node[initial, state]  (0) {$l_{-+}$};
        \node[accepting, state]  (1) at (5,0) {$l_{+-}$};
        \node[state]  (2) at (2.5,-4) {$l_{--}$};
        \path  (0) edge[in=180, out=240, color=black] node[below,sloped, black, align=center] {$b,c_{1}\geq 0\wedge c_{2}\geq 0,$ \\$\{c_1, c_2\}$} (2)
        
        (0) edge[color=black] node[above, sloped,black, align=center] {$a, c_{1}> 1\wedge  c_{2}\leq 1, $ \\$\{c_1,c_2\}$} (2)

        (0) edge[in= 150, out=270, color=black] node[above, sloped,black, align=center] {$a,  c_{1}\leq 1\wedge c_{2}\geq 0,$ \\$\{c_1,c_2\}$} (2)
        
        (1) edge[color=black] node[above,sloped, black, align=center,pos=0.4] {$a,c_{1}\geq 0\wedge c_{2}\geq 0,$ \\$\{c_1, c_2\}$} (2)
        
        (1) edge[in= 30, out=270, color=black] node[above,sloped, black, align=center] {$b,0\leq c_{1}\leq 1 \wedge c_{2}\geq 0,$ \\$\{c_1, c_2\}$} (2)
        
        (0) edge[color=black] node[above, black, align=center] {$a,c_{1}> 1\wedge c_{2} > 1,\{c_2\}$} (1)

        (1) edge[in=120, out=90,loop, color=black] node[above, black, align=center] {$b,1<c_{1}\leq 2 \wedge c_{2}\geq 0,$ \\$\{c_2\}$} (1)

        (1) edge[in=0, out=300, color=black] node[below, sloped, black, align=center] {$b,c_{1}>2 \wedge c_{2}\geq 1,$ \\ $\{c_1, c_2\}$} (2)
        
        (1) edge[in=30, out=60,loop, color=black] node[right, black, align=center] {$b,c_{1}>2\wedge c_{2}< 1,$ \\$\{c_2\}$} (1)
        
        (2) edge[in=210, out=240,loop, color=black] node[left, black, align=center] {$a, c_{1}\geq 0\wedge c_{2}\geq 0, $ \\$\{c_1,c_2\}$} (2)
        
        (2) edge[in=330, out=300,loop, color=black] node[right, black, align=center] {$b, c_{1}\geq 0\wedge c_{2}\geq 0, $ \\$\{c_1,c_2\}$} (2);

        \node [below=40pt, align=flush center,text width=3cm] at (2) {$\mathcal{H}_8$};
\end{tikzpicture}
}
\end{center}
\end{minipage}
%}
%}
\end{center}
\caption{Iterations of intermediate DFA $\mathrm{M}$ and hypothesis $\mathcal{H}$ during the learning process w.r.t $\mathcal{A}$ in Figure~\ref{fig:timedautomata}}
\label{fig:case-automata}
\end{figure*}
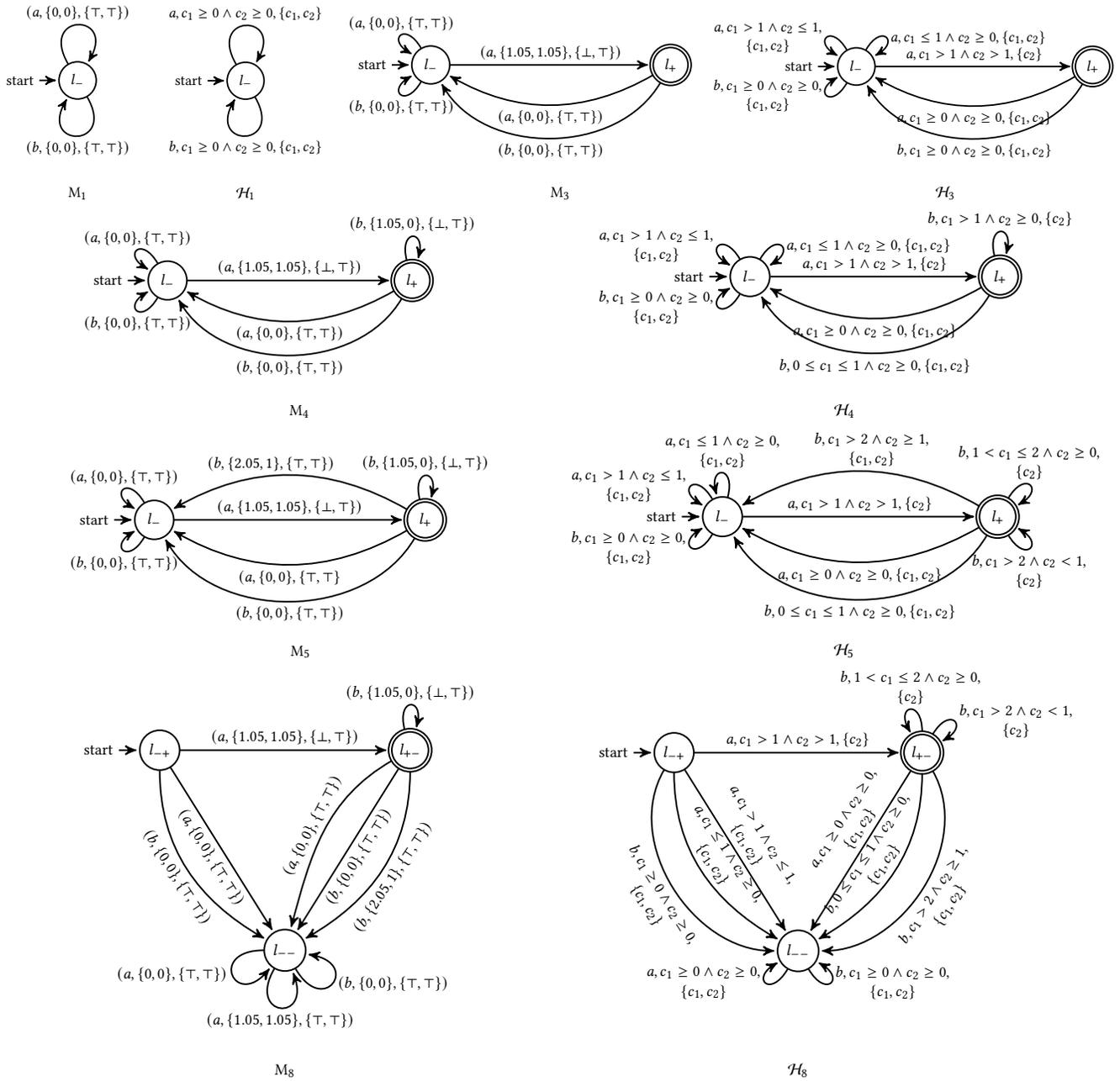

\begin{figure*}
  \ContinuedFloat
%\resizebox{1\textwidth}{!}{
%}
%\resizebox{1\textwidth}{!}{
%\resizebox{1\textwidth}{!}{
\begin{center}
\begin{minipage}[b]{0.49\textwidth}
\begin{center}
\resizebox{0.75\textwidth}{!}{
\begin{tikzpicture}[scale=0.65, ->, >=stealth', shorten >=1pt, auto, node distance=2cm, semithick, every node/.style={scale=0.65}]
\centering
        \node[initial, state]  (0) {$l_{-+}$};
        \node[accepting, state]  (1) at (5,0) {$l_{+-}$};
        \node[state]  (2) at (2.5,-4) {$l_{--}$};
        \path  (0) edge[in=150, out=270, color=black] node[below,sloped, black, align=center] {$(b,\{0,0\},\{\top, \top\})$} (2)
        
        (0) edge[color=black] node[below,sloped, black, align=center] {$(a,\{0,0\},\{\top, \top\})$} (2)
        
        (1) edge[in=75, out=210, color=black] node[above,sloped, black, align=center] {$(a,\{0,0\},\{\top, \top\})$} (2)
        
        (1) edge[color=black] node[below,sloped, black, align=center] {$(b,\{0,0\},\{\top, \top\})$} (2)
        
        (0) edge[color=black] node[above, black, align=center] {$(a,\{1.05,1.05\},\{\bot, \top\})$} (1)
        
        (1) edge[in=75, out=105,loop, color=black] node[above, black, align=center] {$(b,\{1.05,0\},\{\bot, \top\})$} (1)
        
        (1) edge[in= 30, out=270, color=black] node[below, sloped, black, align=center] {$(b,\{2.05,1\},\{\top, \top\})$} (2)
        
        (2) edge[in=230, out=180,loop, color=black] node[left, black, align=center] {$(a,\{0,0\},\{\top, \top\})$} (2)
        
        (2) edge[in=290, out=240,loop, color=black] node[below, black, align=center] {$(a,\{1.05,1.05\},\{\top, \top\})$} (2)
        
        (2) edge[in=350, out=300,loop, color=black] node[right, black, align=center] {$(b,\{0,0\},\{\top, \top\})$} (2)
        
        (1) edge[in=30, out=150,color=black] node[above, black, align=center] {$(a,\{2.1,1.05\},\{\top, \top\})$} (0);

         \node [below=40pt, align=flush center,text width=3cm] at (2) {$\text{M}_9$};
\end{tikzpicture}
}
\end{center}
\end{minipage}
\begin{minipage}[b]{0.49\textwidth}
\begin{center}
\resizebox{0.9\textwidth}{!}{
\begin{tikzpicture}[scale=0.65, ->, >=stealth', shorten >=1pt, auto, node distance=2cm, semithick, every node/.style={scale=0.65}]
\centering
        \node[initial, state]  (0) {$l_{-+}$};
        \node[accepting, state]  (1) at (5,0) {$l_{+-}$};
        \node[state]  (2) at (2.5,-4) {$l_{--}$};
        \path  (0) edge[in=180, out=240, color=black] node[below,sloped, black, align=center] {$b,c_{1}\geq 0\wedge c_{2}\geq 0,$ \\$\{c_1, c_2\}$} (2)
        
        (0) edge[color=black] node[above,sloped, black, align=center] {$a, c_{1}> 1\wedge  c_{2}\leq 1, $ \\$\{c_1,c_2\}$} (2)

        (0) edge[in= 150, out=270, color=black] node[above,sloped, black, align=center] {$a,  c_{1}\leq 1\wedge c_{2}\geq 0,$ \\$\{c_1,c_2\}$} (2)
        
        (1) edge[color=black] node[above,sloped, black, align=center,pos=0.4] {$a,c_{1}\leq 2\wedge c_{2}>1,$ \\$\{c_1, c_2\}$} (2)
        
        (1) edge[in= 30, out=270, color=black] node[above, sloped,black, align=center] {$b,0\leq c_{1}\leq 1 \wedge c_{2}\geq 0,$ \\$\{c_1, c_2\}$} (2)
        
        (0) edge[color=black] node[above, black, align=center] {$a,c_{1}> 1\wedge c_{2} > 1,\{c_2\}$} (1)

        (1) edge[in=105, out=75,loop, color=black] node[above, black, align=center] {$b,1<c_{1}\leq 2 \wedge c_{2}\geq 0,$ \\$\{c_2\}$} (1)

        (1) edge[in=350, out=310, color=black] node[above,sloped, black, align=center] {$b,c_{1}>2 \wedge c_{2}\geq 1,$ \\$\{c_1, c_2\}$} (2)
        
        (1) edge[in=30, out=60,loop, color=black] node[right, black, align=center] {$b,c_{1}>2\wedge c_{2}< 1,$ \\$\{c_2\}$} (1)
        
        (2) edge[in=210, out=240,loop, color=black] node[left, black, align=center] {$a, c_{1}\geq 0\wedge c_{2}\geq 0, $ \\$\{c_1,c_2\}$} (2)
        
        (2) edge[in=310, out=280,loop, color=black] node[below, black, align=center] {$b, c_{1}\geq 0\wedge c_{2}\geq 0, $ \\$\{c_1,c_2\}$} (2)
        
        (1) edge[in=30, out=150,color=black] node[above, black, align=center] {$a,c_1>2\wedge c_2>1,\{c_1, c_2\}$} (0)
        
        (1) edge[in=340, out=330, color=black] node[below, sloped, black, align=center] {$a,c_{1}\geq 0\wedge c_{2}\leq 1,$ \\$\{c_1, c_2\}$} (2);

        \node [below=40pt, align=flush center,text width=3cm] at (2) {$\mathcal{H}_9$};
\end{tikzpicture}
}
\end{center}
\end{minipage}
\begin{minipage}[b]{0.49\textwidth}
\begin{center}
\resizebox{0.75\textwidth}{!}{
\begin{tikzpicture}[scale=0.65, ->, >=stealth', shorten >=1pt, auto, node distance=2cm, semithick, every node/.style={scale=0.65}]
\centering
        \node[initial, state]  (0) {$l_{-+}$};
        \node[accepting, state]  (1) at (5,0) {$l_{+-}$};
        \node[state]  (2) at (2.5,-4) {$l_{--}$};
        \path  (0) edge[in=150, out=270, color=black] node[below,sloped, black, align=center] {$(b,\{0,0\},\{\top, \top\})$} (2)
        
        (0) edge[color=black] node[below,sloped, black, align=center] {$(a,\{0,0\},\{\top, \top\})$} (2)
        
        (1) edge[color=black] node[below,sloped, black, align=center] {$(a,\{0,0\},\{\top, \top\})$} (2)
        
        (1) edge[in= 30, out=270, color=black] node[below, sloped,black, align=center] {$(b,\{0,0\},\{\top, \top\})$} (2)
        
        (0) edge[color=black] node[above, black, align=center] {$(a,\{1.05,1.05\},\{\bot, \top\})$} (1)
        
        (1) edge[in=75, out=105,loop, color=black] node[above, black, align=center] {$(b,\{1.05,0\},\{\bot, \top\})$} (1)
        
        (1) edge[in=0, out=300, color=black] node[below,sloped, black, align=center] {$(b,\{2.05,1\},\{\top, \top\})$} (2)
        
        (2) edge[in=230, out=180,loop, color=black] node[left, black, align=center] {$(a,\{0,0\},\{\top, \top\})$} (2)
        
        (2) edge[in=290, out=240,loop, color=black] node[below, black, align=center] {$(a,\{1.05,1.05\},\{\top, \top\})$} (2)
        
        (2) edge[in=350, out=300,loop, color=black] node[right, black, align=center] {$(b,\{0,0\},\{\top, \top\})$} (2)
        
        (1) edge[in=30, out=150,color=black] node[above, black, align=center] {$(a,\{2.1,1.05\},\{\top, \top\})$} (0)

        (1) edge[in=80, out=200, color=black] node[below,sloped, black, align=center] {$(a,\{3,1.95\},\{\top, \top\})$} (2);

         \node [below=40pt, align=flush center,text width=3cm] at (2) {$\text{M}_{10}$};
\end{tikzpicture}
}
\end{center}
\end{minipage}
\begin{minipage}[b]{0.49\textwidth}
\begin{center}
\resizebox{0.9\textwidth}{!}{
\begin{tikzpicture}[scale=0.65, ->, >=stealth', shorten >=1pt, auto, node distance=2cm, semithick, every node/.style={scale=0.65}]
\centering
        \node[initial, state]  (0) {$l_{-+}$};
        \node[accepting, state]  (1) at (5,0) {$l_{+-}$};
        \node[state]  (2) at (2.5,-4) {$l_{--}$};
        \path  (0) edge[in=180, out=240, color=black] node[below,sloped, black, align=center] {$b,c_{1}\geq 0\wedge c_{2}\geq 0,$ \\$\{c_1, c_2\}$} (2)
        
        (0) edge[color=black] node[above, sloped,black, align=center,pos=0.45] {$a, c_{1}> 1\wedge  c_{2}\leq 1, $ \\$\{c_1,c_2\}$} (2)

        (0) edge[in= 150, out=270, color=black] node[above,sloped, black, align=center] {$a,  c_{1}\leq 1\wedge c_{2}\geq 0,$ \\$\{c_1,c_2\}$} (2)
        
        (1) edge[color=black] node[below,sloped, black, align=center] {$a,c_{1}\leq 2\wedge c_{2}\geq 0,$ \\$\{c_1, c_2\}$} (2)
        
        (1) edge[in= 30, out=270, color=black] node[below, sloped,black, align=center] {$b, c_{1}\leq 1 \wedge c_{2}\geq 0,\{c_1, c_2\}$} (2)
        
        (0) edge[color=black] node[above, black, align=center] {$a,c_{1}> 1\wedge c_{2} > 1,\{c_2\}$} (1)

        (1) edge[in=105, out=75,loop, color=black] node[above, black, align=center] {$b,1<c_{1}\leq 2 \wedge c_{2}\geq 0,$ \\$\{c_2\}$} (1)

        (1) edge[in=0, out=300, color=black] node[below,sloped, black, align=center,] {$b,c_{1}>2 \wedge c_{2}\geq 1,\{c_1, c_2\}$} (2)
        
        (1) edge[in=30, out=60,loop, color=black] node[right, black, align=center] {$b,c_{1}>2\wedge c_{2}< 1,$ \\$\{c_2\}$} (1)
        
        (2) edge[in=210, out=240,loop, color=black] node[left, black, align=center] {$a, c_{1}\geq 0\wedge c_{2}\geq 0, $ \\$\{c_1,c_2\}$} (2)
        
        (2) edge[in=310, out=280,loop, color=black] node[below, black, align=center] {$b, c_{1}\geq 0\wedge c_{2}\geq 0, $ \\$\{c_1,c_2\}$} (2)
        
        (1) edge[in=30, out=150,color=black] node[above, black, align=center] {$a,2<c_1<3\wedge c_2>1,\{c_1, c_2\}$} (0)
        
        (1) edge[in=320, out=350, color=black] node[below, sloped,black, align=center] {$a,c_{1}>2\wedge c_{2}\leq 1,\{c_1, c_2\}$} (2)

        (1) edge[in=85, out=200, color=black] node[below, sloped,black, align=center] {$a, c_1\geq 3\wedge c_2>1,$ \\$\{c_1, c_2\}$} (2);

        \node [below=40pt, align=flush center,text width=3cm] at (2) {$\mathcal{H}_{10}$};
\end{tikzpicture}
}
\end{center}
\end{minipage}
%}
\end{center}
\caption{Iterations of intermediate DFA $\mathrm{M}$ and hypothesis $\mathcal{H}$ during the learning process w.r.t $\mathcal{A}$ in Figure~\ref{fig:timedautomata} (Continued).}
%\label{fig:case-automata}
\end{figure*}

\begin{figure*}
\begin{center}
\begin{subfigure}{0.3\textwidth}
\begin{center}
\resizebox{!}{0.8\textwidth}{
\begin{tikzpicture}[scale=0.65, ->, >=stealth', shorten >=1pt, auto, node distance=2cm, semithick, every node/.style={scale=0.65}]
% \centering
        \node[initial, state]  (0) {$l_0$};
        \node[accepting, state](1) at (0,-3) {$l_1$};
        \path  (0) edge node[right, align=center] {$a,$ \\ $c_{1}> 1,$\\ $\{c_2\}$} (1)
        
        (1) edge[in= 240, out=120, color=black] node[left, black, align=center] {$a,$ \\ $ 0 \leq c_{1} < 3\wedge c_{2} > 1,$ \\ $\{c_1,c_2\}$} (0)
        
        (1) edge[loop below] node[left, align=center] {$a,$ \\ $ c_{1}\geq 0\wedge 0\leq c_{2}< 1,$\\ $\{c_2\}$} (1);
\end{tikzpicture}}
\end{center}
\caption{The target DTA with case ID $2\_1\_2\_3$.}
\end{subfigure}
\begin{subfigure}{0.3\textwidth}
\begin{center}
\resizebox{!}{0.8\textwidth}{
\begin{tikzpicture}[scale=0.65, ->, >=stealth', shorten >=1pt, auto, node distance=2cm, semithick, every node/.style={scale=0.65}]
% \centering
        \node[initial, state]  (0) {$l_0$};
        \node[accepting, state](1) at (0,-3) {$l_1$};
        \path  (0) edge node[right, align=center] {$a,$ \\ $c_{1}> 10,$\\ $\{c_2\}$} (1)
        
        (1) edge[in= 240, out=120, color=black] node[left, black, align=center] {$a,$ \\ $ 0 \leq c_{1} < 30\wedge c_{2} > 10,$ \\ $\{c_1,c_2\}$} (0)
        
        (1) edge[loop below] node[left, align=center] {$a,$ \\ $ c_{1}\geq 0\wedge 0\leq c_{2}< 10,$\\ $\{c_2\}$} (1);
\end{tikzpicture}}
\end{center}
\caption{The target DTA with case ID $2\_1\_2\_30$.}
\end{subfigure}
\begin{subfigure}{0.3\textwidth}
\begin{center}
\resizebox{!}{0.6\textwidth}{
\begin{tikzpicture}[scale=0.65, ->, >=stealth', shorten >=1pt, auto, node distance=2cm, semithick, every node/.style={scale=0.65}]
% \centering
        \node[initial, state]  (0) {$l_0$};
        \node[accepting, state](1) at (0,-3) {$l_1$};
         \node[state](2) at (3,-1.5)  {$l_2$};
         
        \path  (0) edge node[left, align=center] {$a,$ \\$c_{1}> 1,$\\ $\{c_2\}$} (1)
        
        (2) edge[color=black] node[right, black, align=center, pos= .8] {$a, 0 \leq c_{1} < 3\wedge c_{2} > 1,$ \\ $\{c_1,c_2\}$} (0)
        
        (1) edge[color=black] node[right=10pt, black, align=center, pos=.1] {$a,c_{1}\geq 0 \wedge 0\leq c_{2}< 1, $\\$\{c_2\}$} (2);
\end{tikzpicture}
}
\end{center}
\caption{The target DTA with case ID $3\_1\_2\_3$.}
\end{subfigure}
\\[0.2cm]
\begin{subfigure}{0.45\textwidth}
\begin{center}
\resizebox{!}{0.45\textwidth}{
\begin{tikzpicture}[scale=0.65, ->, >=stealth', shorten >=1pt, auto, node distance=2cm, semithick, every node/.style={scale=0.65}]
% \centering
        \node[initial, state]  (0) {$l_0$};
        \node[accepting, state](1) at (0,-3) {$l_1$};
         \node[state](2) at (3,-1.5)  {$l_2$};
         
        \path  (0) edge node[left, align=center] {$a,$ \\$c_{1}> 10,$\\ $\{c_2\}$} (1)
        
        (2) edge[color=black] node[right, black, align=center, pos= .8] {$a, 0 \leq c_{1} < 30\wedge c_{2} > 10,$ \\ $\{c_1,c_2\}$} (0)
        
        (1) edge[color=black] node[right=10pt, black, align=center, pos=.1] {$a,c_{1}\geq 0 \wedge 0\leq c_{2}< 10, $\\$\{c_2\}$} (2);
\end{tikzpicture}
}
\end{center}
\caption{The target DTA with case ID $3\_1\_2\_30$.}
\end{subfigure}
\begin{subfigure}{0.45\textwidth}
\begin{center}
\resizebox{!}{0.67\textwidth}{
\begin{tikzpicture}[scale=0.65, ->, >=stealth', shorten >=1pt, auto, node distance=2cm, semithick, every node/.style={scale=0.65}]
        \node[initial, state]  (0) {$l_0$};
        \node[accepting, state](1) at (0,-3) {$l_1$};
         \node[state](2) at (3,-3)  {$l_2$};
         \node[state](3) at (3,0)  {$l_3$};
         
        \path  (0) edge node[left, align=center] {$a,$ \\$c_{1}> 10,$\\ $\{c_2\}$} (1)
        
        (2) edge[color=black] node[right, black, align=center] {$a, $ \\ $ c_{1} \geq 0\wedge c_{2} < 20,$ \\ $\{c_2\}$} (3)
        
        (1) edge[color=black] node[below, black, align=center] {$a,$ \\ $c_{1}\geq 0 \wedge 0\leq c_{2}< 10, $\\$\{c_2\}$} (2)

        (3) edge[color=black] node[above, black, align=center] {$a, $ \\ $ 0 \leq c_{1} < 50\wedge c_{2} > 10,$ \\ $\{c_1,c_2\}$} (0);
\end{tikzpicture}
}
\end{center}
\caption{The target DTA with case ID $4\_1\_2\_50$.}
\end{subfigure}
\end{center}
\caption{The target DTAs used in Table~\ref{tab:DTAresults} of Section~\ref{sec:experiment}}
\label{fig:target_dta_experiment}
\end{figure*}
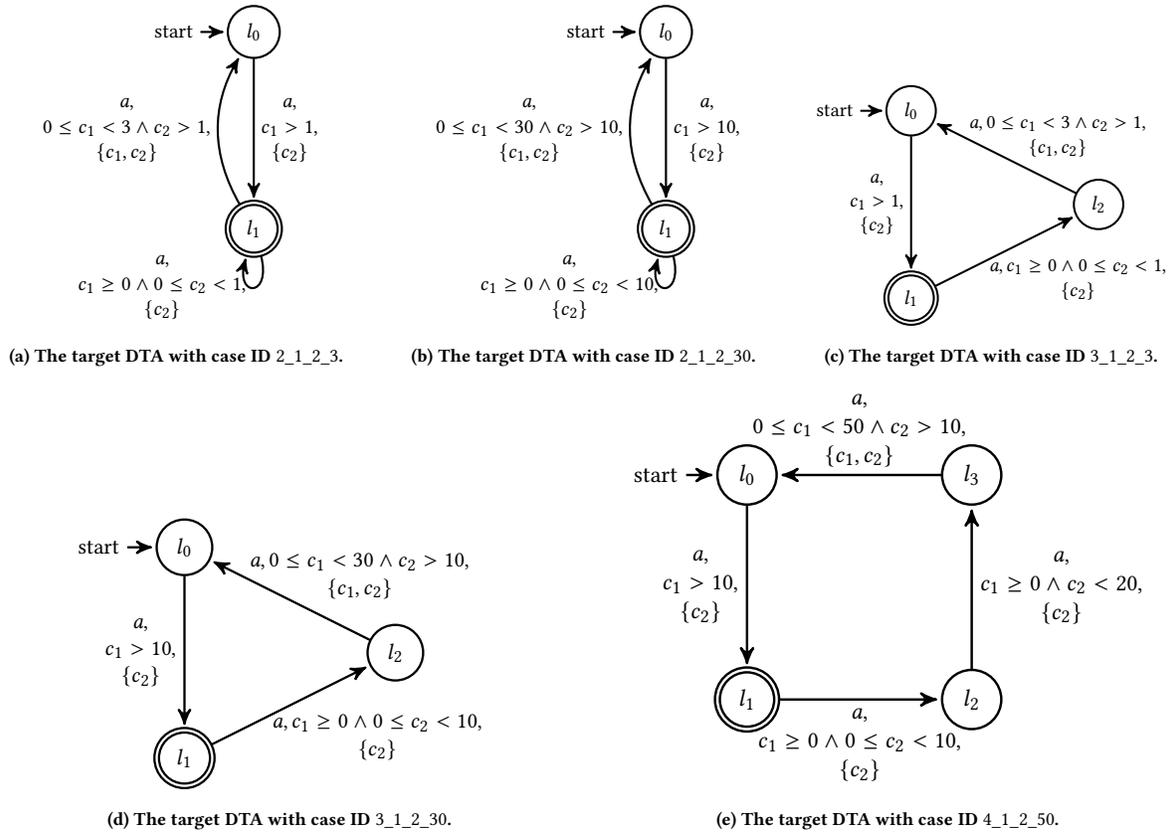

\end{document}